%% file: NN_LRA.tex
\documentclass[11pt]{article} 

\usepackage[margin=1in]{geometry}
\usepackage{float}
\usepackage{amssymb,amsfonts,amsmath}
\usepackage{enumerate}
\usepackage{enumitem}
\usepackage{caption}
\usepackage{amsthm}
\usepackage{color}
\usepackage[colorlinks=true, linkcolor=red, urlcolor=blue, citecolor=blue]{hyperref}
\usepackage{comment}

\newcommand{\C}{\mathbb{C}}	                
\newcommand{\Q}{\mathbb{Q}}                     
\newcommand{\R}{\mathbb{R}}                     
     
\newcommand{\Proj}{\texttt{Proj}}

\def\expec#1#2{{\bf \mathbb{E}}_{#1}[ #2 ]}
\def\expecf#1#2{{\bf \mathbb{E}}_{#1}\left[ #2 \right]}

\def\sgn#1{\mathrm{sgn} (#1)}

\def\AA{\mathbf{A}}
\def\BB{\mathbf{B}}
\def\MM{\mathbf{M}}
\def\NN{\mathbf{N}}

\def\D{\mathbf{D}}
\def\II{\mathbb{I}}
\def\Y{\mathbf{Y}}
\def\X{\mathbf{X}}
\def\P{\mathbf{P}}
\def\U{\mathbf{U}}
\def\V{\mathbf{V}}
\def\E{\mathbf{E}}
\def\S{\mathbf{S}}
\def\T{\mathbf{T}}
\def\W{\mathbf{W}}
\def\Q{\mathbf{Q}}
\def\G{\mathbf{G}}
\def\M{\mathbf{M}}
\def\ZZ{\mathbf{Z}}
\def\RR{\mathbf{R}}

\newcommand{\poly}{\text{poly}}

\newcommand{\pr}[1]{\text{\normalfont Pr}\normalfont\lbrack #1 \rbrack} 
\newcommand{\ex}[1]{\mathbb{E}\normalfont\lbrack #1 \rbrack}
\newcommand{\bpr}[1]{\text{\normalfont Pr}\normalfont \Big[#1 \Big]} 
\newcommand{\bex}[1]{\mathbb{E}\normalfont \Big[#1 \Big]}

\newcommand{\eps}{\epsilon}

\newcommand{\ttx}[1]{\texttt{#1}}

\newcommand{\sg}[1]{\ttx{sign}(#1)}

\newcommand{\req}{\overset{\text{row}}{\simeq}}

\newtheorem{theorem}{Theorem}
\newtheorem{lemma}{Lemma}
\newtheorem{corollary}{Corollary}
\newtheorem{proposition}{Proposition}
\newtheorem{definition}{Definition}

\newtheorem{remark}{Remark}

\usepackage[framemethod=TikZ]{mdframed}
\newcounter{Frame}
\newenvironment{Frame}[1][htb]{%
\refstepcounter{Frame}
    \begin{mdframed}[%
        frametitle={#1},
        skipabove=\baselineskip plus 2pt minus 1pt,
        skipbelow=\baselineskip plus 2pt minus 1pt,
        linewidth=1.0pt,
        frametitlerule=true,
    ]%
}{%
    \end{mdframed}
}

\title{Learning Two Layer Rectified Neural Networks in Polynomial Time}
\newcommand*\samethanks[1][\value{footnote}]{\footnotemark[#1]}
\author{
	Ainesh Bakshi\thanks{The authors thank the partial support by the National Science Foundation under Grant No. CCF-1815840. Part of this work was done while the authors were visiting the Simons Institute for the Theory of Computing.}\\
	Carnegie Mellon University \\
	\texttt{abakshi@cs.cmu.edu}
	\and
	Rajesh Jayaram\samethanks\\
	Carnegie Mellon University\\
	\texttt{rkjayara@cs.cmu.edu}
	\and
	David P. Woodruff\samethanks\\
	Carnegie Mellon University \\
	\texttt{dwoodruf@cs.cmu.edu}
}
\date{}
\begin{document}
\clearpage\maketitle
\thispagestyle{empty}
\setcounter{page}{0}
	
	\begin{abstract}
	    We consider the following fundamental problem in the study of neural networks: given input examples $x \in \mathbb{R}^d$ and their vector-valued labels, as defined by an underlying generative neural network, recover the weight matrices of this network. We consider two-layer networks, mapping $\mathbb{R}^d$ to $\mathbb{R}^m$, with a single hidden layer and $k$ non-linear activation units $f(\cdot)$, where $f(x) = \max \{x , 0\}$ is the ReLU activation function. Such a network is specified by two weight matrices, $\mathbf{U}^* \in \mathbb{R}^{m \times k}, \mathbf{V}^* \in \mathbb{R}^{k \times d}$, such that the label of an example $x \in \mathbb{R}^{d}$ is given by $\mathbf{U}^* f(\mathbf{V}^* x)$, where $f(\cdot)$ is applied coordinate-wise. Given $n$ samples $x^1,\dots,x^n \in \mathbb{R}^d$ as a matrix $\mathbf{X} \in \mathbb{R}^{d \times n}$ and the label $\mathbf{U}^* f(\mathbf{V}^* \mathbf{X})$ of the network on these samples, our goal is to recover the weight matrices $\mathbf{U}^*$ and $\mathbf{V}^*$. More generally, our labels $\mathbf{U}^* f(\mathbf{V}^* \mathbf{X})$  may be corrupted by noise, and instead we observe $\mathbf{U}^* f(\mathbf{V}^* \mathbf{X}) + \mathbf{E}$ where $\mathbf{E}$ is some noise matrix. Even in this case, we may still be interested in recovering good approximations to the weight matrices $\mathbf{U}^*$ and $\mathbf{V}^*$.

	    In this work, we develop algorithms and hardness results under varying assumptions on the input and noise. Although the problem is NP-hard even for $k=2$, by assuming Gaussian marginals over the input $\mathbf{X}$ we are able to develop polynomial time algorithms for the approximate recovery of $\mathbf{U}^*$ and $\mathbf{V}^*$. Perhaps surprisingly, in the noiseless case our algorithms recover $\mathbf{U}^*,\mathbf{V}^*$ \textit{exactly}, i.e., with no error. To the best of the our knowledge, this is the first algorithm to accomplish exact recovery for the ReLU activation function. For the noisy case, we give the first polynomial time algorithm that approximately recovers the weights in the presence of mean-zero noise $\mathbf{E}$. Our algorithms generalize to a larger class of \textit{rectified} activation functions, $f(x) = 0$ when $x\leq 0$, and $f(x) > 0$ otherwise. Although our polynomial time results require $\mathbf{U}^*$ to have full column rank, we also give a fixed-parameter tractable algorithm (in $k$) when $\mathbf{U}^*$ does not have this property. Lastly, we give a  fixed-parameter tractable algorithm for more arbitrary noise matrices $\mathbf{E}$, so long as they are independent of $\mathbf{X}$.
	\end{abstract}

	\newpage
	
	\section{Introduction}
	Neural networks have achieved remarkable success in solving many modern machine learning problems which were previously considered to be intractable. With the use of neural networks now being wide-spread in numerous communities,  the optimization of neural networks is an object of intensive study.
	
	Common usage of neural networks involves running stochastic gradient descent (SGD) with simple non-linear activation functions, such as the extremely popular ReLU function, to learn an incredibly large set of weights. This technique has enjoyed immense success in solving complicated classification tasks with record-breaking accuracy. However, theoretically the behavior and convergence properties of SGD are very poorly understood, and few techniques are known which achieve provable bounds for the training of large neural networks. This is partially due to the hardness of the problem -- there are numerous formulations where the problem is known to be NP-hard \cite{br92, judd1988neural, boob2018complexity, Manurangsi2018complexity}. Nevertheless, given the importance and success in solving this problem in practice, it is important to understand the source of this hardness. 
	
	Typically a neural network can be written in the following form: $\AA = \U^i(\cdots \U^3f(\U^2f(\U^1\X))$, where $i$ is the depth of the network, $\X \in \mathbb{R}^{d \times n}$ is a matrix with columns corresponding to individual $d$-dimensional input samples, and $\AA$ is the output labeling of $\X$. The functions $f$ are applied entry-wise to a matrix, and are typically non-linear. Perhaps the most popular activation used in practice is the ReLU, given by $f(x) = \max\{0,x\}$. Here each $\U^i$ is an unknown linear map, representing the ``weights", which maps inputs from one layer to the next layer. In the reconstruction problem, when it is known that $\AA$ and $\X$ are generated via the above model, the goal is to recover the matrices $\U^1, \ldots, \U^i$.

    In this work, we consider the problem of learning the weights of two layer networks with a single non-linear layer. Such a network can be specified by two weight matrices $\U^* \in \R^{m \times k}$ and $\V^* \in \R ^{k \times d}$, such that, on a $d$-dimensional input vector $x \in \R^d$, the classification of the network is given by $\U^*f(\V^* x) \in \R^m$.  
    Given a training set $\X \in \R^{d \times n} $ of $n$ examples, along with their labeling $\AA = \U^* f(\V^* \X) + \E$, where $\E$ is a (possibly zero) noise matrix, the learning problem is to find $\U$ and $\V$ for which 
    \[ \| \U - \U^*\|_F  + \|\V - \V^*\|_F \leq \eps\]
    
We consider two versions of this problem. First, in the noiseless (or realizable) case, we observe $\AA = \U^* f(\V^*\X)$ precisely. In this setting, we demonstrate that exact recovery of the matrices $\U^*,\V^*$ is possible in polynomial time. 
Our algorithms, rather than exploiting smoothness of activation functions, exploit combinatorial properties of rectified activation functions. 
Additionally, we consider the more general noisy case, where we instead observe $\AA = \U^* f(\V^* \X) + \E$, where $\E$ is a noise matrix which can satisfy various conditions. Perhaps the most common assumption in the literature \cite{ge2018learning, ge2017learning, janzamin2015beating} is that $\E$ has mean $0$ and is sub-Gaussian. Observe that the first condition is equivalent to the statement that $\ex{\AA \; | \; \X } = \U^*f(\V^*\X)$.
 While we primarily focus on designing polynomial time algorithms for this model of noise, in Section \ref{sec:1bit} we demonstrate fixed-parameter tractable (in the number $k$ of ReLUs) algorithms to learn the underlying neural network for a much wider class of noise matrices $\E$. 
We predominantly consider the \textit{identifiable} case where $\U^* \in \R^{m \times k}$ has full column rank, however we also provide supplementary algorithms for the exact case when $m < k$. Our algorithms are robust to the behavior of $f(x)$ for positive $x$, and therefore generalize beyond the ReLU to a wider class of rectified functions $f$ such that $f(x) = 0$ for $x \leq 0$ and $f(x) > 0$ otherwise.

It is known that stochastic gradient descent cannot converge to the ground truth parameters when $f$ is ReLU and $\V^*$ is orthonormal, even if we have access to an infinite number of samples \cite{livni2014computational}. This is consistent with empirical observations and theory, which states that over-parameterization is crucial to train neural networks successfully \cite{hardt2014understanding, soudry2016no}. In contrast, in this work we demonstrate that we can approximate the optimal parameters in the noisy case, and obtain the optimal parameters exactly in the realizable case, in polynomial time, without over-parameterization. In other words, we provide algorithms that do not succumb to spurious local minima, and can converge to the global optimum efficiently, without over-parametrization.
   
\subsection{Our Contributions}
We now state our results more formally. We consider $2$-layer neural networks with ReLU-activation functions 
$f$.  
Such a neural network is specified by matrices $\U^* \in \R^{m \times k}$ and $\V^* \in \R^{k \times d}$. 
We are given $d$-dimensional input examples $x^i \in \R^d$, which form the columns of our input matrix $\X$, and also give the network's $m$-dimensional
classification of $\X$, which is $\AA = \U^* f(\V^*\X)$, where $f$ is applied entry-wise. We note that our formulation corresponds to having one non-linear layer. 

\paragraph{Worst Case Upper Bounds.}
In the worst case setting, no properties are assumed on the inputs $\X,\AA$. While this problem is generally assumed to be intractable, we show, perhaps surprisingly, that when rank$(\AA) = k$ and $k=O(1)$, polynomial time exact algorithms do exist. 
One of our primary techniques throughout this work is the leveraging of combinatorial aspects of the ReLU function. For a row $f(\V^*\X)_{i,*}$, we define a \textit{sign pattern} of this row to simply be the subset of positive entries of the row. Thus, a sign pattern of a vector in $\R^n$ is simply given by the orthant of $\R^n$ in which it lies. 
 We first prove an upper bound of $O(n^k)$ on the number of orthants which intersect with an arbitrary $k$-dimensional subspace of $\R^n$. Next, we show how to enumerate these sign patterns in time $n^{k + O(1)}$. 
 
 We use this result to give an $n^{O(k)}$ time algorithm for the neural network learning problem in the \textit{realizable case}, where $\AA = \U^* f(\V^*\X)$ for some fixed \textit{rank-k} matrices $\U^*,\V^*$. After fixing a sign pattern of $f(\V^*\X)$, we can effectively ``remove" the non-linearity of $f$. Even so, the learning problem is still non-convex, and cannot be solved in polynomial time in the general case (even for fixed $k$). We show, however, that if the \emph{rank} of $\AA$ is $k$, then it is possible to use a sequence of linear programs to recover $\U^*,\V^*$ in polynomial time given the sign pattern, which allows for an $n^{O(k)}$ overall running time. Our theorem is stated below.

\vspace{0.2in} 
\noindent \textbf{Theorem \ref{thm:exact_lp}.} \textit{
Given $\AA \in \R^{m \times n}, \X \in \R^{d \times n}$, such that $\AA = \U^*f(\V^*\X)$ and $\AA$ is rank $k$, there is an algorithm that finds $\U^* \in \R^{m \times k} ,\V^* \in \R^{k \times d}$ such that $\AA = \U^* f(\V^*\X)$ and runs in time $\poly(n,m,d)\min\{n^{O(k)},\allowbreak 2^n\}$.}

\paragraph{Worst Case Lower Bounds.}
 Our upper bound relies crucially on the fact that $\AA$ is rank $k$, which is full rank when $k \leq d,m$. We demonstrate that an $O(n^k)$ time algorithm is no longer possible without this assumption by proving the NP-hardness of the realizable learning problem when rank$(\AA) < k$, which holds even for $k$ as small as $2$. 
 Our hardness result is as follows.
 
\vspace{0.2in} 
 \noindent \textbf{Theorem \ref{prop:truehard}.} \textit{
     For a fixed $\alpha \in \R^{m \times k},\X \in \R^{d \times n},\AA \in \R^{m \times n}$, the problem of deciding whether there exists a solution $\V \in \R^{k \times d}$ to $\alpha f(\V\X) = \AA$ is NP-hard even for $k=2$. Furthermore, for the case for $k=2$, the problem is still NP-hard when $\alpha \in \R^{m \times 2}$ is allowed to be a variable}.

\paragraph{Gaussian Inputs.}
Since non-convex optimization problems are known to be NP-hard in general, it is, perhaps, unsatisfying to settle for worst-case results. Typically, in the learning community, to make problems tractable it is assumed that the input data is drawn from some underlying distribution that may be unknown to the algorithm. So, in the spirit of learning problems, we make the common step of assuming that the samples in $\X$ have a standard Gaussian distribution. More generally, our algorithms work for arbitrary multi-variate Gaussian distributions over the columns of $\X$, as long as the covariance matrix is non-degenerate, i.e., full rank (see Remark \ref{remark:multivariate}). In this case, our running time and sample complexity will blow up by the condition number of the covariance matrix, which we can estimate first using standard techniques. For simplicity, we state our results here for $\mathbf{\Sigma} = \mathbb{I}$, though, for the above reasons, all of our results for Gaussian inputs $\X$ extend to all full rank $\mathbf{\Sigma}$

Furthermore, because many of our primary results utilize the combinatorial sparsity patterns of $f(\V\X)$, where $\X$ is a Gaussian matrix, we do not rely on the fact that $f(x)$ is linear for $x>0$. For this reason, our results generalize easily to other \textit{non-linear} rectified functions $f$. In other words, any function $f$ given by 
\[ f(x) = \begin{cases} 0 & \text{ if } x \leq 0 \\
\phi(x) & \text{ otherwise} \end{cases}\]
where $\phi(x): [0,\infty] \to [0,\infty]$ is a continuous, injective function. In particular, our bounds do not change for polynomial valued $\phi(x) = x^c$ for $c \in \mathbb{N}$. For more details of this generalization, see Appendix \ref{sec:generalf}. Note, however, that our worst-case, non-distributional algorithms (stated earlier), where $\X$ is a fixed matrix, do not generalize to non-linear $\phi(x)$.

We first consider the noiseless setting, also referred to as the exact or realizable setting. Here $\AA = \U^*f(\V^*\X)$ is given for rank $k$ matrices $\U^*$ and $\V^*$, where $\X$ has non-degenerate Gaussian marginals. The goal is then to recover the weights $(\U^*)^T, \V^*$ exactly up to a permutation of their rows (since one can always permute both sets of rows without effecting the output of the network). Note that for any positive diagonal matrix $\D$, $\U^* f(\D\V^*\X) = \U^*\D f(\V^*\X)$ when $f$ is the ReLU. Thus recovery of $(\U^*)^T,\V^*$ is always only possible up to a permutation and positive scaling. We now state our main theorem for the exact recovery of the weights in the realizable (noiseless) setting.
 
\vspace{0.2in} 
 \noindent \textbf{Theorem \ref{thm:exactfinal}. }
\textit{Suppose $\AA = \U^* f(\V^* \X)$ where $\U^* \in \R^{m \times k}, \V^* \in \R^{k \times d}$ are both rank-$k$, and such that the columns of $\X \in \R^{d \times n}$ are mean $0$ i.i.d. Gaussian. Then if $n = \Omega(\poly(d,m,\kappa(\U^*),\kappa(\V^*)) )$, then there is a $\poly(n)$-time algorithm which recovers $(\U^*)^T,\V^*$ exactly up to a permutation of the rows with high probability.} \\

To the best of our knowledge, this is the first algorithm which learns the weights matrices of a two-layer neural network with ReLU activation \textit{exactly} in the noiseless case and with Gaussian inputs $\X$. Our algorithm first obtains good approximations to the weights $\U^*,\V^*$, and concludes by solving a system of judiciously chosen linear equations, which we solve using Gaussian elimination. Therefore, we obtain exact solutions in polynomial time, without needing to deal with convergence guarantees of continuous optimization primitives. 
 Furthermore, to demonstrate the robustness of our techniques, we show that using results introduced in the concurrent and independent work of Ge et. al. \cite{ge2018learning}, we can extend Theorem \ref{thm:exactfinal} to hold for inputs sampled from symmetric distributions (we refer the reader to Corollary \ref{cor:symmetric_input}). We note that \cite{ge2018learning} recovers the weight matrices up to additive error $\eps$ and runs in $\poly\left(\frac{1}{\eps}\right)$-time, whereas our algorithm has no $\eps$ dependency.

The runtime of our algorithm depends on the condition number $\kappa(\V^*)$ of $\V^*$, which is a fairly ubiquitous requirement in the literature for learning neural networks, and optimization in general \cite{ge2018learning, janzamin2015beating, lee2015faster, cohen2017matrix, arora2017provable, zhong2017recovery, sedghi2016provable}. To address this dependency, in Lemma \ref{lem:kappadependency} we give a lower bound which shows at least a linear dependence on $\kappa(\V^*)$ is necessary in the sample and time complexity.

Next, we introduce an algorithm for approximate recovery of the weight matrices $\U^*,\V^*$ when $\AA = \U^* f(\V^*\X) + \E$ for Gaussian marginals $\X$ and an i.i.d. sub-Gaussian mean-zero noise matrix $\E$ with variance $\sigma^2$. 
 
\vspace{0.2in} 
\noindent \textbf{Theorem \ref{thm:noisyfinal}.}
\textit{Let $\AA = \U^*f(\V^*\X) + \E$ be given, where $\U^* \in \R^{m \times k},\V^*\in \R^{k \times d}$ are rank-$k$, $\E$ is a matrix of i.i.d. mean-zero sub-Gaussian random variables with variance $\sigma^2$, and such that the columns of $\X \in \R^{d \times n}$ are i.i.d. Gaussian. Then given $n = \Omega\Big(  \poly\big( d,m,\kappa(\U^*),\kappa(\V^*) , \sigma, \frac{1}{\eps}\big) \Big)$, there is an algorithm that runs in $\poly(n)$ time and w.h.p. outputs $\V,\U$ such that }
\[ \|\U - \U^*\|_F \leq \eps \; \; \; \; \; \|\V - \V^* \|_F \leq \eps\]

Again, to the best of our knowledge, this work is the first which learns the weights of a $2$-layer network in this noisy setting without additional constraints, such as the restriction that $\U$ be positive. Recent independent and concurrent work, using different techniques, achieves similar approximate recovery results in the noisy setting \cite{ge2018learning}. We note that the algorithm of Goel et. al. \cite{goel2017learning} that \cite{ge2018learning} uses, crucially requires the linearity of the ReLU for $x>0$, and thus the work of \cite{ge2018learning} does not generalize to the larger class of rectified functions which we handle. We also note that the algorithm of \cite{ge2017learning} requires $\U^*$ to be non-negative. Finally, the algorithms presented in \cite{janzamin2015beating} work for activation functions that are thrice differentiable and can only recover rows of $\V^*$ up to $\pm 1$ scaling. Note, for the ReLU activation function, we need to resolve the signs of each row.

\paragraph{Fixed-Parameter Tractable Algorithms.}
For several harder cases of the above problems, we are able to provide  Fixed-Parameter Tractable algorithms. First, in the setting where the ``labels'' are vector valued, i.e., $m > 1$, we note prior results, not restricted to ReLU activation, require the rank of $\U^*$ to be $k$ \cite{ge2018learning, janzamin2015beating, ge2017learning}. This implies that $m \geq k$, namely, that the output dimension of the neural net is at least as large as the number $k$ of hidden neurons. Perhaps surprisingly, however, we show that even when $\U^*$ does not have full column rank, we can still recover $\U^*$ exactly in the realizable case, as long as no two columns are non-negative scalar multiples of each other. Note that this allows for columns of the form $[u,-u]$ for $u \in \R^m$ as long as $u$ is non-zero. Our algorithm for doing so is fixed paramater tractable in the condition number of $\V^*$ and the number of hidden neurons $k$. Our results rely on proving bounds on the sample complexity in order to obtain all $2^k$ possible sparsity patterns of the $k$-dimensional columns of $f(\V^*\X)$. 
 
\vspace{0.2in} 
\noindent \textbf{Theorem \ref{thm:fptfinal}.} \textit{
Suppose $\AA = \U^*f(\V^*\X)$ for $\U^* \in \R^{m \times k}$ for any $m \geq 1$ such that no two columns of $\U^*$ are non-negative scalar multiples of each other, and $\V^* \in \R^{k \times n}$ has rank$(\V^* ) = k$, and $n > \kappa^{O(k)} \poly(dkm)$. Then there is an algorithm which recovers $\U^*,\V^*$ exactly with high probability in time $\kappa^{O(k)} \poly(d,k,m)$.}\\

 Furthermore, we generalize our results in the noisy setting to \textit{arbitrary} error matrices $\|\E\|$, so long as they are independent of the Gaussians $\X$. In this setting, we consider a slightly different objective function, which is to find $\U,\V$ such that $\U f(\V\X)$ approximates $\AA$ well, where the measure is to compete against the optimal generative solution $\|\U^*f(\V^* \X) - \AA\|_F = \|\E\|_F$. Our results are stated below. 
 
\vspace{0.2in} 
\noindent \textbf{Theorem \ref{thm:vershyninfinal}.} \textit{
Let $\AA = \U^*f(\V^*\X) + \E$ be given, where $\U^* \in \R^{m \times k},\V^*\in \R^{k \times d}$ are rank-$k$, and $\E \in \R^{m \times n}$ is any matrix independent of $\X$. Then there is an algorithm which outputs $\U \in \R^{m \times k} ,\V \in \R^{k \times d}$ in time $(\kappa/\eps)^{O(k^2)}\poly(n,d,m)$ such that with probability $1 - \exp(-\sqrt{n})$ we have 
\[  \|\AA - \U f(\V\X)\|_F \leq \|\E \|_F + O\Big( \Big[\sigma_{\max} \eps \sqrt{n m}  \|\E\|_2\Big]^{1/2} \Big),\]
where $\|\E\|_2$ is the spectral norm of $\E$.}

Note that the above error bounds depend on the flatness of the spectrum of $\E$.  In particular, our bounds give a $(1+\eps)$ approximation whenever the spectral norm of $\E$ is a $\sqrt{m}$ factor smaller than the Frobenius norm, as is in the case for a wide class of random matrices \cite{vershynin2010introduction}. When this is not the case, we can scale $\eps$ by $1/\sqrt{m}$, to get an $(m \kappa/\eps)^{O(k^2)}$-time algorithm which gives a $(1+\eps)$ approximation for any error matrix $\E$ independent of $\X$ such that $\|\E\|_F = \Omega(\eps \|\U^* f(\V^* \X)\|_F)$.

\paragraph{Sparse Noise.}
Finally, we show that for \textit{sparse noise}, when the network is \textit{low-rank} we can reduce the problem to the problem of exact recovery in the noiseless case. Here, by low-rank we mean that $m > k$. It has frequently been observed in practice that many pre-trained neural-networks exhibit correlation and a low-rank structure \cite{denil2013predicting, denton2014exploiting}. Thus, in practice it is likely that $k$ need not be as large as $m$ to well-approximate the data. For such networks, 
we give a polynomial time algorithm for Gaussian $\X$ for exact recovery of $\U^*,\V^*$. Our algorithm assumes that $\U^*$ has orthonormal columns, and satisfies an \textit{incoherence} property, which is fairly standard in the numerical linear algebra community  \cite{candes2007sparsity, candes2009exact,keshavan2010matrix,candes2011robust, jain2013low, hardt2014understanding}. 
Formally, assume $\AA = \U^*f(\V^*\X) + \E$ where $\X$ is i.i.d. Gaussian, and $\E$ is obtained from the following sparsity procedure. First, fix any matrix $\overline{\E}$, and randomly choose a subset of $nm - s$ entries for some $s< nm$, and set them equal to $0$. The following result states that we can exactly recover $\U^*,\V^*$ in polynomial time even when $s = \Omega(mn)$. 

\vspace{0.2in} 
\noindent \textbf{Theorem \ref{thm:sparse} \& Corollary \ref{cor:sparse}.}\textit{ Let $\U^* \in \R^{m \times k}, \V^* \in \R^{k \times d}$ be rank $k$ matrices, where $\U^*$ has orthonormal columns, $\max_i \| (\U^*)^T e_i\|_2^2 \leq \frac{\mu k}{m}$ for some $\mu$, and $k \leq \frac{m}{\overline{\mu} \log^2(n)}$, where $\overline{\mu} =  O\big( (\kappa(\V^*))^{2} \sqrt{k \log(n) \mu} + \mu +  (\kappa(\V^*))^4 \log(n) \big)$. Here $\kappa(\V^*)$ is the condition number of $\V^*$. Let $\E$ be generated from the $s$-sparsity procedure with $s =  \gamma nm$ for some constant $\gamma > 0$ and let $\AA = \U^* f(\V\X) + \E$. Suppose the sample complexity satisfies $n = \poly(d,m,k,\kappa(\V^*))$ 
Then on i.i.d. Gaussian input $\X$ there is a $\poly(n)$ time algorithm that recovers $\U^*,\V^*$ exactly up to a permutation and positive scaling with high probability. }

\subsection{Related Work}
Recently, there has been a flurry of work developing provable algorithms for learning the weights of a neural network under varying assumptions on the activation functions, input distributions, and noise models \cite{sedghi2016provable, arora2016understanding, goel2016reliably, Manurangsi2018complexity, zhong2017recovery, ge2018learning, ge2017learning, zhong2017recovery, tian2017analytical, li2017convergence, brutzkus2017globally, soltanolkotabi2017learning, goel2018learning, du2018improved}. In addition, there have been a number of works which consider lower bounds for these problems under a similar number of varying assumptions
\cite{goel2016reliably, livni2014computational, zhang2016l1, sedghi2016provable, arora2016understanding, boob2018complexity, Manurangsi2018complexity}.  We describe the main approaches here, and how they relate to our problem.

\paragraph{Learning ReLU Networks without noise.}
In the noiseless setting with Gaussian input, the results of Zhong et al. \cite{zhong2017recovery} utilize a similar strategy as ours. Namely, they first apply techniques from tensor decomposition to find a good initialization of the weights, whereafter they can be learned to a higher degree of accuracy using other methods. At this point our techniques diverge, as they utilize gradient descent on the initialized weights, and demonstrate good convergence properties for \textit{smooth} activation functions. However, their results do not give convergence guarantees for non-smooth activation functions, including the ReLU and the more general class of rectified functions considered in this work. In this work, once we are given a good initialization, we utilize combinatorial aspects of the sparsity patterns of ReLU's, as well as solving carefully chosen linear systems, to obtain exact solutions.  

Li and Yuan \cite{NIPS2017_6662} also analyize stochastic gradient descent, and demonstrate good convergence properties when the weight matrix $\V^*$ is known to be close to the identity, and $\U^* \in \R^{1 \times k}$ is the all $1$'s vector. In \cite{tian2017symmetry}, stochastic gradient descent convergence is also analyzed when $\U^* \in \R^{1 \times k}$ is the all $1$'s vector, and when $\V^*$ is orthonormal. Moreover, \cite{tian2017symmetry} does not give bounds on sample complexity, and requires that a good initialization point is already given. 

For uniformly random and sparse weights in $[-1,1]$, Arora et al. \cite{arora2014provable} provide polynomial time learning algorithms. In \cite{brutzkus2017globally}, the learning of \textit{convolutions neural networks} is considered, where they demonstrate global convergence of gradient descent, but do not provide sample complexity bounds.

\paragraph{Learning ReLU Networks with noise.}
Ge et al. \cite{ge2017learning} considers learning a ReLU network with a single output dimension $\AA = u^T f(\V \X) + \E$ where $u \in \R^k$ is restricted to be entry-wise positive and $\E$ is a zero-mean sub-Gaussian noise vector. In this setting, it is shown that the weights $u,\V$ can be approximately learned in polynomial time when the input $\X$ is i.i.d. Gaussian. However, in contrast to the algorithms in this work, the algorithm of \cite{ge2017learning} relies heavily on the non-negativity of $u$ \cite{rongpersonal}, and thus cannot generalize to arbitrary $u$. 
Janzamin, Sedghi, and Anandkumar \cite{janzamin2015beating} utilize tensor decompositions to approximately learn the weights in the presence of mean zero sub-Gaussian noise, when the activation functions are smooth and satisfy the property that $f(x) = 1- f(-x)$.  Using similar techniques, Sedghi and Anandkumar \cite{sedghi2016provable} provide a polynomial time algorithm to approximate the weights, if the weights are sparse.  

A more recent result of Ge et al. demonstrates polynomial time algorithms for learning weights of two-layer ReLU networks in the presence of mean zero sub-gaussian noise, when the input is drawn from a mixture of a symmetric and Gaussian distribution \cite{ge2018learning}. We remark that the results of \cite{ge2018learning} were independently and concurrently developed, and utilize substantially different techniques than ours that rely crucially on the linearity of the ReLU for $x>0$ \cite{rongpersonal}. For these reasons, their algorithms do not generalize to the larger class of rectified functions which are handled in this work. To the best of the our knowledge, for the case of Gaussian inputs, this work and \cite{ge2018learning} are the first to obtain polynomial time learning algorithms for this noisy setting.

\paragraph{Agnostic Learning.}
A variety of works study learning ReLU's in the more general \textit{agnostic} learning setting, based off Valiant's original PAC learning model \cite{valiant1984theory}. The agnostic PAC model allows for arbitrary noisy and distributions over observations, and the goal is to output a hypothesis function which approximates the output of the neural network. Note that this does not necessarily entail learning the weights of an underlying network.
For instance, Arora et al. \cite{arora2016understanding} gives an  algorithm with $O(n^{d})$ running time to minimize the empirical risk of a two-layer neural network. A closer analysis of the generalization bounds required in this algorithm for PAC learning is given in \cite{Manurangsi2018complexity}, which gives  a $2^{\poly(k/\eps)} \poly(n,m,d,k)$ time algorithm under the constraints that $\U^* \in \{1,-1\}^k$ is given a fixed input, and both the input examples $\X$ and the weights $\V^*$ are restricted to being in the unit ball. In contrast, our $(\kappa/\eps)^{O(k^2)}$ time algorithm for general error matrices $\E$ improves on their complexity whenever $\kappa = O(2^{\poly(k)})$, and moreover can handle arbitrarily large $\V^*$ and unknown $\U^* \in \R^{m \times k}$. We remark, however, that our loss function is different from that of the PAC model, and is in fact roughly equivalent to the empirical loss considered in  \cite{arora2016understanding}.

Note that the above algorithms \textit{properly} learn the networks. That is, they actually output weight matrices $\U,\V$ such that $\U f(\V\X)$ approximates the data well under some measure. A relaxation if this setting is \textit{improper} learning, where the output of the learning algorithm can be any efficiently computable function, and not necessarily the weights  of neural network. Several works have been studied that achieve polynomial running times under varying assumptions about the network parameters, such as \cite{ goel2016reliably, goel2017learning}. The algorithm of \cite{goel2017learning}, returns a ``clipped'' polynomial.  In addition, \cite{zhang2016l1} gives polynomial time improper learning algorithms for multi-layer neural networks under several assumptions on the weights and activation functions.

\paragraph{Hardness.}
Hardness results for learning networks have an extensive history in the literature \cite{judd1988neural,br92}. Originally, hardness was considered for threshold activation functions $f(x) \in \{1,-1\}$, where it is known that even for two ReLU's the problem is NP-hard \cite{br92}. Very recently, there have been several concurrent and independent lower bounds developed for learning ReLU networks. The work of \cite{boob2018complexity} has demonstrated the hardness of a neural network with the same number of nodes as the hard network in this paper, albeit with two applications of ReLU's (i.e., two non-linear layers) instead of one. Note that the hardness results of this work hold for even a single non-linear layer. Also concurrently and independently, a recent result of  \cite{Manurangsi2018complexity} appears to demonstrate the same NP-hardness as that in this paper, albiet using a slightly different reduction. The results of \cite{Manurangsi2018complexity} also demonstrate that \textit{approximately} learning even a single ReLU is NP-hard. In addition, there are also NP-hardness results with respects to improper learning of ReLU networks \cite{goel2016reliably, livni2014computational, zhang2016l1} under certain complexity theoretic assumptions.

\paragraph{Sparsity.} One of the main techniques of our work involves analyzing the sparsity patterns of the vectors in the rowspan of $\AA$. Somewhat related reasoning has been applied by Spielman, Wang, and Wright to the dictionary learning problem  \cite{spielman2012exact}. Here, given a matrix $\AA$, the problem is to recover matrices $\BB,\X$ such that $\AA = \BB \X$, where $\X$ is sparse. They argue the uniqueness of such a factorization by proving that, under certain conditions, the sparsest vectors in the row span of $\AA$ are the precisely rows of $\X$. This informs their later algorithm for the exact recovery of these sparse vectors using linear programming. 
\subsection{Our Techniques}

One of the primary technical contributions of this work is the utilization of the combinatorial structure of sparsity patterns of the rows of $f(\V\X)$, where $f$ is a rectified function, to solve learning problems. Here, a sparsity pattern refers to the subset of coordinates of $f(\V\X)$ which are non-zero, and a rectified function $f$ is one which satisfies $f(x) = 0$ for $x \leq 0$, and $f(x)>0$ otherwise.

\paragraph{Arbitrary Input.}
For instance, given $\AA = \U^* f(\V^*\X)$ where $\U^*,\V^*$ are full rank and $f$ is the ReLU, one approach to recovering the weights is to find $k$-linearly vectors $v_i$ such that $f(v_i \X)$ span precisely the rows of $\AA$. Without the function $f(\cdot)$, one could accomplish this by solving a linear system. Of course, the non-linearity of the activation function complicates matters significantly. Observe, however, that if the sparsity pattern of $f(\V^*\X)$ was known before hand, one could simple *remove* $f$ on the coordinates where $f(\V^*\X)$ is non-zero, and solve the linear system here. On all other coordinates, one knows that $f(\V^*\X)$ is $0$, and thus finding a linearly independent vector in the right row span can be solved with a linear system. Of course, naively one would need to iterate over $2^n$ possible sparsity patterns before finding the correct one. However, one can show that any $k$-dimensional subspace of $\R^n$ can intersect at most $n^k$ orthants of $\R^n$, and moreover these orthants can be enumerated in $n^k \poly(n)$ time given the subspace. Thus the rowspan of $\AA$, being $k$-dimensional, can contain vectors with at most $n^k$ patterns. This is the primary obervation behind our $n^k \poly(n)$-time algorithm for exact recovery of $\U^*,\V^*$ in the noiseless case (for arbitrary $\X$). 

As mentioned before, the prior result requires $\AA$ to be rank-$k$, otherwise the row span of $f(\V\X)$ cannot be recovered from the row span of $\AA$. We show that this difficulty is not merely a product of our specific algorithm, by demonstrating that even for $k$ as small as $2$, if $\U^*$ is given as input then it is NP-hard to find $\V^*$ such that $\U^* f(\V^*\X) =\AA$, thus ruling out any general $n^k$ time algorithm for the problem. For the case of $k=2$, the problem is still NP-hard even when $\U^*$ is not given as input, and is a variable.

\paragraph{Gaussian Input.}
In response to the aformentioned hardness results, we relax to the case where the input $\X$ has Gaussian marginals. In the noiseless case, we \textit{exactly} learn the weights $\U^*,\V^*$ given $\AA = \U^*f(\V^*\X)$ (up to a positive scaling and permutation). As mentioned, our results utilize analysis of the sparsity patterns in the row-span of $\AA$. One benefit of these techniques is that they are largely insensitive to the behavior of $f(x)$ for positive $x$, and instead rely on the rectified property $f(\cdot)$. Hence, this can include even exponential functions, and not solely the ReLU.

Our exact recovery algorithms proceed in two steps. First, we obtain an approximate version of the matrix $f(\V^*\X)$. For a good enough approximation, we can exactly recover the sparsity pattern of $f(\V^*\X)$. Our main insight is, roughly, that the only sparse vectors in the row span of $\AA$ are precisely the rows of $f(\V^*\X)$. Specifically, we show that the only vectors in the row span which have the same sparsity pattern as a row of $f(\V^*\X)$ are scalar multiples of that row. Moreover, we show that no vector in the row span of $\AA$ is supported on a strict subset of the support of a given row of $f(\V^*\X)$. Using these facts, we can then set up a judiciously designed linear system to find these vectors, which allows us to recover $f(\V^*\X)$ and then $\V^*$ exactly. By solving linear systems, we avoid using iterative continuous optimization methods, which recover a solution up to additive error $\eps$ and would only provide rates of convergence in terms of $\eps$. In contrast, Gaussian elimination yields exact solutions in a polynomial number of arithmetic operations. 


The first step, finding a good approximation of $f(\V^* \X)$, can be approached from multiple angles. In this work, we demonstrate two different techniques to obtain these approximations, the first being Independent Component Analysis (ICA), and the second being tensor decomposition. To illustrate the robustness of our exact recovery procedure once a good estimate of $f(\V^*\X)$ is known, we show in Section \ref{sec:symmetric} how we can bootstrap the estimators of recent, concurrent and independent work \cite{ge2018learning}, to improve them from approximate recovery to exact recovery.

\paragraph{Independent Component Analysis.}
In the restricted case when $\V^*$ is orthonormal, we show that our problem can be modeled as a special case of \textit{Independent Component Analysis} (ICA). 
The ICA problem approximately recovers a subspace $\BB$, given that the algorithm observes samples of the form $y =\BB x + \zeta$, where $x$ is i.i.d. and drawn from a distribution that has moments bounded away from Gaussians, and $\zeta$ is a Gaussian noise vector. Intuitively, the goal of ICA is to find a linear transformation of the data such that each of the coordinates or features are as independent as possible. 
By rotational invariance of Gaussians, in this case $\V^*\X$ is also i.i.d. Gaussian, and we know that the columns of $f(\V^*\X)$ have independent components and moments bounded away from a Gaussian. Thus, in the orthonormal case, our problem is well suited for the ICA framework.  


\paragraph{Tensor Decomposition.}

A second, more general approach to approximating $f(\V^*\X)$ is to utilize techniques from \textit{tensor decomposition}. Our starting point is the generative model considered by Janzamin et. al.  \cite{janzamin2015beating}, which matches our setting, i.e., $\AA= \U^* f(\V^* \X)$. The main idea behind this algorithm is to construct a tensor that is a function of both $\AA, \X$ and captures non-linear correlations between them. A key step is to show that the resulting tensor has low CP-rank and the low-rank components actually capture the rows of the weight matrix $\V^*$. Intuitively, working with higher order tensors is necessary since matrix decompositions are only identifiable up to orthogonal components, whereas tensors have identifiable non-orthogonal components, and we are specifically interested in recovering approximations for non-orthonormal $\V^*$.

Next, we run a tensor decomposition algorithm to recover the low-rank components of the resulting tensor. While computing a tensor decomposition is NP-hard in general \cite{hillar2013most}, there is a plethora of work on special cases, where computing such decompositions is tractable \cite{bhaskara2014smoothed,song2016sublinear,wang2016online,goyal2014fourier, ge2015decomposing, barak2016noisy}. Tensor decomposition algorithms have recently become an invaluable algorithmic primitive and with applications in statistical and machine learning \cite{janzamin2015beating,janzamin2014score,ge2017learning,anandkumar2014tensor, barak2015dictionary}.

However, there are several technical hurdles involved in utilizing tensor decompositions to obtain estimates of $\V^*$. The first is that standard analysis of these methods utilizes a generalized version of \textit{Stein's Lemma} to compute the expected value of the tensor, which relies on the smoothness of the activation function. Thus, we first approximate $f(\cdot)$ closely using a Chebyshev polynomial $p(\cdot)$ on a sufficiently large domain. However, we cannot algorithmically manipulate the input to demand that $\AA$ instead be generated as $\U^* p(\V^*\X)$. Instead, we add a small mean-zero Gaussian perturbation to our samples and analyze the variation distance between $\AA = \U^*f(\V^*\X) + \G$ and $\U^*p(\V^*\X) + \G$. For a good enough approximation $p$, this variation distance will be too small for any algorithm to distinguish between them, thus standard arguments imply the success of tensor decomposition algorithms when given the inputs $\AA + \G$ and $\X$. 

Next, a key step is to construct a non-linear transformation of the input by utilizing knowledge about the underlying density function for the distribution of $\X$, which we denote by $p(x)$. The non-linear function considered is the so-called Score Function, defined in \cite{janzamin2014score}, which is the normalized $m$-th order derivative of the input probability distribution function $p(x)$.
Computing the score function for an arbitrary distribution can be computationally challenging. However, as mentioned in \cite{janzamin2014score}, 
we can use Hermite polynomials that help us compute a closed form for the score function, in the special case when $x \sim \mathcal{N}(0,\II)$. 


\paragraph{Sign Ambiguity.}
A further complication arises due to the fact that this form of tensor decomposition is agnostic to the signs of $\V$. Namely, we are guaranteed vectors $v_i$ from tensor decomposition such that $\|v_i -  \xi_i \V_{i,*}^*\|_F < \eps$, where $\xi_i \in \{1,-1\}$ is some unknown sign. Prior works have dealt with this issue by considering restricted classes of smooth activation functions which satisfy $f(x) = 1-f(-x)$ \cite{janzamin2015beating}. For such functions, one can compensate for not knowing the signs by allowing for an additional affine transformation in the neural network. Since we consider non-affine networks and rectified functions $f(\cdot)$ which do not satisfy this restriction, we must develop new methods to recover the signs $\xi_i$ to avoid the exponential blow-up needed to simply guess them. 

For the noiseless case, if $v_i$ is close enough to $\xi_i \V^*_{i,*}$, we can employ our previous results on the uniqueness of sparsity patterns in the row-span of $\AA$. Namely, we can show that the sparsity pattern of $f(\xi v_i)$ will in fact be feasible in the row-span of $\AA$, whereas the sparsity pattern of $f(-\xi v_i)$ will not, from which we recover the signs $\xi_i$ via a linear system. 

In the presence of noise, however, the problem becomes substantially more complicated. Because we do not have the true row-span of $f(\V^*\X)$, but instead a noisy row-span given by $\U^*f(\V^*\X) + \E$, we cannot recover the $\xi_i$'s by feasibility arguments involving sparsity patterns. 
Our solution to the sign ambiguity in the noisy case is a projection-based scheme. Our scheme for determining $\xi_i$ involves constructing a $2k-2$ dimensional subspace $S$, spanned by vectors of the form $f( \pm v_j \X)$ for all $j \neq i$. We augment this subspace as $S^{1} = S \cup \{f(v_i \X)\}$ and $S^{-1} = S \cup \{f(-v_i \X)\}$. We then claim that the length of the projections of the rows of $\AA$ onto the $S^{\xi}$ will be \textit{smaller} for $\xi = \xi_i$ than for $\xi = -\xi_i$. Thus by averaging the projections of the rows of $\AA$ onto these subspaces and finding the subspace which has the smaller projection length on average, we can recover the $\xi_i$'s with high probability. Our analysis involves bounds on projections onto perturbed subspaces, and a spectral analysis of the matrices $f(\W\X)$, where $\W$ is composed of up to $2k$ rows of the form $\V^*_{i,*}$ and $-\V^*_{i,*}$.

\paragraph{FPT Algorithms.}
In addition to our polynomial time algorithms, we also demonstrate how various seemingly intractable relaxations to our model, within the Gaussian input setting, can be solved in fixed-parameter tractable time in the number $k$ of hidden units, and the condition numbers $\kappa$ of $\U^*$ and $\V^*$. Our first result demonstrates that, in the noiseless case, exact recovery of $\U^*,\V^*$ is still possible even when $\U^*$ is not rank $k$. Note that the assumption that $\U^*$ is rank $k$ is required in many other works on learning neural networks \cite{ge2017learning, ge2018learning,janzamin2015beating, sedghi2016provable} 

We demonstrate that taking $\poly(d)\kappa^{O(k)}$ columns of $\X$, where $\kappa$ is the condition number of $\V^*$, is sufficient to obtain $1$-sparse vectors in the columns of $f(\V^*\X)$. As a result, we can look for column of $\AA$ which are positive scalar multiples of each other, and conclude that any such pair will indeed be a positive scaling of a column of $\U^*$ with probability $1$. This allows for exact recovery of $\U^*$ for any $\U^* \in \R^{m \times k}$ and $m \geq 1$, as long as no two columns of $\U^*$ are positive scalar multiples of each other. Thereafter, we can recover $\V^*$ by solving a linear system on the subset of $1$-sparse columns of $f(\V\X)$, and argue that the resulting constraint matrix is full rank. The result is a $\poly(d,k,m)\kappa^O(k)$ time algorithm for exact recovery of $\U^*,\V^*$.

Our second FPT result involves a substantial generalization of the class of error matrices $\E$ which we can handle. In fact, we allow arbitrary $\E$, so long as they are independent of the input $\X$. Our primary technical observation is as follows. Suppose that we were given $f(v \X) + \E$, where $\E$ is an arbitrary, possibly very large, error vector, and $v \in \R^d$.  Then one can look at the sign of each entry $i$, and consider it to be a noisy observation of which side of a halfspace the vector $\X_{*,i}$ lies within. In other words, we couch the problem as a noisy half-space learning problem, where the half-space is given by the hyperplane normal to $v$, and the labeling of $\X_{*,i}$ is the sign of $(f(v \X) + \E)_{i}$.

 Now while the error on each entry will be large, resulting in nearly half of the labelings being flipped incorrectly, because $\E$ is \textit{independent} of $\X$, we are able to adapt recent techniques in noisy-halfspace learning to recover $v$ in polynomial time. In order to utilize these techniques without knowing anything about $\E$, we must first {\it smooth out} the error $\E$ by adding a large Gaussian matrix. The comparatively small value of $f(v\X)$ is then able to shift the observed distribution of signs sufficiently to have non-trivial correlation with the true signs. Taking polynomially many samples, our algorithms detect this correlation, which will allow for accurate recovery of $v$.
 
 To even obtain a matrix of the form $f(v \X) + \E$, where $v$ is a row of $\V^*$, we can guess the pseudo-inverse of $\U^*$. To reduce the dependency on $m$, we first sketch $\U^*$ by a \textit{subspace-embedding} $\S \in \R^{O(k) \times d}$, which will be a random Gaussian matrix and approximately preserve the column span of $\U^*$. In particular, this approximately preserves the spectrum of $\U^*$. The resulting matrix $\S\U^*$ has $O(k^2)$ entries, and, given the maximum singular value of the inverse (which can be guessed to a factor of $2$), can be guessed accurately enough for our purposes in time $(\kappa/\eps)^{O(k^2)}$, which dominates the overall runtime of the algorithm.

\subsection{Roadmap}
In Section \ref{sec:exact} we introduce our $n^{O(k)}$ time exact algorithm when rank$(A) = k$ and arbitrary $\X$, for recovery of rank-$k$ matrices $\U^*,\V^*$ such that $\U^*f(\V^*\X) = \AA$. In this section, we also demonstrate that for a very wide class of distributions for \textit{random} matices $\X$, the matrix $\U^*f(\V^*\X)$ is in fact full rank with high probability, and therefore can be solved with our exact algorithm. Then, in Section \ref{section:hard}, we prove NP-hardness of the learning problem when rank$(A) < k$. Next, in Section \ref{sec:polyexact}, we give a polynomial time algorithm for exact recovery of $\U^*,\V^*$ in the case when $\X$ has Gaussian marginals in the realizable setting. Section \ref{sec:ICA} develops our Independenct Component Analysis Based algorithm, whereas Section \ref{subsec:generalnoiseless} develops our more general exact recovery algorithm. In Section \ref{sec:symmetric}, we show how recent concurrent results can be bootstrapped via our technqiues to obtain exact recovery for a wider class of distributions. 

In Section \ref{sec:noisycase}, we demonstrate how to extend our algorithm to the case where $\AA = \U^*f(\V^*\X) + \E$ where $\E$ is mean $0$ i.i.d. sub-Gaussian noise. Then in Section \ref{sec:FPT}, we give a fixed-paramater tractable (FPT) (in $k$ and $\kappa(\V^*)$) for the exact recovery of $\U^*,\V^*$ in the case where $\U^*$ does not have full column rank. We give our second FPT algorithm in Section \ref{sec:1bit}, which finds weights which approximate the optimal network for arbitrary error matrices $\E$ that are independent of $\X$. 
In Section \ref{sec:sparsenoise}, we demonstrate how the weights of certain \textit{low-rank} networks, where $k<d,m$, can be recovered exactly in the presence of a class of arbitrary sparse noise in polynomial time. Finally, in Appendix \ref{sec:generalf}, we give further details on generalizing the ReLU to the class of rectified activation functions. 

\subsection{Preliminaries}

For a positive integer $k$, we write $[k]$ to denote the set $\{1,2,\dots,k\}$. We use the term \textit{with high probability} (w.h.p.) in a parameter $r > 1$ to describe an event that occurs with probability $1- \frac{1}{\poly(r)}$. For a real $r$, we will often use the shorthand $\poly(r)$ to denote a sufficiently large constant degree polynomial in $r$. Since for simplicity we do not seek to analyze or optimize the polynomial running time of our algorithms, we will state many of our error bounds within technical lemmas as $\frac{1}{\poly(r)}$ where $r$ constitutes some set of relevant parameters, with the understanding that this polynomial can be made arbitrarily large by increasing the sample complexity $n$ of our algorithms by a polynomial factor.

In this work we use boldface font $\AA,\V,\U,\W$ to denote matrices, and non-boldface font $x,y,u,v$ to denote vectors.
For a vector $x$, we use $\|x\|_2$ to denote the $\ell_2$ norm of $x$. 
For any  
matrix $\W$ with $p$ rows and $q$ columns, for all $i\in [p]$, let $\W_{i,*}$  denote 
the $i$-th row of $\W$,  for all $j \in [q]$ let $\W_{*,j}$ denote the $j$-th column
and let $\W_{i,j}$ denote the $i,j$-th entry of $\W$. 
Further, the 
singular value decomposition of $\W$, denoted by $\texttt{SVD}(\W) = \U
\mathbf{\Sigma} \V^T$, is such that $\U$ is a $p \times r$ matrix with 
orthonormal columns, $\V^T$ is a $r \times q$ matrix with orthonormal 
rows and $\mathbf{\Sigma}$ is an $r \times r$ diagonal matrix, where $r$ is the rank of $\W$. The entries 
along the diagonal are the singular values of $\W$, denoted by 
$\sigma_{\max} = \sigma_1(\W) \geq  \sigma_2(\W) \geq \ldots\geq  \sigma_r(\W) = \sigma_{\min}(\W)$. We write $\|\W\|_F = (\sum_{p,q} \W_{p,q}^2\big)^{1/2}$ to denote the Frobenius norm of $\W$, and $$\|\W\|_2 = \textrm{sup}_x \frac{\|\AA x\|_2}{\| x\|_2} = \sigma_{\textrm{max}}(\W)$$ 
to denote 
the spectral norm. We will write $\mathbb{I}_k$ to denote the $k \times k$ square identity matrix. We use the notation $\Proj_{\W}(w)$ to denote the projection of the vector $w$ onto the \textit{row-span} of $\W$. In other words, if  $x^* = \arg \min_x \|x\W - w\|_2$, then $\Proj_{\W}(w) = x^* \W$. We now recall the condition number of a matrix $\W$.

\begin{definition}
For a rank $k$ matrix $\W \in \R^{p \times q}$, let $\sigma_{\max}(\W) = \sigma_1(\W) \geq \sigma_2(\W) \geq \dots \geq \sigma_k(\W) = \sigma_{\min}(\W)$ be the non-zero singular values of $\W$. Then the \textit{condition number} $\kappa(\W)$ of $\W$ is given by
\[ \kappa(\W) = \frac{\sigma_{\max}(\W)}{\sigma_{\min}(\W)} \]
Note that if $\W$ has full column rank (i.e., $k=q$), then if $\W^\dagger$ is the pseudo-inverse of $\W$ we have $\W^\dagger \W = \mathbb{I}_q$ and 
\[ \kappa(\W) = \|\W^\dagger\|_2 \|\W\|_2 \]
where $\|\W\|_2 = \sigma_1(\W)$ is the spectral norm of $\W$. Similarly if $\W$ has full row rank (i.e. $k=p$), then $\W \W^\dagger = \mathbb{I}_p$ and 
\[ \kappa(\W) = \|\W^\dagger\|_2 \|\W\|_2 \]
\end{definition}

A real $m$-th order tensor is $\T \in \otimes^m \R^{d}$ is the outer product of $m$ $d$-dimensional Euclidean spaces. A third order tensor $\T \in \otimes \R^{d}$ is defined to be rank-$1$ if $\T = w\cdot a \otimes b \otimes c$ where $a, b,c \in \R^d$. Further, $\T$ has Candecomp/Parafac (CP) rank-$k$ if it can be written as the sum of $k$ rank-$1$ tensors, i.e., $$\T = \sum^k_{i=1} w_i a_i \otimes b_i \otimes c_i$$ is such that $w_i \in \R, a_i, b_i, c_i \in \R^d$. Next, given a function $f(x):\R^d \to \R$, we use the notation $\nabla^{m}_x f(x) \in \otimes^m \R^d$ to denote the $m$-th order derivative operator w.r.t. the variable $x$, such that $$\left[\nabla^{m}_x f(x)\right]_{i_1, i_2, \ldots i_m} = \frac{\partial f(x)}{\partial x_{i_1}\partial x_{i_2} \ldots \partial x_{i_m} }$$.


In the context of the ReLU activation function, a useful notion to consider is that of a sign pattern, which will be used frequently in our analysis. 
\begin{definition}
 For any matrix dimensions $p, q$, a\textit{ sign pattern} is simply a subset of $[p] \times [q]$. For a matrix $\W \in \R^{p \times q}$, we let $\sg{\W}$ be the sign pattern defined by 
\[\sg{\W} = \{ (i,j) \in [p] \times [q] \; | \; \W_{i,j} > 0 \} \]
\end{definition}

Intuitively, in the context of rectified activation functions, the sign pattern is an important notion since $\sg{\W}$ is invariant under application of $f$, in other words $\sg{\W} = f(\sg{\W})$. We similarly define a \textit{sparsity-pattern} of a matrix  $\W \in \R^{p \times q}$ as a subset of $[p]\times[q]$ where $\W$ is non-zero. Note that a sign and sparsity pattern of $\W$, taken together, specify precisely where the strictly positive, negative, and zero-valued entries are in $\W$.

We use the notation $\mathcal{N}(\mu,\sigma^2)$ to denote the Gaussian distribution with mean $\mu$ and variance $\sigma^2$. More generally, we write $\mathcal{N}(\mu,\mathbf{\Sigma})$ to denote a $k$-dimensional multi-variate Gaussian distribution with mean $\mu \in \R^k$ and variance $\mathbf{\Sigma} \in \R^{k \times k}$. We make use of the $2$-stability of the Gaussian distribution several times in this work, so we now recall the following definition of stable random variables. We refer the reader to \cite{indyk2006stable} for a further discussion of such distributions.
\begin{definition}\label{def:stable}
	A distribution $\mathcal{D}_p$ is said to be $p$-stable if whenever $\X_1,\dots,\X_n \sim \mathcal{D}_p$ are drawn independently, we have\[	\sum_{i=1}^n a_i \X_i \sim \|a\|_p \X	\]
	for any fixed vector $a \in \R^n$, where $\X \sim \mathcal{D}_p$ is again distributed as a $p$-stable random variable. In particular, the Gaussian random variables $\mathcal{N}(0,\sigma^2)$ are $p$-stable for $p=2$ (i.e., $\sum_i a_i g_i = \|a\|_2$, where $g,g_1,\dots,g_n \sim \mathcal{N}(0,\sigma^2)$). 
\end{definition}

Finally, we  remark that in this paper, we will work in the common real RAM model of computation, where arithmetic operations on real numbers can be performed in constant time.

\section{Exact solution when rank$(\AA) = k$}\label{sec:exact}
In this section, we consider the exact case of the neural network recovery problem. Given an input matrix $\X \in \R^{d \times n}$ of examples, and a matrix $\AA \in \R^{m \times n}$ of classifications, the exact version of the recovery problem is to obtain rank-$k$ matrices $\U^*,\V^*$ such that $\AA = \U^*f(\V^*\X)$, if such matrices exist. In this section we demonstrate the existence of an $n^{O(k)} \poly(md)$-time algorithm for exact recovery when $\text{rank}(\AA) =  k$. We demonstrate that this assumption is likely necessary in Section \ref{section:hard}, where we show that if $\text{rank}(\AA) <  k$ then the problem is NP-hard even for any $k\geq 2$ when the matrix $\U$ is given as input, and NP-hard for $k=2$ when $\U^*$ is allowed to be a variable. This rules out the existence of a general $n^{O(k)}$ time algorithm for this problem.

The main theorem we prove in this section is that there is an algorithm with running time dominated by $\min\{n^{O(k)},\allowbreak 2^n\}$ such that it recovers the underlying matrices $\U^*$ and $\V^*$ exactly. Intuitively, we begin by showing a structural result that there are at most $n^{O(k)}$ sign patterns that lie in the row space of $f(\V^*\X)$ and we can efficiently enumerate over them using a linear program. For a fixed sign pattern in this set, we construct a sequence of $k$ linear programs (LP) such that the $i$-th LP finds a vector $y^i$, $f(y^i)$ is in the row span of $f(\V^*\X)$, subject to the fixed sign pattern, and the constraint that $f(y^i)$ is not a linear combination of $f(y^1), f(y^2), \ldots f(y^{i-1})$. We note that $f(y^i)$ being linearly independent is not a linear constraint, but we demonstrate how it can be linearized in a straightforward manner. 

Crucially, our algorithm relies on the fact that we have the row-span of $f(\V^*\X)$. Note that this is implied by the assumption that $\AA$ is rank $k$. Knowing the rowspan allows us to design the constraints in the prior paragraph, and thus solve the LP to recover the rows of $f(\V^*\X).$ On the other hand, if the rank of $\AA$ is less than $k$, then it no longer seems possible to efficiently determine the row span of $f(\V^*\X)$. In fact, our NP-Hardness result of Section \ref{section:hard} demonstrates that, given $\U^*$ as input, if the rank of $\AA$ is strictly less than $k$, the problem of determining the exact row-span of $f(\V^*\X)$
is NP-Hard. The main result of this section is then as follows.

\begin{theorem}
\label{thm:exact_lp}
Given $\AA \in \R^{m \times n}, \X \in \R^{d \times n}$, there is an algorithm that finds $\U^* \in \R^{m \times k} ,\V^* \in \R^{k \times d}$ such that $\AA = \U^* f(\V^*\X)$ and runs in time $\poly(nmd)\min\{n^{O(k)},\allowbreak 2^n\}$, if $\textrm{rank}(\AA) =k$.
\end{theorem}

Let $\V' \in \R^{k \times n}$ be a basis for the row-span of $\AA$. 
	For two matrices $\Y,\ZZ$ of the same dimension, we will write $\Y \req \ZZ$ if the row spans of $\Y$ and $\ZZ$ are the same. The first step in our algorithm is to obtain a \textit{feasible set} $\mathcal{S}$ of sign patterns, within which the true sign pattern of $f(\V^* \X)$ must lie. 
	
\begin{lemma}\label{prop:signs}
Given $\AA \in \R^{m \times n}, \X \in \R^{d \times n}$, such that $\textrm{rank}(\AA) = k$, there is an algorithm which runs in time $\min\{n^{O(k)},2^n\}$ and returns a set of sign patterns $\mathcal{S} \subset 2^{[m] \times [n]}$ with  $|\mathcal{S}| = \min\{n^{O(k)},2^n\}$ such that for any rank-$k$ matrices $\U^* \in \R^{m \times k},\V^* \in \R^{k \times d}$ such that $\AA = \U^* f(\V^*\X)$ and any row $i \in [k]$, $\sg{(\V^*\X)_i} = \sg{S}$ for some $S \in \mathcal{S}$. 
\end{lemma}
\begin{proof}

	Recall, $\AA$ is rank $k$. Thus there is a subset $\V' \in \R^{k \times n}$ of $k$ rows of $\AA$ which span all the rows of $\AA$. Critically, here we require that the rank of $\AA$ is $k$ and thus the row space of $\AA$ is the same as that of $f(\V^*\X)$. Since $\AA = \U^* f(\V^*\X)$ and $\V',f(\V^*\X)$ have the same dimensional row space, the row spaces of $\V'$ and $f(\V^*\X)$ are precisely the same, and so there must be an invertible change of basis matrix $\W$ such that $\W\V' = f(\V^*\X)$. Now note that $\sg{\V^*\X} = \sg{f(\V^*\X)} = \sg{\W\V'}$, and thus it suffices to return a set of sign patterns $\mathcal{S}$ which contains $\sg{\W\V'}$. Therefore, consider any fixed sign pattern $S  \subset [n]$, and fix a row $j \in [k]$, and consider the following feasibility linear program in the variables $w_j$
	\[  (w_j\V')_i \geq 1, \; \; \; \; \text{ for all } i \in  \sg{S} \] 
	\[  (w_j\V')_i \leq 0 , \; \; \; \; \text{ for all } i  \notin  \sg{S} \] 
	Note that if the sign pattern $S$ is feasible by some $w_j\V'$, then the above LP will be feasible with a suitably large positive scaling to $w_j$. Now the LP has $k$ variables and $n$ constraints, and thus a solution is obtained by choosing the $w_j$ that makes a subset of $k$ linearly independent constraints tight. Observe in any such LP of the above form, there are at most $2n$ possible constraints that can ever occur. Thus if $S$ is realizable as the sign pattern of some $w_j\V'$, then it is obtained by the unique solution to a system which chooses to make $k$ of these constraints tight. Formally, if $S,b$ are the constraints for which $w_j S \geq b$ in the LP, then a solution is given by $w_j S' = b'$ where $S',b'$ are a subset of $k$ of the constraints.  Since there are at most $\binom{2n}{k} = O(n^k)$ such possible choices, it follows that there are at most $O(\min\{n^{O(k)},2^n\})$ realizable sign patterns, and these can be enumerated in  $O(\min\{n^{O(k)},2^n\})$ time by simply checking the sign pattern which results from the solution (if one exists) to $w_j S' = b'$ taken over all subsets $S',b'$ of constraints of size $k$.

\end{proof}

Given access to the set of candidate sign patterns, $S \in \mathcal{S}$, and vectors $y^1,y^2,...,y^{i-1} \in \R^n$, we can define the following iterative feasibility linear program, that at each iteration $i$ finds a vector $y^i$ which is equal to some vector in the row span of $\X$, and such that $f(y^1),f(y^2),\dots, f(y^i)$ are all linearly independent and in the row span of $\AA$. 

\begin{Frame}[\textbf{Algorithm \ref{alg:iter_lp} : Iterative LP}$\left( \X, S,y^1, y^2, \ldots y^{i-1}\right)$.]
\label{alg:iter_lp}
\ttx{Input}: Matrix $\X$, a sign pattern $S$,  vectors $y^1, y^2, \ldots y^{i-1}$ such that $f(y^1), f(y^2), \ldots f(y^{i-1})$ are linearly independent.
\begin{enumerate}
	   \item Let $y^i,z^i,w^i$ be variables in $\R^n$. 
	   \item Let $\Q \in \R^{(i-1) \times n}$ be a matrix such that for all $j \in [i-1]$, $\Q_{j,*} = f(y^j)$. Construct the projection matrix $\P^{i-1}$ onto  $\textrm{span}\left\{f(y^1),f(y^2),...,f(y^{i-1})\right\}$. Note, the projection matrix is given by $\P^{i-1} = \Q^T(\Q^T\Q)^{-1} \Q$. 
	   \item Define $f_S(y^i)$ w.r.t. the sign pattern $S$ such that 
	    \[ f_S(y^i_j) = \begin{cases} (y^i_j) & \text{ if } j \in S\\
	    0 & \text { otherwise }
	    \end{cases} \] 
\end{enumerate}
 \ttx{Output}: A feasible solution to the following LP:
    \begin{equation*}
        \begin{array}{ll@{}ll}
        \forall j \in [n] & \displaystyle    & y^{i}_j \geq 1,  & \textrm{ if } j \in S \\
        \forall j \in [n] & \displaystyle    & y^{i}_j \leq 0,  & \textrm{ if } j \notin S \\
        & &  y^i = w^i \X &\\
        & &  f_S(y^i) = z^i \V' &\\
        & &   f_S(y^i)(I - \P^{i-1})\neq 0 &\\
        \end{array}
    \end{equation*}
\end{Frame}

\begin{remark}\label{remark:linear}
Observe, while the last constraint is not a linear constraint, it can be made linear by running $2n$ consecutive LP's, such that, for $t \in [n]$, in the $2t$-th LP we replace the constraint $f_S(y^i)(\mathbb{I} - \P^{i-1})\neq 0$ above with 
$$\left[f_S(y^i)\left(\mathbb{I} - \P^{i-1}\right)\right]_t \geq 1$$
and in the $(2t-1)$-th LP we replace  constraint $f_S(y^i)\left(\mathbb{I} - \P^{i-1}\right)\neq 0$ with 
$$\left[f_S(y^i)\left(\mathbb{I}  - \P^{i-1}\right)\right]_t \leq -1$$
Note, the modified constraints are linear in the variables $y^i$.
If there is a vector $y^i$ which satisfies the above constraints such that $f_S(y^i)(\mathbb{I}  - \P^{i-1})\neq 0$, then by scaling $y^i,w^i,z^i$ all by a sufficiently large positive constant, then $y^i$ will also satisfy one of the $2n$ LPs described above, thus giving a solution to the original feasibility problem by returning the first feasible solution returned among the $2n$ new LPs. 
\end{remark}

Using Algorithm \ref{alg:iter_lp} as a sub-routine, we iterate over all sign patterns 
$S\in \mathcal{S}$, such that we recover a linearly independent set of $k$ vectors 
$f(y^1),f(y^2), \ldots f(y^k)$. Let $\Y$ be a matrix such that the $j$-th row corresponds
to $y^j$. We then set up and solve two linear systems in $\U$ and $\V$, given by 
$\AA = \U f(\Y)$  and $\Y = \V \X$. We show that the solutions to the linear system 
correspond to $\U^*$ and $\V^*$. Here, we note that since the optimal $\U^*$ and $\V^*$ 
are solutions to a linear system, we can recover them exactly.

\begin{Frame}[\textbf{Algorithm \ref{alg:overall_exact} : ExactNeuralNet}$(\AA,  \X, \mathcal{S})$.]
\label{alg:overall_exact}
\ttx{Input}: Matrices $\AA, \X$, a set of sign patterns $\mathcal{S}$.
\begin{enumerate}
    \item For $i=1,2,\dots,k$
	 \begin{enumerate}[topsep=0pt,itemsep=-1ex,partopsep=1ex,parsep=1ex] 
			\item $t=1$.
					\item While$(t \leq |\mathcal{S}|)$
					\begin{enumerate}[topsep=0pt,itemsep=-1ex,partopsep=1ex,parsep=1ex] 
					\item If \texttt{Iterative LP}$\left(\X, S_t,y^1,y^2,...,y^{i-1}\right)$ is feasible, let $y^i$ be the output, and set $t =  |\mathcal{S}| + 1$.
    					\item Else $t \leftarrow t+1$. 
				\end{enumerate}
				
	\end{enumerate}
    \item Let $\Y \in \R^{k \times n}$ be the matrix with $j$-th row equal to $y^j$ and let $S$ be the corresponding sign pattern. 
    \item Let $\U^*$ be the solution to the linear system in $\U$ given by $\AA = \U f_S(\Y)$. 
    \item Let $\V^*$ be the solution to the linear system in $\V$ given by $\Y = \V\X$. 
\end{enumerate}
\ttx{Output:}  $\U^* ,\V^*$. 
\end{Frame}
		
		\begin{lemma}\label{prop:following}
		 For any $i \in [k]$ vectors $y^1,y^2,...,y^{i-1} \in \R^n$ and $S \in \mathcal{S}$, let $y^i$ be a feasible solution to \texttt{Iterative LP}$\left(\X, S,y^1,y^2,...,y^{i-1}\right)$. Then all of the following hold:
		 \begin{enumerate}
		  \item $y^i$ is in the row span of $\X$.
		    \item $f(y^i)$ is in the row span of $\AA$.
		     \item $f(y^i)$ is independent of $f(y^1),f(y^2),...,f(y^{i-1})$.
		 \end{enumerate}
		\end{lemma}
	
	\begin{proof}
	    The first condition follows due to the third constraint $y^i = w^i \X$. The first and second constraint ensure that $f_S(y^i) = f(y^i),$ thus along with the fourth constraint and the fact that $\V'$ spans the rows of $\AA$, the second condition follows. For the last condition, it suffices to show that if $\|f(y^i)(\mathbb{I} - \P^{i-1})\| \geq 1$ then $f(y^i)$ is not in the span of $\{f(y^1),\dots,f(y^{i-1})\}$. Now if $f(y^i)(\mathbb{I} - \P^{i-1}) = z \neq 0$, then $f(y^i) = z + \Proj_{i-1}(f(y^i))$, where $\Proj_{i-1}(f(y^i))$ is the projection of $f(y^i)$ onto the subspace spanned by $\{f(y^1),\dots,f(y^{i-1})\}$. If $f(y^i)$ was in this subspace, then we would have $\Proj_{i-1}(f(y^i)) = y^i$, but this is impossible since $z \neq 0$, which completes the proof.
	\end{proof}

	\begin{lemma}
	 Suppose that there exist matrices $\U^* \in \R^{m \times k} ,\V^* \in \R^{k \times d}$ with $\AA = \U^* f(\V^*\X)$. Then in the above algorithm, for each $i \in [k]$  \texttt{Iterative LP}$\left(\X, S_t,y^1,y^2,...,y^{i-1}\right)$ will be feasible for at least one $S_t \in \mathcal{S}$. 
	\end{lemma}
	\begin{proof}
	    The proof is by induction. For $i=1$,  since $f(\V^*\X)$ has rank $k$ and spans the rows of $\AA$, it follows that there must be some $j \in [k]$ such that the $j$-th row $f(\V^*\X)_j$ of $f(\V\X)$ is in in the row span of $\V'$, and clearly $(\V^*\X)_j$ is in the row span of $\X$. The last constraint is of the LP non-existent since $i=1$. Furthermore, $(\V^*\X)_j$ has some sign pattern $S^*$, and it must be that $S^* \in \mathcal{S}$ by construction of $\mathcal{S}$. Then there exists a positive constant $c > 0$ such that $(c\V^*\X)_j$ satisfies the last constraints of \texttt{Iterative LP}$\left(\X, S^*,y^1,y^2,...,y^{i-1}\right)$ (made linear as described in Remark \ref{remark:linear}), and multiplying $(\V^*\X)_j$ by a positive constant does not affect the fact that $(c\V^*\X)_j$ is in the row space of $\X$ and $f(c\V^*\X)_j$ is in the row space of $\AA$ by closure of subspaces under scalar multiplication. Thus the \texttt{Iterative LP}$\left(\X, S^*,y^1,y^2,...,y^{i-1}\right)$ has a feasible point. 
	    
	    Now suppose we have feasible points $y^1,\dots,y^{i-1}$, with $i \leq k$. Note that this guarantees that $f(y^1),\dots,f(y^{i-1})$ are linearly independent. Since $f(\V^*\X)$ spans the $k$-dimensional row-space of $\AA$, there must be some $j$ with $f(\V^*\X)_j$ that is linearly independent of $f(y^1),\dots,f(y^{i-1})$ such that $f(\V^*\X)_j$ is in the row span of $\AA$. Then $(\V^*\X)_j$ is in the row span of $\X$, and similarly $(\V\X)_j$ has some sign pattern $S^*$, and after multiplication by a suitably large constant it follows that the \texttt{Iterative LP}$\left(\X, S^*,y^1,y^2,...,y^{i-1}\right)$ will be feasible. The proposition follows by induction. 
	\end{proof}
	
\noindent \textbf{Proof of Theorem \ref{thm:exact_lp}.}	
By Proposition \ref{prop:following}, $f(y^1),\dots,f(y^{k})$ are independent, and give a solution to $f(\V\X) \req \AA$. Thus we can find a $\U  \in \R^{d \times k}$ in polynomial time via $d$ independent linear regression problems that solves $\U f(\V\X) = \AA$. By Proposition \ref{prop:signs}, there are at most $\min\{n^{O(k)},2^n\}$ sign patterns in the set $\mathcal{S}$, and solving for each iteration of \ttx{Iterative LP} takes $\poly(nm)$-time. Thus the total time is $\poly(nmd)\min\{n^{O(k)},2^n\} $ as stated.

\subsection{Rank$(\AA) = k$ for random matrices $\X$.}
	We conclude this section with the observation that if the input $\X$ is drawn from a large class of independent distributions, then the resulting matrix $\U^* f(\V^*\X)$ will in fact be rank $k$ with high probability if $\U^* $ and $\V^*$ are rank $k$. Therefore, Algorithm \ref{alg:overall_exact} recovers $\U^*, \V^*$ in $\poly(nmd)\min\{n^{O(k)},2^n\} $ for all such input matrices $\X$.

\begin{lemma}\label{prop:random}
	Suppose $\AA = \U^* f(\V^*\X)$ for rank $k$ matrices $\U^*  \in \R^{m \times k}$ and $\V^* \in \R^{k \times d}$, where $\X \in \R^{d \times n}$ is a matrix of random variables such that each column $\X_{*,i}$ is drawn i.i.d. from a distribution $\mathcal{D}$ with continuous p.d.f. $p(x): \R^d \to \R$ such that $p(x) >0$ almost everywhere in $\R^d$, and such that 
	\[  \inf_{v \in \R^d} \text{Pr}_{x \sim \mathcal{D}}\big[ \langle v , x\rangle > 0 \big] > 10 k \log(k/\delta)/n \]
	Then rank$(A) = k$ with probability $1 - O(\delta)$.
\end{lemma}
\begin{proof}
	By Sylverster's rank inequality, it suffices to show $f(\V^*\X)$ is rank $k$. By symmetry and i.i.d. of the $\X_{ij}$'s in a fixed row $i$, each entry $f(\V^*\X)_{ij}$ is non-zero with probability at least $10k\log(k/\delta)/n$ independently (within the row $i$). Then by Chernoff bounds, a fixed row $(\V^*\X)_{i,*}$ will have at least $k$ positive entries with probability at least $1- 2^{-k \log(k/\delta)}$, and we can then union bound over all $k$ rows to hold with probability at least $1-O(\delta)$. Thus one can pick a $k \times k$ submatrix $\W$ of $f(\V^*\X)$ such that, under some permutation $\W'$ of the columns of $\W$, the diagonal of $\W'$ is non-zero.
	
	Since $\V^*$ is rank $k$, $\V^*$ is a surjective linear mapping of the columns of $\X$ from $\R^d$ to $\R^k$. Since $p(x) > 0$ almost everywhere, it follows that $p_{\V^*}(x) > 0$ almost everywhere, where $p_{\V^*}(x)$ is the continuous pdf of a column of $\V^*\X$. Then if $\X'$ is any matrix of $k$ columns of $\X$, by independence of the columns, if $p_{k \times k}:\R^{k^2} \to \R$ is the joint pdf of all $k^2$ variables in $\V^*\X'$, it follows that $p_{k \times k}(x) > 0$ for all $x \in \R^{k^2}$. Thus, by conditioning on any sign pattern $S$ of $\V^*\X'$, this results in a new pdf $p^S_{k \times k}$, which is simply $p_{k \times k}$ where the domain is restricted to an orthant $\Omega$ of $R^{k^2}$. Since $p_{k \times k}$ is continuous and non-zero almost everywhere, it follows that the support of the pdf $p^S_{k \times k}: \Omega \to \R$ is all of $\Omega$. In particular, the Lesbegue measure of the support $\Omega$ inside of $\R^{k^2}$ is non-zero (note that this would not be true if $\V^*$ has rank $k' < k$, as the support on each column would then be confined to a subspace of $\R^k$, which would have Lesbegue measure zero in $\R^k$). 
	
Now after conditioning on a sign pattern, $\det(\W')$ is a non-zero polynomial in $s$ random variables, for $k \leq s \leq k^2$, and it is well known that such a function cannot vanish on any non-empty open set in $\R^s$ (see e.g. Theorem 2.6 of \cite{ConradNotes}, and note the subsequent remark on replacing $\C^s$ with $\R^s)$.  It follows that the set of zeros of $\det(\W')$ contain no open set of $\R^s$, and thus has Lesbegue measure $0$ in $\R^s$. By the remarks in the prior paragraph, we know that the Lesbegue measure (taken over $\R^s$) of the support of the joint distribution on the $s$ variables is non-zero (after restricting to the orthant given by the sign pattern). In particular, the set of zeros of $\det(\W')$ has Lesbegue measure $0$ inside of the support of the joint pdf of the non-zero variables in $\W'$. We conclude that the joint density of the variables of $\W'$, after conditioning on a sign pattern, integrated over the set of zeros of $\det(\W')$ will be zero, meaning that $\W'$ will have full rank almost surely, conditioned on the sign pattern event in the first paragraph when held with probability $1-O(\delta)$.
	\end{proof}

\begin{remark}
Note that nearly all non-degenerate distributions $\mathcal{D}$ on $d$-dimensional vectors will satisfy $\inf_{v \in \R^d} \text{Pr}_{x \sim \mathcal{D}}\big[ \langle v , x\rangle > 0 \big] = c= \Omega(1)$. For instance any multi-variate Gaussian distribution with non-degenreate (full-rank) covariance matrix $\mathbf{\Sigma}$ will satisfy this bound with $c = 1/2$, and this will also hold for any symmetric i.i.d. distribution over the entries of $x \sim \mathcal{D}$. Thus it will suffice to take $n = \Omega( k \log(k/\delta))$ for the result to hold. 

\end{remark}

\begin{corollary}\label{cor:random_input}
Let $\AA = \U^* f(\V^*\X)$ for rank $k$ matrices $\U^*  \in \R^{m \times k}$ and $\V^* \in \R^{k \times d}$, where $\X \in \R^{d \times n}$ is a matrix of random variables such that each column $\X_{*,i}$ is drawn i.i.d. from a distribution $\mathcal{D}$ with continuous p.d.f. $p(x): \R^d \to \R$ such that $p(x) >0$ almost everywhere in $\R^d$, and such that 
\[  \inf_{v \in \R^d} \text{Pr}_{x \sim \mathcal{D}}\big[ \langle v , x\rangle > 0 \big] = \Omega(k \log(1/\delta)/n) \]
Then, there exists an algorithm such that, with probability $1 - O(\delta)$, recovers $\U^*$, $\V^*$ exactly and runs in time $\poly(n,m,d,k)\min\{n^{O(k)},2^n\}$.
\end{corollary}

\input{Hardness.tex}
\input{SimpleExactGaussianCase.tex}

\input{ExactGaussianCase.tex}

\input{1bitCompressed.tex}

\input{SparseNoise.tex}

\bibliography{cluster}
\input{Appendix.tex}

\end{document}

%% file: Hardness.tex


\section{NP-Hardness}\label{section:hard}
The goal of this section is to prove that the problem of deciding whether there exists $\V \in \R^{k \times d}$ that solves the equation $\alpha f(\V\X) = w$ for fixed input $\alpha \in \R^{m \times k},X \in \R^{d \times n},A \in \R^{m \times n}$, is NP-hard.
We will first prove the NP-hardness of a geometric separability problem, which will then be used to prove NP-hardness for the problem of deciding the feasibility of $\alpha f(\V\X) = w$.
Our hardness reduction is from a variant of Boolean SAT, used in \cite{megiddo1988complexity} to prove NP-hardness of a similar geometric seperability problem, 
called \textit{reversible 6-SAT}, which we will now define. For a Boolean formula $\psi$ on variables $\{u_1,\dots,u_n,\overline{u_1},\dots,u_{n}\}$ (where $\overline{u_i}$ is the negation of $u_i$), let $\overline{\psi}$ be the formula where every variable $u_i$ and $\overline{u_i}$ appearing in $\psi$ is replaced  with $\overline{u_i}$ and $u_i$ respectively. For instance, if $\psi = (u_1 \vee u_2 \vee \overline{u_3} ) \wedge (\overline{u_2} \vee u_3)$ then $\overline{\psi} = (\overline{u_1} \vee \overline{u}_2 \vee u_3) \wedge (u_2 \vee \overline{u_3})$. 
	\begin{definition}
	A Boolean formula $\psi$ is said to be reversible if $\psi$ and $\overline{\psi}$ are both either satisfiable or not satisfiable. 
	\end{definition}
	The reverse $6$-SAT problem is then to, given a reversible Boolean formula $\psi$ where each conjunct has exactly six literals per clause, determine whether or not $\psi$ is satisfiable.  Observe, if $\xi$ is a satisfying assignment to the variables of a reversible formula $\psi$, then $\overline{\xi}$, obtained by negating each assignment of $\xi$, is a satisfying assignment to $\overline{\psi}$. The following can be found in \cite{megiddo1988complexity}.
	
	\begin{proposition}[NP-Hardness of Reversible 6-SAT]\cite{megiddo1988complexity}]\label{prop:hard}
	     Given a reversible formula $\psi$ in conjunctive normal form where each clause has exactly six literals, it is NP-hard to decide whether $\psi$ is satisfiable. 
	     \end{proposition}

We now introduce the following \textit{ReLU-seperability} problem, and demonstrate NP-hardness via a reduction from reversible $6$-SAT.

\begin{definition}[ReLU-separability.]\label{def:relusep}
     Given two sets $P=\{p_1,...,p_r\}, Q = \{q_1,\dots,q_s\}$ of vectors in $R^d$, the \textit{ReLU-seperability} is to find vectors $x,y \in \mathbb{R}^d$ such that
     \begin{itemize}
    \item For all $p_i \in P$, both $p_i^T x \leq 0$ and $p_i^T y \leq 0$.
    \item For all $q_i \in Q$, we have $f(q_i^Tx) + f(q_i^T y) = 1$ where $f(\cdot) = \max (\cdot,0)$ is the ReLU function.
\end{itemize}
We say that an instance of ReLU-seperability is satisfiable if there exists such an $x,y \in \R^d$ that satisfy the above conditions. 
\end{definition}
\begin{proposition}\label{prop:reluhard}
 It is NP-Hard to decide whether an instance of ReLU-seperability is satisfiable. 
\end{proposition}
\begin{proof}
    
 Let $u_1,\dots,u_n$ be the variables of the reversible 6-SAT instance $\psi$, and set $d = n+2$, and let $x,y$ be the solutions to the instance of ReLU separability which we will now describe. The vector $x$ will be such that $x_i$ represents the truth value of $u_i$, and $y_i$ represents the truth value of $\overline{x_i} = \overline{u_i}$. For $j \in [n+2]$, let $e_j \in \R^{n+2}$ be the standard basis vector with a $1$ in the $j$-th coordinate and $0$ elsewhere. For each $i \in [n]$, we insert $e_i$ and $-e_i$ into $Q$. This ensures that $f(x_i) + f(y_i) = 1$ and $f(-x_i) + f(-y_i) = 1$. This occurs iff either $x_i = 1$ and $y_i = -1$ or $x_i = -1$ and $y_i = 1$, so $y_i$ is the negation of $x_i$. In other words, the case $x_i = 1$ and $y_i = -1$ means $u_i$ is true and $\overline{u_i}$ is false, and the case $x_i = -1$ and $y_i = 1$ means $u_i$ is false and $\overline{u_i}$ is true. Now suppose we have a clause of the form $u_1 \vee \overline{u_2} \vee u_3 \vee u_4 \vee \overline{u_5} \vee u_6$ in $\psi$. Then this clause can be represented equivalently by the inequality $x_1 - x_2 + x_3 + x_4 - x_5 + x_6 \geq -5$. 

To represent this affine constraint, we add additional constraints that force $x_{n+1} + x_{n+2} = 1/2$ and $y_{n+1} + y_{n+2} = 1/2$ (note that the $n+1$, and $n+2$ coordinates do not correspond to any of the $n$ variables $u_i$). We force this as follows. Add $e_{n+1}$ and $e_{n+2}$ to $Q$, and add $-2e_{n+1}$, $-2e_{n+2}$  to $Q$.  This forces $f(x_i) + f(y_i) = 1$ and $f(-2x_i) + f(-2y_i) = 1$ for each $i \in \{n+1,n+2\}$. For each $i \in \{n+1,n+2\}$ there are only two solutions, either $x_i = 1$ and $y_i = -1/2$ or $x_i = -1/2$ and $y_i=1$. Finally, we add the vector $e_{n+1} + e_{n+2}$ to $Q$, which forces $f(x_{n+1} + x_{n+2}) + f(y_{n+1} + y_{n+2}) = 1$. Now if $x_{n+1} = 1$, then $x_{n+2}$ must be $-1/2$ since otherwise there is no solution to $2 + f(\cdot) = 1$, and we know $x_{n+2} \in \{1,-1/2\}$. This forces $y_{n+2} = 1$, which forces $x_{n+1} + x_{n+2} = 1/2 = y_{n+1} + y_{n+2}$, and a symmetric argument goes through when one assumes $y_{n+1} = 1$. This lets us write affine inequalities as follows. For the clause $u_1 \vee \overline{u_2} \vee u_3 \vee u_4 \vee \overline{u_5} \vee u_6$, we can write the corresponding equation  $x_1 - x_2 + x_3 + x_4 - x_5 + x_6 \geq -5$ precisely as a point constraint, which for us is $(-1,1,-1,-1,1,-1,0,0,\dots,0,-10,-10) \in P$ (the two $-10$'s are in coordinate positions $n+1$ and $n+2$). Now this also forces the constraint $y_1 - y_2 + y_3 + y_4 - y_5 + y_6 \geq -5$, but since the formula is reversible so we can assume WLOG that $\overline{u_1 }\vee u_2 \vee \overline{u_3} \vee \overline{u_4 } \vee u_5 \vee \overline{u_6}$ is also a conjunct and so the feasible set is not affected, and the first $n$ coordinates of any solution $x$ will indeed correspond to a satisfying assignment to $\psi$ if one exists. Since reversible $6$-SAT is NP-hard by Proposition \ref{prop:hard}, the stated result holds.
\end{proof}

\begin{theorem}\label{prop:truehard}
     For a fixed $\alpha \in \R^{m \times k},\X \in \R^{d \times n},\AA \in \R^{m \times n}$, the problem of deciding whether there exists a solution $\V \in \R^{k \times d}$ to $\alpha f(\V\X) = \AA$ is NP-hard even for $k=2$. Furthermore, for the case for $k=2$, the problem is still NP-hard when $\alpha \in \R^{m \times 2}$ is allowed to be a variable. 
\end{theorem}
\begin{proof}

Now we show the reduction from ReLU-separability to our problem. Given an instance $(P,Q)$ of ReLU separability as in Definition \ref{def:relusep}, set $\alpha = [1,1]$, and $w = [0,0,\dots,0,1,1,\dots,1]$ so $w_i = 0$ for $i \leq r$ and $w_i = 1$ for $r< i \leq r+s$. Let $\X = [ p_1, \; p_2, \; \dots, p_r,\; q_1 ,\dots,\; q_s] \in \mathbb{R}^{d \times (r+s)}$. Now suppose we have a solution $\V = [x,y]^T \in \mathbb{R}^{(r+s) \times 2}$ to $\alpha f(\V\X) = w$. This means $f(p_i^T x) + f(p_i^t y) = 0$ for all $p_i \in P$, so it must be that both $p_i^T x \leq 0 $and $p_i^T y \leq 0$. Also, we have $f(q_i^T x) + f(q_i^t y) = 1$  for all $q_i \in Q$. These two facts together mean that $x,y$ are a solution to ReLU-separability. Conversely, if solutions $x,y$ to ReLU separability exist, then for all $p_i \in P$, both $p_i^T x \leq 0$ and $p_i^T y \leq 0$ implies $f(p_i^T x) + f(p_i^t y) = 0$, and for all $q_i \in Q$ we get $f(q_i^Tx) + f(q_i^T y) = 1$, so $\V=[x,y]^T$ is a solution to our factoring problem. Using the NP-hardness of ReLU-separability by Proposition \ref{prop:reluhard}, the result follows. Note here that $k=2$ is a constant, but for larger $\alpha \in \R^{m \times k}$ with $m$ rows and $k$ columns, we can pad the new entries with zeros to reduce the problem to the aforementioned one, which completes the proof for a fixed $\alpha$.

Now for $k =2 $ and $\alpha$ a variable, we add the following constraints to reduce to the case of $\alpha = [1,1]$, after which the result follows. First, we add $2$ new columns and $1$ new row to $\X$, giving $\X' \in \R^{(d+1) \times (r+s+2)}$. We set
\[ \X' = \begin{bmatrix} \X & \vec{0} & \vec{0}  \\ \vec{0}^T & 1&-1   \\   \end{bmatrix} \]
Where $\X$ is as in the last paragraph, where $\vec{0}$ is a column vector of the appropriate dimensions above. Also, we set $\AA' = [\AA , 1 , 1] \in \R^{r+s+4}$. Let $\V = [x,y]^T$ as before. This ensures that $\alpha_1 f(x_{d+1} ) + \alpha_2 f(y_{d+1} ) = 1$ and $\alpha_1 f(-x_{d+1} ) + \alpha_2 f(- y_{d+1} ) = 1$. As before, we cannot have that both $(x_{d+1} )$ and $(y_{d+1})$ are negative, or that both are positive, as then one of the two constraints would be impossible. WLOG, $(y_{d+1} ) < 0$. Then we have $\alpha_1f(x_{d+1} ) = 1$, which ensures $\alpha_1 >0$, and $\alpha_2 f(- y_{d+1}) = 1$, which ensures $\alpha_2 > 0$.

Now suppose we have a solution to $\V = [x,y]^T$ and $\alpha \in \R^2$ to this new problem with $\X',\AA'$. Then we can set $x' = x/ \alpha_1$ and $y' = y/\alpha_2$, and $\alpha' = [1,1]$, and we argue that we have recovered a solution $[x',y']$ to ReLU separability. Note that $[1,1] f([x',y']^T \X') = \AA'$, since we can always pull a positive diagonal matrix in and out of $f$. Then restricting to the first $r+s$ columns of $\X',\AA'$, we see that $[1,1]f([x',y']^T \X) = \AA$, thus $[x',y']$ are a solution to the neural-net learning problem as in the first paragraph, so as already seen we have that $x',y'$ is a solution to ReLU-separability. Similarly, any solution $x,y$ to ReLU separability can easily be extended to our learning problem by simply using $\V = \begin{bmatrix} x & 1 \\ y & -1 \end{bmatrix}$ and $\alpha = [1,1]$, which completes the proof.
\end{proof}

%% file: SimpleExactGaussianCase.tex
\section{A  Polynomial Time Exact Algorithm for Gaussian Input}\label{sec:polyexact}

In this section, we study an exact algorithm for recovering the weights of a neural network in the realizable setting, i.e., the labels are generated by a neural network when the input is sampled from a Gaussian distribution. We also show that we can use independent and concurrent work of Ge et. al. \cite{ge2018learning} to extend our algorithms to the input being sampled from a symmetric distribution. Our model is similar to non-linear generative models such as those for neural networks and generalized linear models already well-studied in the literature \cite{sedghi2016provable, sedghi2014provable, kakade2011efficient,mondelli2018connection}, but with the addition of the ReLU activation function $f$ and the second layer of weights $\U^*$. In other words, we receive as input i.i.d. Gaussian\footnote{See Remark \ref{remark:multivariate}} input $\X \in \R^{d \times n}$ and the generated output is $\AA=\U^* f(\V^*\X)$, where $\U^* \in \R^{m \times k}$ and $\V^* \in \R^{k \times d}$. For the remainder of the section, we assume that both $\V^*$ and $\U^*$ are rank $k$. Note that this implies that $d \geq k$ and $k \leq m$. In Section \ref{sec:FPT}, however, we show that if we allow for a larger ($(\kappa(\V^*))^{O(k)}$) sample complexity, we can recover $\U^*$ even when it is not full rank.

We note that the generative model considered in \cite{sedghi2014provable} matches our setting, however, it requires the function $f$ to be differentiable and $\V^*$ to be sparse. In contrast, we focus on $f$ being ReLU. The ReLU activation function has gained a lot of popularity recently and is ubiquitous in applications \cite{comon1994independent, hyvarinen1999fast,frieze1996learning, hyvarinen2000independent, arora2012provable, liu2012two,hsu2013learning}. As mentioned in Sedghi et. al. \cite{sedghi2014provable}, if we make no assumptions on $\V^*$, the resulting optimal weight matrix is not identifiable. Here, we make no assumptions on $\U^*$ and $\V^*$ apart from them being full rank and show an algorithm that runs in polynomial time. The main technical contribution is then to recover the optimal $\U^*$ and $\V^*$ exactly, and not just up to $\eps$-error. By solving linear systems at the final step of our algorithms, as opposed to iterative continuous optimization methods, our algorithms terminate after a polynomial number of arithmetic operations.

Formally, suppose there exist fixed rank-$k$ matrices $\U^* \in \R^{m \times k} ,\V^* \in \R^{k \times d}$ such that $\AA = \U^* f(\V^*\X)$, and $\X$ is drawn from an i.i.d. Gaussian distribution. Note that we can assume that each row $\V^*_i$ of $\V^*$ satisfies $\|\V^*_i\|_2 = 1$ by pulling out a diagonal scaling matrix $\D$ with positive entries from $f$, and noting $\U^* f(\D\V^*\X) = (\U^* \D)f(\V^*\X)$. Our algorithm is given as input both $\AA$ and $\X$, and tasked with recovering the underlying generative neural network $\U^* ,\V^*$. In the  context of training neural networks, we consider $\X$ to be the feature vectors and $\AA$ to be the corresponding labels.  Note $\U^*,\V^*$ are oblivious to $\X$, and are fixed prior to the generation of the random matrix $\X$.  In this section we present an algorithm that is polynomial in all parameters, i.e., in the rank $k$, the condition number of $\U^*$ and $\V^*$, denoted by  $\kappa(\U^*), \kappa(\V^*)$ and $n, m , d$. 

Given an approximate solution to $\U^*$, we show that there exists an algorithm that outputs $\U^*, \V^*$ exactly and runs in time polynomial in all parameters.
We begin by giving an altenative algorithm for orthonormal $\V^*$ based on \textit{Independent Component Analysis}. We believe that this perspective on learning neural networks may be useful beyond our results. Next, we will give a general algorithm for exact recovery of $\U^*,\V^*$ which does not require $\V^*$ to be orthonormal. This algorithm is based on the completely different approach of tensor decomposition, yet yields the same polynomial running time for exact recovery in the noiseless case. We now pause for a brief aside on the generalization of our results to the non-identity covariance case.

\begin{remark}\label{remark:multivariate} 
While our results are stated for when the columns of $\X$ are Gaussian with identity covariance, they can naturally be extended to $\X$ with arbitrary non-degenerate (full-rank) covariance $\mathbf{\Sigma}$, by noting that $\X = \mathbf{\Sigma}^{1/2}\X'$ where $\X'$ is i.i.d. Gaussian, and then implicitly replacing $\V^*$ with $\V^* \mathbf{\Sigma}^{1/2}$ so that $f(\V^*\X) = f( (\V^* \Sigma^{1/2}) \X')$, and noting that $\kappa(\V^* \Sigma^{1/2})$ blows up by a $\sqrt{\kappa(\mathbf{\Sigma})}$ factor from $\kappa(\V^*)$. All our remaining results, which do not require $\V^*$ to be orthonormal, hold with the addition of polynomial dependency on $\sqrt{\kappa(\mathbf{\Sigma})}$, by just thinking of $\V^*$ as $\V^* \Sigma^{1/2}$ instead.
\end{remark}
We use the sample covariance as our estimator for the true covariance $\mathbf{\Sigma}$ and have the following guarantee:
\begin{lemma}(Estimating Covariance of $\X$ \cite{vershynin2018high}.)
Let $\X \in \R^{d \times N}$ such that for all $i \in [N]$, $\X_{*,i} \sim \mathcal{N}(0, \mathbf{\Sigma})$. Let $\mathbf{\Sigma}_N = \frac{1}{N} \sum_{i\in[N]} \X_{*,i} \X_{*,i}^T$. With probability at least $1 - 2e^{-\delta}$,
\begin{equation*}
    \|\mathbf{\Sigma} - \mathbf{\Sigma}_n \|_2 \leq c\frac{d +\delta}{N} \|\mathbf{\Sigma}\|_2
\end{equation*}
for a fixed constant $c$.
\end{lemma}
We can then estimate $\Sigma$ using a holdout set of $N  = \Omega(  n^2 \delta^2)$ samples, which suffices to get an accurate estimate of the covariance matrix. We point out that, other than the tensor decomposition algorithm of Section \ref{subsec:generalnoiseless} and the noisy half-space learning routine in Section \ref{sec:1bit}, our algorithms do not even need to estimate the covariance matrix $\mathbf{\Sigma}$ in the multivariate case in order to approximately (or exactly) recover $\U^*,\V^*$. With regards to our tensor decomposition algorithms, while our estimator for the covariance introduces small error in the computation of the Score Function and the resulting tensor decomposition, this can be handled easily in the perturbation analysis of Theorem \ref{thm:anandkumar} (refer to Remark 4 in \cite{janzamin2015beating}).
For our half-space learning algorithm in Section \ref{sec:1bit}, the error caused by estimating $\Sigma$ is negligible, and can be added to the ``advesarial'' error $\BB$ of Theorem \ref{thm:modular} which is already handled. 

In the following warm-up Section \ref{sec:ICA}, where it is assumed that $\V^*$ is orthonormal, we cannot allow $\X$ to have arbitrary covariance, since then $\V^* \mathbf{\Sigma}^{1/2}$ would not be orthonormal. However, for in the more general algorithm which follows in Section \ref{subsec:generalnoiseless}, arbitrary non-degenerate covariance $\mathbf{\Sigma}$ is allowed.

\subsection{An Independent Component Analysis Algorithm for Orthonormal $\V^*$}\label{sec:ICA}

We begin with making the simplifying assumption that the optimal $\V^*$ has orthonormal rows, as a warm-up to our more general algorithm. Note, if $\V^*$ is orthonormal and $\X$ is standard normal, then by $2$-stability of Gaussian random variables,  $\V^*\X$ is a matrix of i.i.d. Gaussian random variables. Since Gaussian random variables are symmetric around the origin, each column of $f(\V^*\X)$ is sparse, has i.i.d entries, and has moments bounded away from Gaussians. Using these facts, we form a connection to the Independent Component Analysis (ICA) problem, and use standard algorithms for ICA to recover an approximation to $\U^*$. 

The ICA problem approximately recovers a subspace $\BB$, given that the algorithm observes samples of the form $y =\BB x + \E$, where $x$ is i.i.d. and drawn from a distribution that has moments bounded away from Gaussians and $\E$ is Gaussian noise. The ICA problem has a rich history of theoretical and applied work \cite{comon1994independent, frieze1996learning, hyvarinen1999fast,hyvarinen2000independent, frieze2004fast, liu2012two, arora2012provable, hsu2013learning}. Intuitively, the goal of ICA is to find a linear transformation of the data such that each of the coordinates or features are as independent as possible. For instance, if the dataset is generated as $y = \BB x$, where $\BB$ is an unknown affine transformation and $x$ has i.i.d. components, with no noise added, then applying $\BB^{-1}$ to $y$ recovers the independent components exactly, as long as $x$ is non-Gaussian. Note, if $x \sim \mathcal{N}(0,\II_m)$, then by rotational invariance of Gaussians, we can only hope to recover $\BB$ up to a rotation and the identity matrix suffices as a solution.  

\begin{definition}(Independent Component Analysis.)
Given $\epsilon >0$ and samples of the form $y_i = \BB x_i + \E_i$, for all $i \in [n]$, such that $\BB \in \R^{m \times m}$ is unknown and full rank, $x_i \in \R^m$ is a vector random variable with independent components and has fourth moments strictly less than that of a Gaussian, the ICA problem is to recover an additive error approximation to $\BB$, i.e., recover a matrix $\widehat{\BB}$ such that $\| \widehat{\BB} - \BB\|_F  \leq  \epsilon$.
\end{definition}

We use the algorithm provided in Arora et. al. \cite{arora2012provable} as a black box for ICA. We note that our input distribution is rectified Gaussian, which differs from the one presented in \cite{arora2012provable}. Observe, our distribution is invariant to permutations and \textit{positive} scaling, is sub-Gaussian, and has moments that are bounded away from Gaussian. The argument in \cite{arora2012provable} extends to our setting, as conveyed to us via personal communication \cite{rongpersonal}. We have the following formal guarantee :

\begin{theorem}(Provable ICA, \cite{arora2012provable} and \cite{rongpersonal}.) \label{theorem:ICA}
Suppose we are given samples of the form $y_i = \BB x_i + \E_i$ for $i=1,2,\dots,n$, where $\BB \in \R^{m \times m}$, the vector $x_i \in \R^m$ has i.i.d. components and has fourth moments strictly bounded away from Gaussian, and $\E_i \in \R^{m}$ is distributed as $\mathcal{N}(0,\mathbb{I}_m)$, there exists an algorithm that with high probability recovers $\widehat{\BB}$ such that $\| \widehat{\BB} - \BB \mathbf{\Pi} \D \|_F \leq \epsilon$, where $\mathbf{\Pi}$ is a permutation matrix and $\D$ is a diagonal matrix such that it is entry-wise positive. Further, the sample complexity is $n = \poly\left( \kappa(\BB), \frac{1}{\epsilon}\right)$ and the running time is $\poly(n,m)$. 
\end{theorem}


We remark that ICA analyses typically require $\BB$ to be a square matrix, and recall  that$\U^*$ is $m \times k$ for $m \geq k$. 
To handle this, we sketch our samples using a dense Gaussian matrix with exactly $k$ columns, and show this sketch is rank preserving. We will denote the resulting matrix by $\T \U^*$.

\begin{Frame}[\textbf{Algorithm \ref{alg:overall_exact_fpt} : ExactNeuralNet}$(\AA, \X)$]
\label{alg:overall_exact_fpt}
\ttx{Input:} Matrices  $\AA \in \R^{d \times n}$ and $\X \in \R^{r \times n}$ such that each entry in $\X \sim \mathcal{N}(0, 1)$. \\
\begin{enumerate}
    \item Let $\T \in \R^{k \times m}$ be a matrix such that for all $i\in[k]$, $j \in [m]$, $\T_{i,j} \sim \mathcal{N}(0,1)$. Let $\T\AA$ be the matrix obtained by applying the sketch to $\AA$.  
    \item Consider the ICA problem where we receive samples of the form $\T\AA = \T \U^* f(\V^*\X)$. 
    \item Run the ICA algorithm, setting $\epsilon = \frac{1}{\poly\left(m,d,k,  \kappa(\U^*)\right)}$, to recover $\widehat{\T\U}$ such that $\|\widehat{\T\U} - \T\U^*\mathbf{\Pi}\D \|_F \leq \frac{1}{\poly\left(m,d,k,  \kappa(\U^*)\right)}$.
    \item Let $\overline{\X}$ be the first $\ell = \poly(d,m,k, \kappa(\U^*),\kappa(\V^*) )$ columns of $\X$, and let $\overline{\AA} = \U^*f(\V^*\overline{\X})$.
  Let $\tau = \frac{1}{\poly(\ell)}$ be a threshold. Then for all $i \in [k], j \in [\ell]$, set
    \begin{equation*}
        \widehat{f(\V\overline{\X})}_{i,j} = \begin{cases}
            0 &\text{if $ \big( ( \widehat{\T\U})^{-1} \T\overline{\AA}\big)_{i,j} \leq \tau$}\\
        \big( ( \widehat{\T\U})^{-1} \T\overline{\AA}\big)_{i,j} &\text{otherwise}\\
        \end{cases}
    \end{equation*}
    \item Let $S_j$ be the sparsity pattern of the vector $\widehat{f(\V \overline{\X})}_{j,*}$. For all $j \in [k]$, and $r \in [k]$, solve the following linear system of equations in the unknowns $x_j^r \in \R^k$. 
    \begin{equation*}
        \begin{array}{ll@{}ll}
        \forall i \in [\ell] \setminus  S_j & \displaystyle    &  (x_j^r \overline{\AA})_i = 0, \\
         & \displaystyle    &  (x_j^r)_r = 1 \\
        \end{array}
    \end{equation*}
    Where $(x_j^r)_r$ is the $r$-th coordinate of $x_j^r$.
    \item Set $w_j$ to be the first vector $x_j^r$ such that a solution exists to the above linear system. 
    \item Let $\W \in \R^{k \times \ell}$ be the matrix where the $i$-th row is given by $w_i  \overline{\AA}$. Flip the signs of the rows of $\W$ so that $\W$ has no strictly negative entries. 
    
    \item For each $i \in [k]$, solve the linear system $(\W_{i,*})_{S_i} = \V_{i,*} \overline{\X}_{S_i}$ for $\V \in \R^{k \times d}$, where the subscript $S_i$ means restricting to the columns of $S_i$. Normalize $\V$ to have unit norm rows. Finally, solve the linear system $\AA = \U f(\V \X)$ for $\U$, using Gaussian Elimination. 

\end{enumerate}
\ttx{Output:}  $\U ,\V$. 
\end{Frame}

\begin{lemma}(Rank Preserving Sketch.)
\label{lem:rank_preserving}
Let $\T \in \R^{k \times m}$ be a matrix such that for all $i\in[k]$, $j \in [m]$, $\T_{i,j} \sim \mathcal{N}(0,1)$.  Let $\U^* \in \R^{m \times k}$ such that rank$(\U^*) = k$ and $m > k$. Then, $\T \U^* \in \R^{k \times k}$ has rank $k$.   Further, with probability at least $1-\delta$,  $\kappa(\T \U^*) \leq ( k^2 m /\delta) \kappa(\U^*)$.
\end{lemma}
\begin{proof}
Let $\M \mathbf{\Sigma} \NN^T $ be the SVD of $\U^*$, such that $\M \in \R^{m \times k}$ and $\mathbf{\Sigma} \NN^T  \in \R^{k \times n}$. Since columns of $\M$ are orthonormal and Gaussians are 
rotationally invariant, $\T\M \in \R^{k \times k}$ is i.i.d. standard normal. Further, 
$\mathbf{\Sigma} \NN^T$ has full row rank and thus has a right inverse, i.e., $\NN \mathbf{\Sigma}^{-1}$.
Then, $\textrm{rank}(\T\U) = \textrm{rank}(\T \M \mathbf{\Sigma} \NN^T) \leq \textrm{rank}(\T \M)$. Further $\T \M = \T \U \mathbf{\Sigma}^{-1}$, and therefore $\textrm{rank}(\T \M) = \textrm{rank}(\T \U \mathbf{\Sigma}^{-1}) \leq \textrm{rank}(\T \U)$. Recall, $\T\M$ is a $m \times k$ matrix of standard Gaussian random variables and has a non-zero determinant with probability $1$.    

Next, $\kappa(\T \U^*) \leq \kappa(\T) \kappa(\U^*)$. Note $\T$ is at least $k+1 \times k$ and by Theorem 3.1 in \cite{rudelson2010non}, with probability $1-\delta$, $\sigma_{\min}(\T) \geq k \delta$. Similarly, by Proposition 2.4 \cite{rudelson2010non}, with probability $1-1/e^{\Omega(1/\delta)}$, $\sigma_{\max}(\T) \leq km/\delta$. Union bounding over the two events, with probability at least $1-1/\textrm{poly}(k)$, $\kappa(\T) \leq \poly(k)$ and thus $\kappa(\T \U^*) \leq \kappa(\U)  k^2 m /\delta$. 
\end{proof}

Algorithmically, we sketch the samples $\T \AA$ such that they are of the form $\T \U^* f(\V^* \X)$. By Lemma \ref{lem:rank_preserving}, $\T \U^*$ is a square matrix and has rank $k$. Since $\V$ is orthonormal, each column of $f(\V^*\X)$ has entries that are i.i.d. $\max\{\mathcal{N}(0,1), 0\}$. Note, the samples $\T \AA$ now fit the ICA framework, the noise $\E =0$, and thus we can approximately recover $\U^*$, without even looking at the matrix $\X$. Here, we set $\epsilon = \frac{1}{\poly\left(m,d,k,  \kappa(\U^*)\right)}$ to get the desired running time. Recall, given the polynomial depedence on $1/\eps$, we cannot recover $\U^*$ exactly. 

\begin{corollary}(Approximate Recovery using ICA.) \label{cor:ICA}
Given $\AA \in \R^{m \times n}, \X \in \R^{d \times n}$, and a sketching matrix $\T \in \R^{k \times m}$ such that $\AA = \U^* f(\V^*\X)$ and for all $i\in[k]$, $j \in [m]$, $\T_{i,j} \sim \mathcal{N}(0,1)$, there exists an algorithm that outputs an estimator to $\widehat{\T\U^*}$ such that $\|\widehat{\T\U} - \T\U^*\mathbf{\Pi}\D \|_F \leq \frac{1}{\poly\left(m,d,k,  \kappa(\U^*)\right)}$, where $\mathbf{\Pi}$ is a permutation matrix and $\D$ is strictly positive diagonal matrix. Further, the running time is $\poly\left(m,d,k,  \kappa(\U^*)\right)$.
\end{corollary}

\paragraph{Exact Recovery:}
By Corollary \ref{cor:ICA}, running ICA on $\T\AA = \T \U^*f(\V^*\X)$,
we recover $\T\U^*$ approximately up to a permutation and positive scaling of the column. Note that we can disregard the permutation by simply assuming $\V$ has been permuted to agree with the $\mathbf{\Pi}$.  Let $\widehat{\T\U}$ be our estimate of $\T\U^*$.  We then restrict our attention to the first $\ell = \poly\left(d,m,k,\kappa(\U^*),\kappa(\V^*)\right)$ columns of $\X$, and call this submatrix $\overline{\X}$, and $\overline{\AA} = \U^* f(\V^* \overline{\X})$. 
We then multiply $\T \overline{\AA}$ by the inverse $(\widehat{\T\U})^{-1}$,  which we show allows us to recover $\D^{-1}f(\V^*\overline{X})$ up to additive $\eps$ error where $\eps$ is at most $O\left(\frac{1}{\poly(d,m,k,\kappa(\U^*),\kappa(\V^*)}\right)$. Since the sketch $\T$ will preserve rank, $\T\U$ will have an inverse, and thus  $(\widehat{\T\U})$ will be invertible (we can always perturbe the entries of our estimate by $1/\poly(n)$ to ensure this). The inverse can then be computed in a polynomial number of arithmetic operations via Gaussian elimination. By a simple thresholding argument, we show that after rounding off the entries below $\tau =1/\poly(\ell)$ in $(\widehat{\T\U})^{-1}\T\overline{\AA}$, we in fact recover the \textit{exact} sign pattern of $f(\V^*\overline{\X})$. 

 Our main insight is now that the only sparse vectors in the row space of $\overline{\AA}$ are precisely the rows (up to positive a scaling) of $f(\V^* \overline{\X})$. Specifically, we show that the only vectors in the row span of $\U^*f(\V^*\overline{\X})$ which have the same sign and sparsity pattern as a row of $f(\V^*\overline{\X})$ are positives scalings of the rows of $f(\V^*\overline{\X})$. Here, by \textit{sparsity pattern}, we mean the subset of entries of a row that are non-zero. Since each row of $f(\V^*\X)$ is non-negative, the sign and sparsity patterns of $f(\V^*\overline{\X})$ together specify where the non-zero entries are (which are therefore strictly positive). 
 
 Now after exact recovery of the sign pattern of $f(\V^*\overline{\X})$, we can set up a linear system to find a vector in the row span of $\AA$ with this sign pattern, thus recovering each row of $f(\V^*\overline{\X})$ exactly. Critically,  we exploit the combinatorial structure of ReLUs together with the fact that linear systems can be solved in a polynomial number of arithmetic operations. This allows for exact recovery of $\U^*$ thereafter. Recall that we assume the rows of $\V^*$ have unit length, which removes ambiguity in the positive scalings used for the rows of $\V*$ (and similarly the columns of $\U^*$).

  We begin by showing that the condition number of $\V^*$ is inversely proportional to the minimum angle between the rows of $\V^*$, if they are interpreted as vectors in $\R^{d}$. This will allow us to put a lower bound on the number of disagreeing sign patterns between rows of $f(\V^*\overline{\X})$ in Lemma \ref{prop:uniquesign}. We will then use these results to prove the uniqueness of the sign and sparsity patterns of the rows of $f(\V^*\overline{\X})$ in Lemma \ref{lem:uniquesign}.

\begin{lemma}\label{prop:condition}
Let $\theta_{\min}\in [0,\pi]$ be the smallest angle between the lines spanned by two rows of the rank $k$ matrix $\V \in \R^{k \times d}$ which unit norm rows, in other words $\theta_{\min} = \min_{i,j} \arccos ( \langle \V_{i,*} , \V_{j,*} \rangle )$ where $\arccos$ takes values in the principle range $[0,\pi]$. Then $\kappa(\V) > \frac{c}{\theta_{\min}}$ for some constant $c$.
\end{lemma}
\begin{proof}
 Let $i,j$ be such that $ \arccos ( |\langle \V_{i,*} , \V_{j,*} \rangle |) = \theta_{\min}$. Let $\V^-$ be the pseudo-inverse of $\V$. Since $\V$ has full row rank, it follows that $\V(\V^-)^T = I_k$, thus $\langle \V_{i,*} , \V^-_{j,*} \rangle =0$ and $\langle \V_{j,*} , \V^-_{j,*} \rangle = 1$. The first fact implies that $\V^-_{j,*}$ is orthonormal to $\V_{i,*}$, and the second that $\cos(\theta(\V_{j,*}, \V^-_{j,*}))= (\|\V^-_{j,*}\|_2)^{-1}$ where $\theta(\V_{j,*}, \V^-_{j,*})$ is the angle between $\V_{j,*}$ and $\V^-_{j,*}$.

Now let $x = \V_{i,*}, y = \V^-_{j,*}/\|\V^-_{j,*}\|, z = \V_{j,*}$. Note that $x,y,z$ are all points on the unit sphere in $r$ dimensions, and since scaling does not effect the angle between two vectors, we have $\theta(x,y) = \theta(\V_{i,*}, \V_{j,*}^-)$. We know $\theta(x,y) = \pi/2$, and $\theta_{\min} = \theta(x,z)$, so the law of cosines gives $\cos(\theta(y,z) ) =  \frac{2 - \|y - z\|_2^2}{2 }$. We have $\|y - z\|_2  = \|(y - x) - (z - x)\|_2 \geq |\sqrt{2} - \|z - x\|_2|$. Again by the law of cosines, we have $\|z - x\|_2^2 = 2 - 2\cos(\theta_{\min})$. Since $\cos(x) \approx 1 - \Theta(x^2)$ for  small $x$ (consider the Taylor expansion), it follows that $\|z - x\|_2 \leq  c' \theta_{\min}$ for some constant $c'$. So $\|y - z\|_2^2 \geq 2 - 2\sqrt{2} \|z - x\|_2 + \|z - x\|_2^2 \geq 2 - c'' \theta_{\min} $ for another constant $c''$. It follows that $$\cos(\theta(y,z) ) \leq \frac{c'' \theta_{\min}}{2}$$
From which we obtain $\|\V_{j,*}^-\|_2 \geq 2/(c''\theta_{\min})$. It follows that $\sigma_{1}(\V^-) \geq \|e_{j,*}^T \V^-\|_2 = \|\V^-_{j,*}\|_2 \geq \frac{2}{c'' \theta_{\min}}$. Since the rows of $\V$ have unit norm, we have $\sigma_{1}(\V) \geq 1$, so $\kappa(\V) = \sigma_{1}(\V) \sigma_{1}(\V^-) \geq   \frac{2}{c'' \theta_{\min}}$
which is the desired result setting $c=\frac{2}{c''}$.
\end{proof}

\begin{lemma} \label{prop:uniquesign}
Fix any matrix $\V \in \R^{k \times d}$ with unit norm rows. Let $\X \in \R^{d \times \ell}$ be an i.i.d. Gaussian matrix for any $\ell \geq  t \poly(  k, \kappa)$, where $\kappa = \kappa(\V)$. For every pair $i,j \in [k]$ with $i \neq j$, with probability $1 - 1/\poly(\ell)$ there are at least $t$ coordinates $p \in [\ell]$ such that $(\V\X)_{i,p} < 0$ and $(\V\X)_{j,p} > 0$.
\end{lemma}
\begin{proof}
 We claim that $\pr{ (\V\X)_{i,p} < 0 , (\V\X)_{j,p} > 0 } = \Omega(1/\kappa)$. To see this, Consider the $2$-dimensional subspace $H$ spanned by $\V_{i,*}$ and $\V_{j,*}$. Let $\theta$ be the angle between $\V_{i,*}$ and $\V_{j,*}$ in the plane $H$. Then the event in question is the event that a random Gaussian vector, when projection onto this plane $H$, lies between two vectors with angle $\theta$ between each other. By the rotational invariance and spherical symmetric of Gaussians (see, e.g. \cite{bryc2012normal}), this probability is $\frac{\theta}{2 \pi}$. 
 Since $\kappa(\V) > \frac{c}{\theta_{\min}}= \Omega(\frac{1}{\theta})$ by Lemma \ref{prop:condition}, it follows that a random gaussian splits $\V_{i,*}$ and $\V_{j,*}$ with probability $\Omega(1/\kappa)$ as desired.
 
Thus on each column $p$ of $f(\V\X)$, $f(\V_{i,*}\X_{*,p}) < 0$ and $f(\V_{j,*}\X_{*,p}) > 0$ with probability at least $\Omega(1/\kappa)$. Using the fact that the entries in separate columns of $\V\X$ are independent, by Chernoff bounds, with probability greater than $1-k^2\exp(\Omega(\ell/\kappa)) > 1 - 1/\poly(\ell)$, after union bounding over all $O(k^2)$ ordered pairs $i,j$, we have that $f(\V_{i,*}\X) < 0$ and $f(\V_{j,*}\X) > 0$ on at least $\Omega(\ell/\kappa) > t$ coordinates.

\end{proof}

\begin{lemma}\label{prop:pdf}
Let $\ZZ_i$ be the $i$-th column of $(\V\X)$, where $\V$ has rank $k$. Then the covariance of the coordinates of $\ZZ_i$ are given by the $k \times k$ posiitve definite covariance matrix $\V \V^T$, and the joint density function is given by:
\[ p(\ZZ_{i,1}, \dots, \ZZ_{i,k}) = \frac{\exp\big( - \frac{1}{2} \ZZ_i^T (\V\V^T)^{-1} \ZZ_i \big)}{\sqrt{(2 \pi)^k \det(\V\V^T) }}  \]
In particular, the joint probability density of any subset of entries of $\V\X$ is smooth and everywhere non-zero.
\end{lemma}

\begin{proof}
    Since $\ZZ_i = \V(\X^T_i)^T$, where $\X^T_i$ are i.i.d. normal random variables, t is well known that the covariance is given by $\V\V^T$ \cite{gut2009intermediate}, which is positive definite since $\V$ has full row rank. These are sufficient conditions (\cite{UIUCLecture}) for the pdf to be given in the form as stated in the Proposition. Since distinct columns of $\V\X$ are statistically independent (as they are generated by separate columns of $\X$), the last statement of the proposition follows.   
    
\end{proof}

The following Lemma demonstrates that the the only vectors in the row span of $f(\V^*\X)$ with the same sign and sparsity pattern as $f(\V^*\X)_{i,*}$, for any given row $i$, are positive scalings of $f(\V^*\X)_{i,*}$. Recall that a sparsity pattern $S \subseteq [n]$ of a vector $y \in \R^n$ is just set of coordinates $i \in S$ such that $y_i > 0$. 

\begin{lemma}\label{lem:uniquesign}Let $\X \in \R^{d \times \ell}$ be an i.i.d. Gaussian matrix for any $\ell > t \poly(  k, \kappa(\V^*))$.
Let $S$ be the sparsity pattern of a fixed row $f(\V^*\X)_{i,*}$, and let $\emptyset \subsetneq S' \subseteq S$. Then w.h.p. (in $t$), the only vectors in the row span of $f(\V^*\X)$ with sparsity pattern $S'$, if any exist, are non-zero scalar multiples of $f(\V^*\X)_{i,*}$.
\end{lemma}\begin{proof}
	Suppose  $\ZZ = w f(\V^*\X)$ had sparsity pattern $S'$  and was not a scaling of $f(\V^* \X)_{i,*}$. Then $w$ is not $1$-sparse, since otherwise it would be a scaling of a another row of $f(\V^*\X)$, and by Proposition \ref{prop:uniquesign} no row's sparsity pattern is contained within any other row's sparsity pattern. Let $\W$ be $f(\V^* \X)$ restricted to the rows corresponding to the non-zero coordinates in $w$, and write $\ZZ = w \W$ (where now $w$ has also been restricted to the appropriate coordinates). Since $\W$ has at least $2$ rows, and since the sparsity pattern of $w\W$ is contained within the sparsity pattern of $f(\V^* \X)_{i,*}$, by Proposition \ref{prop:uniquesign}, taking $t = 10k^2$, we know that there are at least $10k^2$ non-zero columns of $\W$ for which $w\W$ is $0$, so let $\W'$ be the submatrix of all such columns.

	Now for each row $\W_i'$ of $\W'$ with less than $k$ non-zero entries, remove this row $\W_i'$ and also remove all columns of $\W'$ where $\W_i'$ was non-zero. Continue to do this removal iteratively until we obtain a new matrix $\W''$ where now every row has at least $k$ non-zero entries. 
	Observe that the resulting matrix $\W''$ has at least $9k^2$ columns. If there are no rows left, then since we only removed $k$ columns for every row removed, this means there were at least $9k^2$ columns of $\W'$ which contained only zeros, which is a contradiction since by construction the columns of $\W'$ were non-zero to begin with. So, let $k' \leq k$ be the number of rows remaining in $\W''$. Note that since the rows we removed were zero on the columns remaining in $\W''$, there must still be a vector $w'$, which in particular is $w$ restricted to the rows of $\W''$, which has no zero-valued entries and such that $w' \W'' = 0$. 
	
	Now observe once we obtain this matrix $\W''$, note that we have only conditioned on the sparsity pattern of the entries of $\W''$ (over the randomness of the Gaussians $\X$), but we have not conditioned on the values of the non-zero entries of $\W''$.  Note that this conditioning does not change the continuity of the joint distributions of the columns of $\W''$, since this conditioning is simply restricting the columns to the non-zero intersection of half spaces which define this sign pattern. Since the joint density function of the columns of $\V\X$ is non-zero on all of $\R^k$ by Lemma \ref{prop:pdf}, it follows that, after conditioning, any open set in this intersection of half spaces which defines the sparsity pattern of $\W''$ has non-zero probability measure with respects to the joint density function.

 Given this, the argument now proceeds as in Lemma \ref{prop:random}. Since each row of $\W''$ has at least $k$ non-zero entries, we can find a square 
	matrix $\W^\dagger \in \R^{k' \times k'}$ obtained by a taking a subset of $k' < 9k^2$ 
	columns of $\W''$ and permuting them such that the diagonal of $\W^\dagger$ has a non-zero 
	sign pattern. After conditioning on the sign pattern so that the diagonal is
	non-zero, the determinant $\det(\W^\dagger)$ of $\W^\dagger$ is a non-zero polynomial in $s$ 
	random variables with $k' \leq s \leq (k')^2$. By Lemma \ref{prop:pdf}, the 
	joint density function of these $s$ variables is absolutely continuous and everywhere
	non-zero on the domain. Here the domain $\Omega$ is the intersection of half spaces 
	given by the sign pattern conditioning.  
	
	Since $\Omega$ is non-empty, it has
	unbounded Lebesgue measure in $\R^s$. Since $\det(\W^\dagger)$ is a non-zero polynomial in 
	$s$ real variables, it is well known that $\det(\W^\dagger)$ cannot vanish on any non-empty
	open set in $\R^s$ (see e.g. Theorem 2.6 of \cite{ConradNotes}, and note the
	subsequent remark on replacing $\C^s$ with $\R^s)$. It follows that the set of zeros 
	of $\det(\W^\dagger)$ contain no open set of $\R^s$, and thus has Lesbegue measure $0$ in $\Omega$. Integrating the joint pdf of the $s$ random variables over this subset of measure $0$, we conclude that the probability that the realization of the random variables is in this 
	set is $0$.  So the matrix $\W''$ has rank $k'$, and so 
	$w' \W'' = 0$ is impossible, a contradiction. It follows that $\ZZ$ is a scaling of a row of $f(\V^* \X)$ as needed.
\end{proof}

We will now need the following perturbation bounds for the pseudo-inverse of matrices.
\begin{proposition}[Theorem 1.1 \cite{MENG2010956}]\label{prop:psuedoinverse}
Let $\BB^\dagger$ denote the Moore–Penrose Pseudo-inverse of $\BB$, and let $\|\BB\|_2$ denote the operator norm of $\BB$. Then for any $\E$ we have 
\[  \|(\BB + \E)^\dagger - \BB^\dagger \|_F \leq \sqrt{2} \max \left\{\|\BB^\dagger\|_2^2, \|(\BB+\E)^\dagger\|_2^2 \right\}   \|\E\|_F \]  
\end{proposition}

We prove the following corollary which will be useful to us.

\begin{corollary}\label{cor:psuedobound}
For any $\BB,\E$ and $\frac{1}{4} \geq \eps > 0$ with $\|\BB \|_2 \geq 1$, $\|\E\|_F \leq \frac{\eps}{\kappa^2}$ and where $\kappa = \kappa(\BB)$ is the condition number of $\BB$. Then we have 
\[  \|(\BB + \E)^\dagger - \BB^\dagger \|_F \leq O(\eps)\]  
and moreover, if $\BB$ has full column rank, then
\[     \| (\BB + \E)^\dagger \BB - \mathbb{I} \|_F \leq O(\|\BB\|_2\eps)\]
\end{corollary}
\begin{proof}
  We have $\|(\BB + \E)^\dagger - \BB^\dagger \|_F  \leq \max \left\{\|\BB^\dagger\|_2^2, \|(\BB+\E)^\dagger\|_2^2 \right\}  \frac{O(\eps)}{\kappa^2}$ by applying Proposition \ref{prop:psuedoinverse}. In the first case, this is at most $ \frac{1}{\sigma^2_{\min}(\BB) } \frac{O(\eps)}{\kappa^2} = O(\eps)$ as stated. Here we used the fact that $\|\BB\|_2 = \sigma_{\max}(\BB) \geq 1$, so $1/ \sigma_{\min}(\BB) \leq \kappa$.   In the second case of the $\max$, we have $\|(\BB + \E)^\dagger - \BB^\dagger \|_F \leq \|(\BB+\E)^\dagger\|_2^2  \frac{O(\eps)}{\kappa^2} = \sigma_{\min}^{-2}(\BB + \E) \frac{O(\eps)}{\kappa^2} $. By the Courant-Fisher theorem \footnote{See \url{https://en.wikipedia.org/wiki/Min-max_theorem} }, using that $\|\E\|_2 \leq \|\E\|_F \leq 1/(4\kappa)$, we have
  \[    \sigma_{\min}(\BB + \E) \geq \inf_{x : \|x\|_2=1} \| x(\BB + \E) \|_2 \geq\inf_{x : \|x\|_2=1} \big|\;  \|x \BB\|_2 -  \|x \E\|_2 \; \big| \]
  \[ \geq \sigma_{\min}(\BB) - \frac{1}{4 \kappa}  \geq \sigma_{\min}(\BB)/2  \geq  1/(2\kappa) \]
  where the minimum is taken over vectors $x$ with the appropriate dimensions.  Thus in both cases, we have $\|(\BB + \E)^\dagger - \BB^\dagger \|_F \leq O(\eps)$, so
  \[  \| (\BB + \E)^\dagger \BB - \mathbb{I} \|_F = \|( (\BB + \E)^\dagger - \BB^\dagger)\BB \|_F\leq\|\BB\|_2 O(\eps) \]
\end{proof}

\noindent
We now are ready to complete the proof of the correctness of Algorithm \ref{alg:overall_exact_fpt} 

\begin{theorem}\label{thm:recoverVgivenU}(Exact Recovery for Orthonormal $\V^*$.)
Given $\AA = \U^* f(\V^*\X)$, for rank $k$-matrices $\U^* \in \R^{m \times k}, \V^* \in \R^{k \times d}$ where $\V^*$ is orthonormal and $\X \in \R^{d \times n}$ which is i.i.d. Gaussian with $n = \poly ( d,k,m,\kappa(\U^*),\kappa(\V^*) )$, there is a $\poly(n)$-time algorithm which recovers $\U^*,\V^*$ exactly with probability $1-\frac{1}{\poly(d,m,k)}$.
\end{theorem}

\begin{proof}
By Corollary \ref{cor:ICA}, after sketching $\AA$ by a Gaussian matrix $\T \in \R^{k \times m}$ and running ICA on $\T\AA$ in $\poly(d,m,k,\kappa(\U^*))$ time, we recover  $\widehat{\T\U^*}$ such that $\|\widehat{\T\U} - \T\U^*\mathbf{\Pi}\D \|_F \leq \frac{1}{\poly( d,k,m,\kappa(\U^*),\kappa(\V^*))}$ for a sufficiently high constant-degree polynomial, such that $\mathbf{\Pi}$ is a permutation matrix and $\D$ is strictly positive diagonal matrix. We can disregard $\mathbf{\Pi}$ by assuming the rows of $\V^*$ have also been permuted by $\mathbf{\Pi}$, and we can disregard $\D$ by pulling this scaling into $\V^*$ (which can be done since it is a positive scaling). Thus $\|\widehat{\T\U} - \T\U^* \|_F \leq \frac{1}{\poly(d,k,m,\kappa(\U^*),\kappa(\V^*))}$

Observe now that we can assume that $1 \leq \| \T\U^* \|_2 \leq 2$ by guessing a scaling factor $c$ to apply to $\AA$ before running ICA. To guess this scaling factor, we can find the largest column (in $L_2$) $y$ of $\T \AA$, and note that $y = (\T \U^*)f(\V^*\X_{*,j})$ for some $j$. 
Since $\|f(\V^*\X_{*,j})\|_2 \leq O(\sqrt{\log(n)})d$ with high probability for all $j \in [n]$ (using the Gaussian tails of $\X$), it follows that $\|y\|_2 \leq \sigma_{\max}(\T\U^*)  O(\sqrt{\log(n)})d$. Since with w.h.p there is at least one column of $f(\V^*\X)$ with norm at least $1/\poly(n)$, it follows that $\|y\|_2 \geq \sigma_{\min}(\T\U^*) /\poly(n) \geq \frac{\sigma_{\max}(\T\U^*)}{\poly(n,\kappa)}$. Thus one can make $\log\big( \poly(n,\kappa,d)  \big) = O(\log(n))$ guesses in geometrically increasing powers of $2$ between $\|y\|_2/ O(\sqrt{\log(n)})d$ and $\|y\|_2 \poly(n,\kappa)$ to find a guess such that $\|c\T\U^*\|_2  \in (1,2)$ as desired. This will allow us to use Corollary \ref{cor:psuedobound} in the following paragraph.

Now let $\widehat{\T\U}^\dagger$  be the pseduo-inverse of $\widehat{\T\U}$, and let $\overline{\AA} = \U^*f(\V^*\overline{\X})$ where $\overline{\X}$ is the first $\poly(d,k,m,\kappa(\U^*),\kappa(\V^*))$ columns of $\X$.   We now claim that the sign pattern of $(\widehat{\T\U})^{\dagger} \T \overline{\AA} = \widehat{\T\U}^{\dagger}\T \U^*f(\V^*\overline{\X})$ is exactly equal to that of $f(\V^*\overline{\X})$ after rounding all entries of with value less than $1/\poly(\ell)$ to $0$. Note that since $\T\U^*$ is full rank, it has an inverse (which is given by the pseudoinverse $(\T\U^*)^\dagger$. 
  Let $\ZZ$ be the resulting matrix after rounding performing this rounding to $\widehat{\T\U}^{\dagger}\T\AA'$. We now apply Corollary \ref{cor:psuedobound}, with $\T\U^* = \BB$ and $\widehat{\T\U} = \BB + \E$. Since we guesses $\sigma_{\max}(\T\U^*)$ up to a factor of $2$ and normalized $\widehat{\T\U}$ by it, it follows that the entries of the diagonal matrix $\D$ are all at most $2$ and at least $1/(2 \kappa(\T\U^*))$, and then using the fact that $\|f(\V^*\overline{\X})\|_F < \|\V^*\overline{\X}\|_F \leq\sqrt{\ell} \|V^*\|_F \leq \sqrt{\ell k}$ w.h.p. in $\ell$ (using well-known upper bounds on the spectral norm of a rectangular Gaussian matrix, see e.g. Corollary 5.35 if \cite{vershynin2010introduction}) we obtain 
  \begin{equation*}
      \begin{split}
          \|\ZZ - \D f(\V^* \overline{\X})\|_F &  = \|\big(\widehat{\T\U}^\dagger(\T\U^*) - \mathbb{I}\big)\D f(\V^*\X')\|_F \\
          & \leq \frac{1}{\poly\left(d,k,m,\kappa(\U^*),\kappa(\V^*)\right)}
      \end{split}
  \end{equation*}
  
  Note that algorithmically, instead of computing the inverse  $\widehat{\T\U}^\dagger$, we can first randomly perturb $\widehat{\T\U}$ by an entry-wise additive $1/\poly(n)$ to ensure it is full rank, and then compute the true inverse, which can be done via Gaussian elimination in polynomially many arithmetic operations. By the same perturbational bounds, our results do not change when using the $1/\poly(n)$ perturbed inverse, as opposed to the original pseudo-inverse. 
  
 Now since the positive entries of $\D f(\V^*\overline{\X})$ have normal Gaussian marginals, and $\D$ is a diagonal matrix which is entry-wise at most $2$ and at least $1/(2\kappa(\T\U^*))$, the probability that any non-zero entry of $f(\V^* \overline{\X})$ is less than $1/\poly(\ell)$ is at most $2\kappa(\T\U^*)/\poly(\ell)$, and we can then union bound over $\poly(d,k,m,\kappa)$ such entries in $\overline{\X}$. Note that by Lemma \ref{lem:rank_preserving}, $\kappa(\T\U^*) < \poly(k,d,m)\kappa(\U^*)$ w.h.p. in $k,d,m$, so $\poly(\ell) >> \kappa(\T\U^*)$. Conditioned on this, with probability $1 - 1/\poly(d,m,k,\kappa)$ for sufficiently large $\ell = \poly(d,k,m,\kappa)$, every strictly positive
entry of $\D f(\V^* \overline{\X})$, and therefore of $f(\V^*\overline{\X})$, is non-zero in $\ZZ$, and moreover, and every other entry will be $0$ in $\ZZ$, which completes the claim that the sign and sparsity patterns of the two matrices are equal.

Given this, for each $i \in [k]$ we can then solve a linear system to find a vector $w_j$ such that $(w_j \overline{\AA})_p = 0$ for all $p$ not in the sparsity pattern of $\ZZ_{i,*}$. In other words, the sparsity pattern of $(w_j \overline{\AA})$ must be contained in the sparsity pattern of $\ZZ_{i,*}$, which is the sparsity pattern of $f(\V^*\overline{\X})_{i,*}$ be the prior argument. 
By Lemma \ref{lem:uniquesign}, the only vector in the row span of $\overline{\AA}$ (which is the same as the row span of $f(\V^*\overline{\X})$ since $\U^*$ is full rank) which has a non-zero sparsity pattern contained in that of $f(\V^*\overline{\X})_{i,*}$ must be a non-zero scaling of $f(\V^* \overline{\X})_{i,*}$. It follows that there is a unique $w_j$, up to a scaling, such that $w_j \overline{\AA}$ is zero outside of the sparsity pattern of $f(\V^*\overline{\X})_{i,*}$. Since at least one of the entries $r$ of $w_j$ is non-zero, there exists some scaling such that $w_j \AA$ is  zero outside of the sparsity pattern of $f(\V^*\overline{\X})_{i,*}$ \textit{and} $(w_j)_r = 1$ (where $(w_j)_r)$ is the $r$-th coordinate of $w_j$). Since the first constraint is satisfied uniquely up to a scaling, it follows that there will be a unique solution $w_j^r$ to at least one of the $r \in [k]$ linear systems in Step 5 of Algorithm \ref{alg:overall_exact_fpt}, which will therefore be optained by the linear system. This vector $w_j$ we obtain from Steps $5$ and $6$ of Algorithm \ref{alg:overall_exact_fpt} will therefore be such that $w_j \overline{\AA}$ is a non-zero scaling of $f(\V^*\overline{\X})_{i,*}$. 

Then in Step $7$ of Algorithm \ref{alg:overall_exact_fpt}, we construct the matrix $\W$, and flip the signs appropriately so that each row of $\W$ is a strictly positive scaling of a row of $f(\V^*\overline{\X})$. 
We then solve the linear system
$(\W_{i,*})_{S_i} = \V_{i,*} \overline{\X}_{S_i}$ for the unknowns $\V$, which can be done with a polynomial number of arithmetic operations via Gaussian elimination. Recall here that $S_i$ is the set of coordinates where $\W_{i,*}$, and therefore $f(\V^*_{i,*} \overline{\X})$, is non-zero. Since at least $1/3$
of the signs in a given row $i$ will be positive with probability 
$1-2^{-\Omega(\ell)}$ by Chernoff bounds, restricting to this subset $S_i$ of
columns gives the equation $\W_{i,*} = \V^*_{i,*} \overline{\X}_{S_i}$. Conditioned on $S_i$ having at least $d$ columns, we have that 
$\overline{\X}_{S_i}$ is full rank almost surely, since it is a matrix of Gaussians conditioned on the fact that every column lies in a fixed halfspace. To see this, apply induction on the columns of $\overline{\X}_{S_i'}$, and note at every step $i<d$, the Lesbegue measure of the span of the first $i$ columns is $0$ in this halfspace, and thus the $i+1$ column will not be contained in it almost surely. It follows that there is a unique solution $\V_{i,*}$ for each row $i$, which 
must therefore be the corresponding row of $\V^*$ (we normalize the rows of $\V_{i,*}$ to have unit norm so that they are precisely the same). So we recover $\V^*$ 
exactly via these linear systems. Finally, we can solve the linear system 
$\AA = \U f(\V^*\X)$ for the variables $\U$ to recover $\U^*$ exactly in strongly
polynomial time. Note that this linear system has a unique solution, since $f(\V^*\X)$ is full rank w.h.p. by Lemma \ref{prop:random}, which completes the proof.
\end{proof}

%% file: ExactGaussianCase.tex
\subsection{General Algorithm}\label{subsec:generalnoiseless}

We now show how to generalize the algorithm from the previous sub-section to handle non-orthonormal $\V^*$. Observe that when $\V^*$ is no longer orthonormal, the entries within a column of $\V^*\X$ are no longer independent. Moreover, due to the presence of the non-linear function $f(\cdot)$, no linear transformation will exist which can make the samples (i.e. columns of $f(\V^*\X)$) independent entry-wise. 
While the entries do still have Gaussian marginals, they will have the non-trivial covariance matrix $\V^* (\V^*)^T \neq \mathbb{I}_k$. Thus it is no longer possible to utilize previously developed techniques from independent component analysis to recover good approximations to $\U^*$.  This necessitates a new approach. 

Our starting point is the generative model considered by Janzamin et. al.  \cite{janzamin2015beating}, which matches our setting, i.e. $\AA= \U^* f(\V^* \X)$. The main idea behind this algorithm is to construct a tensor that is a function of both $\AA, \X$ and then run a tensor decomposition algorithm to recover the low-rank components of the resulting tensor. While computing a tensor decomposition is NP-hard in general \cite{hillar2013most}, there is a plethora of work on special cases, where computing such decompositions is tractable \cite{bhaskara2014smoothed,song2016sublinear,wang2016online,goyal2014fourier, ge2015decomposing, barak2016noisy}.  
Tensor decomposition algorithms have recently become an invaluable algorithmic primitive and found a tremendous number of applications in statistical and machine learning tasks \cite{janzamin2015beating,janzamin2014score,ge2017learning,anandkumar2014tensor, barak2015dictionary}.

A key step is to construct a non-linear transform of the input by utilizing knowledge about the underlying pdf for the distribution of $\X$, which we denote by $p(x)$. The non-linear function considered is the so called Score Function, defined in \cite{janzamin2014score}, which is the normalized $m$-th order derivative of the input probability distribution function $p(x)$.  

\begin{definition}(Score Function.)
\label{def:score}
Given a random vector $x \in \R^{d}$ such that $p(x)$ describes the corresponding probability density function, the $m$-th order score function $\mathcal{S}_m(x) \in \otimes^{m} \R^{d}$ is defined as 
\begin{equation*}
    \mathcal{S}_m(x) = (-1)^m \frac{\nabla^{(m)}_x p(x)}{p(x)} 
\end{equation*}
\end{definition}

The tensor that Janzamin et. al. \cite{janzamin2014score} considers is the cross moment tensor between $\AA$ and $\mathcal{S}_3(\X)$. This encodes the correlation between the output and the third order score function. Intuitively, working with higher order tensors is necessary since matrix decompositions are only identifiable up to orthogonal components, whereas tensor have identifiable non-orthogonal components, and we are specifically interested in recovering approximations for non-orthonormal $\V^*$. Computing the score function for an arbitrary distribution can be computationally challenging. However, as mentioned in Janzamin et. al. \cite{janzamin2014score}, 
we can use orthogonal polynomials that help us compute the closed form for the score function $\mathcal{S}_{(m)}(x)$, in the special case when $x \sim \mathcal{N}(0,\II)$. 

\begin{definition}(Hermite Polynomials.)
If the input is drawn from the multi-variate Gaussian distribution, i.e. $x \sim \mathcal{N}(0, \II)$, then $\mathcal{S}_{(m)}(x) = \mathcal{H}_m(x)$, where $H_m(x) = \frac{(-1)^m \nabla^{(m)}_x p(x)}{p(x)}$ and $p(x) = \frac{1}{(\sqrt{2\pi})^d} e^{-\frac{\|x\|^2_2}{2}}$. 
\end{definition}

Since we know a closed form for the $m$-th order Hermite polynomial, the tensor $\mathcal{S}_{(m)}$ can be computed efficiently. The critical structural result in the algorithm of \cite{janzamin2015beating} is to show that in expectation, the cross moment of the output and the score function actually forms a rank-$k$ tensor, where the rank-$1$ components capture the rows of $\V^*$. Formally, 

\begin{lemma}(Generalized Stein's Lemma \cite{janzamin2015beating}.)
\label{lem:stein}
Let $\AA, \X$ be input matrices such that $\AA = \U^*f(\V^*\X)$, where $f$ is a non-linear, thrice differentiable activation function. Let $\mathcal{S}_3(x)$ be the $3$-rd order score function from Definition \ref{def:score}. Then,
\begin{equation*}
    \widetilde{\T} = \expecf{}{\sum^{n}_{i=1} \AA_{*,i} \otimes \mathcal{S}_3(\X_{*,i})} = \sum^{k}_{j=1} \expecf{x}{f'''(\V^* x)} \U^*_{*,j} \otimes \V^*_{j,*} \otimes \V^*_{j,*} \otimes \V^*_{j,*}
\end{equation*}
where $f'''$ is the third derivative of the activation function and $x \sim p(x)$.
\end{lemma}

Note, $\widetilde{\T}$ is a $4$-th order tensor and can be constructed from the input $\AA$ and $\X$. The first mode of $\widetilde{\T}$ can be contracted by multiplying it with a random vector $\theta$, therefore,
\begin{equation*}
    \expecf{}{\sum^{n}_{i=1} \AA_{*,i} \otimes \mathcal{S}_3(\X_{*,i})} = \sum^{k}_{j=1} \lambda_j \V^*_{j,*} \otimes \V^*_{j,*} \otimes \V^*_{j,*}
\end{equation*}
where $\lambda_j = \expecf{x}{f'''(\V^* x)} \langle \U^*_{*,j}, \theta\rangle$. Therefore, if we could recover the low-rank components of $\widetilde{\T}$ we would be obtain a approximate solution to $\V^*$. The main theorem in \cite{janzamin2015beating} states that under a set of conditions listed below, there exists a polynomial time algorithm that recovers an additive error approximation to $\V^*$. Formally, 

\begin{theorem}(Approximate recovery \cite{janzamin2015beating})
\label{thm:anandkumar}
Let $\AA \in \R^{m \times n}$, $\X\in \R^{d \times n}$ be inputs such that $\AA = \U^* f(\V^* \X) + \eta$, where $f$ is a non-linear thrice differentiable activation function, $\U^* \in \R^{m \times k}$ has full column rank, $\V^* \in \R^{k \times d}$ has full row rank, for all $i \in [n]$, $\X_{*,i} \sim \mathcal{N}(0, \II)$ and $\eta$ is mean zero sub-Gaussian noise with variance $\sigma_{\textrm{noise}}$. Then, there exists an algorithm that recovers $\widehat{\V}$ such that $\|\widehat{\V} - \D \mathbf{\Pi} \V^* \|_F \leq \epsilon$, where $\D$ is a diagonal $\pm 1$ matrix and $\mathbf{\Pi}$ is a permutation matrix. Further, the algorithm runs in time $$\poly\left(m,d,k,\frac{1}{\epsilon}, \expecf{}{\|\M_3(x)\M_3(x)^T\|_2} ,\expecf{}{\| \mathcal{S}_2(x)\mathcal{S}_2(x)^T\|_2}, \frac{1}{\lambda_{\min}},\lambda_{\max},\frac{\widetilde{\lambda}_{\max}}{\widetilde{\lambda}_{\min}},\kappa(\V^*), \sigma_{\textrm{noise}}\right)$$
where $\mathcal{S}_3$ is the $3$-rd order score function, $\M_3(x) \R^{d \times d^2}$ is the matricization of $\mathcal{S}_3$, $\lambda_j$ $ = \expecf{x}{f'''(\V^* x)}$ $\langle \U^*_{*,j}, \theta\rangle$, $\widetilde{\lambda}_j = \expecf{x}{f''(\V^* x)}$ $\langle \U^*_{*,j}, \theta\rangle$ ,$\kappa(\V^*)$ is the condition number, $\sigma_{\textrm{noise}}$ is the variance of $\eta$ and. Note, in the case where $\X_{*,i} \sim \mathcal{N}(0, \II)$, $\expecf{}{\|\M_3(x)\M_3(x)^T\|_2} = O(d^3)$ and $\expecf{}{\| \mathcal{S}_2(x)\mathcal{S}_2(x)^T\|_2}= O(d^2)$.
\end{theorem}

\begin{remark}\label{remark:anandkumar}
We only use the  Whitening, Tensor Decomposition and Unwhitening steps from Janzamin et. al. \cite{janzamin2015beating}, and therefore the sample complexity and running time only depends on Lemma 9 and Lemma 10 in \cite{janzamin2015beating}. 
\end{remark}
However, there are many technical challenges in extending the aforementioned result to our setting. We begin with using the estimator from Theorem \ref{thm:anandkumar} in the setting where the noise, $\eta$, is $0$. The first technical challenge is the above theorem requires the activation function $f$ to be thrice diffrentiable, however ReLU is not. To get around this, we use a result from approximation theory to show that ReLU can be well approximated every where with a low-degree polynomial. 

\begin{lemma}(Approximating ReLU \cite{goel2017learning}.)
\label{lem:approx_relu}
Let $f(x) = \max(0,x)$ be the ReLU function. Then, there exists a polynomial $p(x)$ such that  $$\sup_{x\in [-1, 1]}|f(x) - p(x)| \leq \eta$$
and $\textrm{deg}(p) = O(\frac{1}{\eta})$ and $p([-1, 1]) \subseteq [0,1]$.
\end{lemma}

This polynomial is at least thrice differentiable and can be easily extended to the domain we care about using simple transformations. We assume that the samples we observe are of the form $\U^* p(\V^*\X)$ corrupted by small adversarial error.  Formally, the label matrix $\AA$ can be viewed as being generated via $\AA = \U^*p(\V^*\X) + \ZZ$, where $\ZZ = \U^*\big( f(\V^*\X) - p(\V^*\X)\big)$.
  We note that we only use the approximation as an analysis technique and show that we can get an approximate solution to $\V^*$. First, we make a brief remark regarding the normalization of the entries in $\AA$.

\begin{remark}\label{remark:normalize}
Observe in both the noiseless and noisy cases, the latter being where $\AA = \U^*f(\V^*\X) + \E$ where $\E$ is i.i.d. mean $0$ with variance $\sigma^2$, that by scaling $\AA$ by $1/\|\AA_{*,\max}\|_2$, where $\|\AA_{*,\max}\|_2$ is the largest column norm of $\AA$, we can ensure that the resulting $\U^*$  has $\|\U^*\|_2 < m \max\{1,\sigma \}  \kappa(\U^*)$, where $\sigma^2$ is the variance of the noise $\E$ (in the noisy case). To see why this is true, suppose this were not the case. Observe that w.h.p. at least half of the columns $\U^*f(\V^*\X)$ which will have norm at least $\omega(1) \sigma^{-1}_{\min}(\U^*)$ (since w.h.p. half the columns of $f(\V^*\X)$ have norm $\omega(1)$), thus if $\|\U^*\|_2 > m \max\{1,\sigma \}\kappa(\U^*)$ after normalization, then then at least half of the normalized columns of $\U f(\V^*\X)$ will have norm $\omega(m \max\{1,\sigma \})$. By Markov inequality and a Chernoff bound, strictly less than $1/4$ of the columns of the original $\E$ can have norm $\omega(m \sigma)$ w.h.p., and since the normalized $\E$ is strictly smaller, by triangle inequality there will be a column of $\AA = \U^*f(\V^*\X) +\E$ after normalization with larger than unit norm, a contradiction. Thus we can assume this normalization, giving $\eta<< \frac{1}{\|\U^*\|_2}$ for sufficiently small $\eta = O(\frac{1}{\poly(n,d,m,\kappa(\U^*),\kappa(\V^*),\sigma)})$. 
\end{remark}

We now set $\eta$ in Lemma \ref{lem:approx_relu} to be  $ \frac{1}{\poly(n,d,m,\kappa(\U^*), \kappa(\V^*), \sigma)}$.  By the operator norm bound of Lemma \ref{prop:gaussiancondition}, we know that $\|\V^* \X\|_F = O(\sqrt{nk}),$ w.h.p., so $\|\ZZ\|_F = O(\|\U^*\|_2 \sqrt{nk} \eta) = O(\frac{1}{\poly(n)})$ as needed.  We again construct the same tensor, $\widetilde{\T} = \expecf{}{\sum^{n}_{i=1} \AA_{*,i} \otimes \mathcal{S}_3(\X_{*,i})}$. Our analysis technique is now as follows. We add a light $\mathcal{N}(0,1)$ random matrix to our input $\AA$, and argue that the variation distance between the distribution over inputs $\AA$ (for a fixed $\X$), between the case of $\AA$ using $f$ and $\AA$ using the polynomial $p$ as a non-linear activation, is at most $1/\poly(n)$. As a result, the input using ReLUs is statistically indistinguishable in variation distance from samples generated using the polynomial approximation to the ReLU function. Thus, any algorithm that succeeds on such a polynomial approximation must also succeed on the ReLU. 
Therefore, the algorithm from Theorem \ref{thm:anandkumar} still holds for approximate recovery using ReLUs. Formally, 

\begin{lemma}
The variational distance between $n$ samples of the form $\AA = \U^* f(\V^*\X) + \G$, where the columns of $\G$ are $\mathcal{N}(0,\mathbb{I}_d)$ and $\X$ is fixed, and $\AA' = \U^* p(\V^* \X) + \G + \ZZ$  where $\|\ZZ\|_F = \frac{1}{\poly(n)}$ is at  most $\frac{1}{\poly(n)}$.
\end{lemma}
\begin{proof}
Given two independent Gaussian $\mathcal{N}(\mu_1, \mathbb{I}) , \mathcal{N}(\mu_2, \mathbb{I})$, a standard result in probability theory is that their variations distance is $\Theta(\|\mu_1 - \mu_2\|_2)$ \cite{dasgupta2008asymptotic}. Thus the variation distance between the $i$-th column of $\AA$ and $\AA'$ is $O(\|\ZZ_{*,i}\|_2)$. Since the columns of the input are independent, the overall distribution is a product distribution so the variation distance adds. Thus the total variation distance is at most $O(\|\ZZ_{i,*}\|_F^2) = \frac{1}{\poly(n)}$ as needed.
\end{proof}

It follows from the above lemma that the algorithm corresponding to Theorem \ref{thm:anandkumar} cannot distinguish  between receiving samples from the ReLU distribution with artificially added Gaussian noise or the samples from the polynomial approximation with small adversarial noise. Therefore, the algorithm recovers an approximation to the underlying weight matrix $\V^*$ in polynomial time. Formally, if we have an algorithm which can solve a class of problems coming from a distribution $\mathcal{D}$ with failure probability at most $\delta$, then it can solve problems coming a distribution $\mathcal{D}'$ with failure probability at most $O(\delta + \delta')$, where $\delta'$ is the variational distance between $\mathcal{D}$ and $\mathcal{D}'$. Since $\delta'$ in our case is $\frac{1}{\poly(n)}$, we can safely ignore this additional failure probability going forward.
This is summarized in the following lemma, which follows directly from the definition of variation distance. Namely, that the probability of any event in one distribution can change by at most the variation distance in another distribution, in particular the event that an algorithm succeeds on that distribution.

\begin{lemma}
Suppose we have an algorithm $\mathcal{A}$ that solves a problem $\mathcal{P}$ taken from a distribution $\mathcal{D}$ over $\R^n$ with probability $1-\delta$. Let $\mathcal{D}'$ be a distribution over $\R^n$ with variation distance at most $\delta' \geq 0$ from $\mathcal{D}$. Then if $\mathcal{P}'$ is drawn from $\mathcal{D}'$, algorithm $\mathcal{A}$ will solve $\mathcal{P}'$ with probability $1 - O(\delta + \delta')$. 
\end{lemma}

\begin{corollary}(Approximate ReLU Recovery.) \label{cor:relurecover}
Let $\AA \in \R^{m \times n}$, $\X\in \R^{d \times n}$ be inputs such that $\AA = \U^* f(\V^* \X)$, where $f$ is the ReLU activation function, $\U^* \in \R^{m \times k}$ has full column rank, $\V^* \in \R^{k \times d}$ has full row rank, for all $i \in [n]$, $\X_{*,i} \sim \mathcal{N}(0, \II)$. Then, there exists an algorithm that recovers $\widehat{\V}$ such that $\|\widehat{\V} - \D \mathbf{\Pi} \V^* \|_F \leq \frac{1}{\poly(n,m,d,\kappa(\U^*))}$, where $\D$ is a diagonal $\pm 1$ matrix and $\mathbf{\Pi}$ is a permutation matrix. Further, the running time of this algorithm is $\poly(n,m,d,\kappa(\U^*))$.
\end{corollary}

First observe that we can assume WLOG that $\Pi = \mathbb{I}$, in other words that we recover an approximate $\V^*$ only up to its signs and not a permutation. We do this by simply (implicitly) permuting the rows of $\V^*$ to agree with our permutation, and permuting the columns of $\U^*$ by the same permutation. The resulting $\AA$ is identical, and so we can assume that we know the permutation already.

\begin{Frame}[\textbf{Algorithm \ref{alg:overall_exact_poly} : ExactNeuralNet}$(\AA, \X)$]
\label{alg:overall_exact_poly}
\ttx{Input:} Matrices  $\AA \in \R^{m \times n}$ and $\X \in \R^{d \times n}$ such that each entry in $\X \sim \mathcal{N}(0, 1)$. \\
\begin{enumerate}
    \item Let $\mathcal{S}_3(x) = H_3(x)$, where $H_3(x) = \frac{- \nabla^{(3)}_x p(x)}{p(x)}$ is the 3-rd order Hermite polynomial and and $p(x) = \frac{1}{(\sqrt{2\pi})^d} e^{-\frac{\|x\|^2_2}{2}}$.
    \item Let $\AA' = \AA + \G$ where $\G \in \R^{m \times n}$ and $\G_{i,j} \sim \mathcal{N}(0,1)$. 
    \item Compute the 4-th order tensor $\widetilde{\T} = \frac{1}{n}\sum^{n}_{i=1} \AA'_{*,i} \otimes \mathcal{S}_3(\X_{*,i})$. Collapse the first mode using a random vector $\theta$. By Lemma \ref{lem:stein}, $\widetilde{T}(\theta, \II, \II, \II) = \sum^{k}_{j=1} \lambda_j \V^*_{j,*} \otimes \V^*_{j,*} \otimes \V^*_{j,*}$, where $\lambda_j = \expecf{x}{f'''(\V^* x)} \langle \U^*_{*,j}, \theta\rangle$. 
    \item Compute a CP-decomposition of $\widetilde{T}(\theta, \II, \II, \II)$ using Tensor Power Method corresponding to Theorem \ref{thm:anandkumar}, \cite{janzamin2015beating}, with accuracy parameter $\epsilon = \frac{1}{\poly(d,m, \kappa(\V), \kappa(\U))}$ to obtain $\widehat{\V}$ such that $\|\widehat{\V} - \D \mathbf{\Pi} \V^* \|_F \leq \frac{1}{\poly(d,m, \kappa(\V), \kappa(\U))}$, where $\D$ is a diagonal $\pm 1$ matrix and $\mathbf{\Pi}$ is a permutation matrix.
    \item Run the Recovering Signs Algorithm (\ref{alg:recoversigns}) on $\widehat{\V}$, $\AA$ and $\X$ to obtain $\V^*$. 
    \item Using the matrix $\V^*$ obtained above, set up and solve the following linear system for the matrix $\U$:
    \begin{equation}
        \AA = \U f(\V^* \X)
    \end{equation}
    \item Let $\U^*$ be the solution to the above linear system. 
\end{enumerate}
\ttx{Output:}  $\U^* ,\V^*$.
\end{Frame}

Unfortunately, the ambiguity in signs resulting from the algorithm of Theorem \ref{thm:anandkumar} is a non-trivial difficulty, and must be resolved algorithmically. This is due to the fact that the ReLU is sensitive to negative scalings, as $f(\cdot)$ only commutes with positive scalings. Suppose the diagonal of $\D$ of Corollary \ref{cor:relurecover} is given by the coefficents $\xi_i \in \{1,-1\}$. Then in order to recover the weights, we must recover the terms $\xi_i$,
Naively trying each sign results in a running time of $2^k$, which is no longer polynomial\footnote{We remark that some prior results \cite{janzamin2015beating} were able to handle this ambiguity by considering only a restricted class of smooth activation functions $f(\cdot)$ with the property that $f(x) = 1-f(-x)$ for all $x \in \R$. Using affine transformations after application of the ReLU, this sign ambiguity for such activation functions can be accounted for. Since firstly the ReLU does not satisfy this condition and is non-trivially sensitive to the signs of its input,and secondly we are restricting to optimization over networks without affine terms, a more involved approach to dealing with sign ambiguity is required (especially for the noisy case).}. Thus, a considerably technical challenge will be to show how to determine the correct scaling for each row even in the presence of noise. We begin with the case where there is no noise.

\paragraph{Recovering $\V^*$ from the Tensor Decomposition in the Noiseless Case.}
Recall that the tensor power method provides us with row vectors $v_i$ such that $\|v_i - \xi_i \V_{i,*}^* \|_2 \leq \eps$ where $\eps = O\left(\allowbreak \frac{1}{\poly(d,m,k,\kappa(\U^*),\kappa(\V^*))}\right)$ for $\xi_i \in \{\V^*_{i,*}, -\V^*_{i,*}\}$. Thus, the tensor power method gives us a \textit{noisy} version of \textit{either} $\V^*_{i,*}$ or $-\V^*_{i,*}$, however we do not know which. A priori, it would require $2^k$ time to guess the correct signs of the $k$ vectors $v_i$. In this section, we show that using the combinatorial sparsity patterns in the row span of $\AA$, we can not only recover the signs, but recover the matrix $\V^*$ \textit{exactly}. Our procedure is detailed in Algorithm \ref{alg:recoversigns} below, which takes the outputs $v_i$ from the tensor power method and returns the true matrix $\V^*$ up to a permutation of the rows.

\begin{Frame}[\textbf{Algorithm  \ref{alg:recoversigns}: Exact Recovery of $\V^*$ }]
\label{alg:recoversigns}
\ttx{Input:} Matrices  $\AA = \U^*f(\V^*\X) + \E$, and $v^T_i \in \R^d$ s.t. $\|v_i - \xi_i \V_{i,*}^*\|_2 \leq \eps$ for some $\eps = O\left(\frac{1}{\poly(d,m,\kappa(\U^*),\kappa(\V^*))}\right)$ for some unknown $\xi_i \in \{ 1,-1\}$ and each $i=1,2,\dots,k$.
\begin{enumerate}
\item Let $\overline{\X} \in \R^{d \times \ell}$ be the first $\ell = \poly\left(k,d,m,\kappa(\U^*),\kappa(\V^*)\right)$ columns of $\X$, and let $\overline{\AA} = \U^*f(\V^*\overline{\X})$. 
    \item Let $\tau = \Theta(1/\poly(\ell))$ be a thresholding value. Define the row vectors $v_i^+,v_i^- \in \R^\ell$ via
    
    \[(v_i^+)_j = \begin{cases}
     f(v_i\overline{\X})_{j} & \text{ if } f(v_i\overline{\X})_{j} > \tau \\
     0 & \text{ otherwise }
    \end{cases} \; \; \; \; \; \; v_i^- = \begin{cases}
     f(-v_i\overline{\X})_{j} & \text{ if } f(-v_i\overline{\X})_{j} > \tau \\
     0 & \text{ otherwise }
    \end{cases} \]
  for $j=1,2,\dots,\ell$.
    \item Let $S^+_i$ be the sign pattern of $v_i^+$, and $S^-_i$ be the sign pattern of $v_i^-$.  For $q \in \{+,-\}$, solve define the $r$ linear systems of equations in the variable $w^q_i \in \R^k$, where the $r$-th system is given by
      \[ (w^q_i\overline{\AA})_j  = 0 \; \; \text{ for } j \notin S^q_i \]
      \[ (w^q_i)_r = 1 \]
      Where $(w^q_i)_r$ is the $r$-th coordinate of $(w^q_i)$. Then let $(w^q_i)$ be the vector returned from the first linear system which had a solution. 
     \item Let $q'$ be such that the above linear system returns a solution $w^{q'}_i$ with the constraints given by $S^{q'}_i$ (and at least one of the constraints of the form $(w^{q'}_i)_r = 1$). We output \textbf{FAIL} if this occurs for both $q \in \{+,-\}$.  
     \item Output $\V_{i,*}=z_i/\|z_i\|_2$ where $z_i$ is the  solution to the following linear system.
     \begin{equation*}
    \begin{array}{ll@{}ll}
    \text{ for all } j \in S^{q'}_i & \displaystyle    &(z_i \overline{\X})_j = (w^q_i\overline{\AA})_j  & \\ 
    \end{array}
    \end{equation*}
\end{enumerate}
\ttx{Output:} $\V$ such that $\V = \V^*$.
\end{Frame}

 Before we proceed, we recall a standard fact about the singular values of random Gaussian matrices.

	\begin{lemma}[Corollary 5.35 \cite{vershynin2010introduction}] \label{prop:gaussiancondition}
	Let $\S \in \R^{k \times n}$ be a matrix of i.i.d. normal $\mathcal{N}(0,1)$ random variables, with $k < 10n$. Then with probability $1- 2e^{-n/8}$, for all row vectors $w \in \R^{\ell}$ we have 
	\[  \sqrt{n}/3 \|w\|_2  \leq  \|w \S\|_2 \leq 2\sqrt{n} \|w\|_2 \]
	In other words, we have $  \sqrt{n}/3   \leq \sigma_{\min}(\S) \leq \sigma_{\max}(\S) \leq \ 2\sqrt{n} $.
	\end{lemma}

\begin{theorem}\label{thm:recoverexactsigns}
With high probability in $d,m$, Algorithm \ref{alg:recoversigns} does not fail, and finds $\V$ such that $\V = \V^*$.
\end{theorem}
\begin{proof}
  Fix a $i \in [k]$, and WLOG suppose the input row $v_i$ is such that $\|v_i - \V_{i,*}^*\|_2 \leq \eps$ (i.e. WLOG suppose $\xi_i = 1$). Then $\|v_i\overline{\X} - \V_{i,*}^*\overline{\X}\|_2 \leq  O(1) \sqrt{\ell} \eps $ by the operator norm bound of Lemma \ref{prop:gaussiancondition}, and since $f$ can only decrease the distance between matrices, it follows that  $\|f(v_i\overline{\X}) - f(\V_{i,*}^*\overline{\X})\|_2 \leq O(1) \sqrt{\ell} \eps $ . Similarly, we have $\|f(-v_i\overline{\X}) - f(-\V_{i,*}^*\overline{\X})\|_2 \leq  O(1) \sqrt{\ell} \eps $ . 
  
  We now condition on the event that none of the non-zero entries of $f(\V_{i,*}^*\overline{\X}),$ and  $f(-\V_{i,*}^*\overline{\X})$ are less than $\tau = \Theta(1/\poly(\ell))$ (where $\tau$ is as in Algorithm \ref{alg:recoversigns}), which holds by a union bound with probability $1-1/\poly(\ell)$ (high prob in $d,m$) and the fact that the non-zero entries of these matrices have folded Gaussian marginals (distributed as the absolute value of a Gaussian). Given this, it follows that the sign patterns $S^+_i$ and $S^-_i$ of $v_i^+$ and $v_i^-$ are precisely the sign patterns of $f(\V_{i,*}^*\overline{\X})$ and $f(-\V_{i,*}^*\overline{\X})$ respectively. Since $f(\V_{i,*}^*\overline{\X})$ is in the row space of $\AA$, at least one of the the linear systems run on $S^+_i$ will have a unique solution given by taking $w^+_i = c e_i^T \U^{-1}$ for an appropriate constant $c \neq 0$ such that one of the constraints of the form $(w^+_i)_r = 1$ is satisfied, and where $\U^{-1}$ is the left inverse of $\U$. 
  
  Now consider the matrix $\W$ such that $\W$ is $\V^*$ with the row $-\V_{i,*}^*$ appended at the end. Then Applying the same argument as in Lemma \ref{prop:uniquesign}, we see that the sign patterns of every pair of rows of $f(\W\overline{\X})$ disagrees on at least $\poly(k)$ signs w.h.p.. This is easily seen for all pairs which contain one of $\{\V^*_{i,*},-\V^*_{i,*}\}$ by applying the exact argument of the lemma and noting that the condition number of the matrix $\V^*$ does not change after negating the $i$-th row. The pair $\{f(\V^*_{i,*}\overline{\X}),f(-\V^*_{i,*}\overline{\X})\}$ itself disagrees on all sign patterns, which completes the proof of the claim. Note here that disagree means that, for any two rows $y^i,y^j$ in question there are at least $\poly(k)$ coordinates such that \textit{both} $y^i_p > 0$ and $y^j_p < 0$ \textit{and} vice-versa. Thus no sparsity pattern is contained within any other.  Then by Lemma \ref{lem:uniquesign}, it follows that with high probability the only vector in the row span of $f(\W\overline{\X})$ which has a sparsity pattern contained within $S^-_i$ is a scalar multiple of $f(-\V^*_{i,*}\overline{\X})$. Since no vector with such a sparsity pattern exists in the row span of $f(\V^*\overline{\X})$, the linear system with constraints given by $S^-_i$ will be infeasible with high probability. 
  
  We conclude from the above $q' = +$ in the fourth step of Algorithm \ref{alg:recoversigns}, and that $w^{q'}_i = w^+_i$ is such that the sign pattern of $w^+_i\overline{\AA}$ is $S^+_i$, which is also the sign pattern of $f(\V^{*}_{i,*} \overline{\X})$. Since $w^+_i\overline{\AA}$ is in the row span of $f(\V^*\overline{\X})$, again by Lemma \ref{lem:uniquesign}, we conclude that $w^+_i\overline{\AA} = cf(\V^{*}_{i,*} \overline{\X})$ for some constant $c > 0$ (we can enfoce $c>0$ by flipping the sign of $w^+_i$ so that $w^+_i\overline{\AA}$ has no strictly negative entries). The linear system in step $5$ solves the equation $z_i \overline{\X}_{S^+_i} = w^+_i\overline{\AA}_{S^+_i}$, where $\overline{\X}_{S^+_i}$ is $\overline{\X}$ restricted to the columns corresponding to indices in $S^+_i$, and similarly with $\overline{\AA}_{S^+_i}$. This will have a unique solution if $\overline{\X}_{S^+_i}$ has full row rank. Since an index is included in $S^+_i$ with probability $1/2$ independently, it follows that $|S^+_i| > \ell/3 > \poly(d)$ with probability $1-2^{-\Omega(\ell)}$. A column of $\overline{\X}_{S^+_i}$ is just an i.i.d. Gaussian vector conditioned on being in a fixed half-space. Then if the first $i<d$ columns of $\overline{\X}_{S^+_i}$ are independent, they span a $i-1$ dimensional subspace. The Lebesgue measure of this subspace intersected with the halfspace has has measure $0$, since the half-space is $d$-dimensional and the subspace is $i$-dimensional. It follows that the probability that the $i+1$ column of $\overline{\X}_{S^+_i}$ is in this subspace is $0$, from which we conclude by induction that $\overline{\X}_{S^+_i}$ has rank $d$. Thus the solution $z_i$ is unique, and must therefore be equal to $\frac{1}{c}\V^*_{i,*}$ as we also have $\V^*_{i,*}\X_{S^+_i} = cw^+_i\overline{\AA}$ for $c>0$. After normalizing $z_i$ to have unit norm, we conclude $\V_{i,*} = z_i/\|z\|_2 = \V^*_{i,*}$ as needed.
  
\end{proof}

Given the results developed thus far, the correctness of our algorithm for the exact recovery of $\U^*,\V^*$ in the realizable (noiseless) case follows immediately. Recall we can always assume WLOG that $\|\V^*_{i,*}\|_2 = 1$ for all rows $i \in [k]$.
\begin{theorem}(Exact Recovery for Gaussian Input.)\label{thm:exactfinal}
Suppose $\AA = \U^* f(\V^* \X)$ where $\U^* \in \R^{m \times k}, \V^* \in \R^{k \times d}$ are both rank-$k$, and such that $\X \in \R^{d \times n}$ is i.i.d. Gaussian.  Assume WLOG that $\|\V^*_{i,*}\|_2 = 1$ for all rows $i \in [k]$. If $n = \Omega(\poly(d,m,\kappa(\U^*),\kappa(\V^*)))$, then Algorithm \ref{alg:overall_exact_poly} runs in $\poly(n)$-time and recovers $(\U^*)^T,\V^*$ exactly up to a permutation of the rows w.h.p. (in $d,m$).  
\end{theorem}
\begin{proof}
 By Theorem \ref{thm:anandkumar} and Corollary \ref{cor:relurecover}, we can recover $\D\V^*$ up to $\eps =  \frac{1}{\poly(d,m,\kappa(\U^*),\kappa(\V^*))}$ error in polynomial time, and then by Theorem \ref{thm:recoverexactsigns} we can not only recover the signs $\xi_i$ that constitute the diagonal of $\D$, but also recover $\V^*$ \textit{exactly} (all in a polynomial number of arithmetic operations). Given the fact that $f(\V^*\X)$ is full rank by Lemma \ref{prop:random}, the solution $\U$ to the linear system $\U f(\V^*\X) = \AA$ is unique, and therefore equal to $\U^*$.  This linear system can be solved in polynomial time by Gaussian elimination, and thus the runtime does not depend on the bit-complexity of $\U^*,\V^*$ in the real RAM model. So the entire procedure runs in time polynomial in the sample complexity $n$, which is polynomial in all relevant parameters as stated in the Theorem.
\end{proof}

\subsection{Extension to Symmetric Input Distributions.} \label{sec:symmetric}
The independent and concurrent work of Ge et al. \cite{ge2018learning} demonstrates the existence of an algorithm that approximately recovers $\U^*$, $\V^*$ in polynomial time, given that the input $\X$ is drawn from a mixture of a symmetric probability distribution and a Gaussian. In this section, we observe how our techniques can be combined with the those of \cite{ge2018learning} to achieve exact recovery of $\U^*,\V^*$ for this broader class of distributions. Namely,
 that we can replace running the tensor decomposition algorithm from \cite{janzamin2015beating} with the algorithm of \cite{ge2018learning} instead to obtain good approximations to $\U^*,\V^*$, and then use our results on the uniqueness of sparsity patterns in the row-span of $f(\V^*\X)$ to 
 obtain exactly recovery. Only minor changes are needed in the proofs of our sparsity pattern uniqueness results (Lemmas \ref{prop:uniquesign} and \ref{lem:uniquesign}) to extend them to mixtures of symmetric distributions and Gaussians.

 . 

\begin{definition}(Symmetric Distribution.)
Let $x\in \R^d$ be a vector random variable and $\mathcal{D}$ be a probability distribution function such that $x \sim \mathcal{D}$. Then, $\mathcal{D}$ is a symmetric distribution if for all $x$, the probability of $x$ and $-x$ is equal, i.e. $\mathcal{D}(x) = \mathcal{D}(-x)$. 
\end{definition}

Ge et. al. \cite{ge2018learning} define an object called the distinguishing matrix, denoted by $\M$, and require that the minimum singular value of $\M$ is bounded away from $0$. 

\begin{definition}(Distinguishing Matrix \cite{ge2018learning}.)
Given an input distribution $\mathcal{D}$ the distinguishing matrix is defined as  $\NN^{\mathcal{D}} \in \R^{d^2 \times \binom{k}{2}}$, whose columns are indexed by $i,j$ such that $1\leq i < j\leq k$ and
\[
\NN^{\mathcal{D}}_{i,j} = \frac{1}{n}\sum_{k\in[n]} (\V^*_{i,*}\X_{*,k})(\V^*_{j,*}\X_{*,k})(\X_{*,k} \otimes \X_{*,k}) \mathbf{1}\left\{(\V^*_{i,*}\X_{*,k})(\V^*_{j,*}\X_{*,k})\leq 0\right\}
\]
Similarly an augmented distinguishing matrix $\M^{\mathcal{D}} \in \R^{d^2 \times \left(\binom{k}{2}+1\right)}$ has all the same columns as $\NN^{\mathcal{D}}$ with the last column being $\frac{1}{n}\sum_{k\in[n]} \X_{*,k} \otimes \X_{*,k}$.
\end{definition}
In order to bound the singular values of the distinguishing matrix, Ge et. al. consider input distributions that are perturbations of symmetric distributions. In essence, given a desired target distribution $\mathcal{D}$, the algorithm of Ge et. al. can handle a similar distribution $\mathcal{D}_\gamma$, which is obtained by mixing $\mathcal{D}$ with a Gaussian with random covariance.

More formally, the perturbation is paramaterized by $\gamma\in(0,1)$, which will define the mixing rate. It is required that $\gamma > \frac{1}{\poly(N)}$ in order to achieve polynomial running time (where $\poly(N)$ is the desired running time of the algorithm). First, let $\G$ be an i.i.d. entry-wise $\mathcal{N}(0,1)$ random Gaussian matrix, which will be used to give the random the covariance.  To generate $\mathcal{D}_\gamma$, first define a new distribution $\mathcal{N}_{\G}'$ as follows. To sample a point from $\mathcal{N}_{\G}'$, first sample a Gaussian $g \sim \mathcal{N}(0, \II_d)$ and then output $\G g$. Then the perturbation $\mathcal{D}_{\gamma}$ of the input distribution $\mathcal{D}$ is a mixture between $\mathcal{D}$ and $\mathcal{N}_{\G}'$. To sample $\X_{*,i}$ from $\mathcal{D}_{\gamma}$, 
pick $z$ as a Bernoulli random variable where $\Pr[z = 1] = \gamma$ and $\Pr[z = 0] = 1 - \gamma$, then  for $i \in [n]$ 
\begin{equation*}
  \X_{*,i} \sim 
     \begin{cases}
     \mathcal{D} \; \;\;\;\;\;\; \;\;\;\;\;\; \;\;\;\;\; \textrm{ if $z = 0$ } \\
     \G g \; \;\;\;\; \;\;\;\;\;\; \;\;\;\;  \textrm{ otherwise}
     \end{cases}
\end{equation*}
If the input is drawn from a mixture distribution $\mathcal{D}_{\gamma}$,  $\sigma_{\min}(\M)$ is bounded away from $0$. We refer the reader to Section 2.3 in \cite{ge2018learning} for further details. We observe that we can extend the main algorithmic result therein with our results on exact recover to recover $\U^*$, $\V^*$ with zero-error in polynomial time. 

\begin{theorem}(Informal Theorem 7 in \cite{ge2018learning}.)\label{thm:rongs}
Let $\U^*\in \R^{m \times k}, \V^* \in \R^{k \times d}$ be full rank $k$ such that $\AA = \U^* f(\V^* \X)$, $f$ is ReLU, and for all $i \in [n]$ $\X_{*,i}\sim \mathcal{D}_{\gamma}$ as defined above. Let $\M$ be the distinguishing matrix as defined in \cite{ge2018learning}. For all $i \in [n]$, let $\Gamma$ be such that $\| \X_{*,i}\|_2\leq \Gamma$. Then, there exists an algorithm than runs in time $\poly\left(\Gamma, 1/\eps, 1/\delta , \|\U^*\|_2, \frac{1}{\sigma_{\min}(\expec{}{\X_{*,i} \X_{*,i}^T})}, \frac{1}{\sigma_{\min}(\U^*)}, \frac{1}{\sigma_{\min}(\M)}\right)$ and with probability $1-\delta$ outputs a matrix $\widehat{\U}$ such that $\|\widehat{\U} - \U^*\mathbf{\Pi}\D\|_F \leq \eps$. 
\end{theorem}

We use the algorithm corresponding to the aforementioned theorem to obtain an approximation $\widehat{\U}$ to $\U^*$, and then obtain an approximation to $f(\V^*\X)$ by multiplying $\AA$ on the left by $\widehat{\U}^{-1}$. The error in our approximation of $\V^*$ obtained via $\widehat{\U}^{-1}\AA$ is analyzed in Section \ref{sec:ICA}. Given this approximation of $\V^*$, we observe that running steps $4$-$8$ of our Algorithm \ref{alg:overall_exact_fpt}  recovers $\V^*,\U^*$ exactly (see Remark \ref{rem:symmetric} below).  Note that the only part of Algorithm  \ref{alg:overall_exact_fpt} that required $\V^*$ to be orthonormal is step $3$ which runs ICA, which we are replacing here with the algorithm of Theorem \ref{thm:rongs}.

Here we remove the random matrix $\T$ from Algorithm \ref{alg:overall_exact_fpt}, as it is not needed if we are already given an approximation $\widehat{\U}$ of $\U^*$. Thus we proceed exactly as in Algorithm \ref{alg:overall_exact_fpt} by restricting $\X,\AA$ to $\ell = \poly(d,m,k,\frac{1}{\gamma},\kappa(\U^*),\kappa(\V^*))$ columns $\overline{\X}, \overline{\AA}$, and then rounding the entries of $\widehat{f(\V\overline{\X})} = \widehat{\U}^{-1} \overline{\AA}$ below $\tau$ to $0$. Finally, we solve the same linear system as in Algorithm \ref{alg:overall_exact_fpt} to recover the rows of $f(\V^*\X)$ exactly, from which $\V^*$ and then $\U^*$ can be exactly recovered via solving the final two linear systems in Algorithm \ref{alg:overall_exact_fpt}. We summarized this formally as follows.

\begin{corollary}(Exact Recovery for Symmetric Input.)
\label{cor:symmetric_input}
Suppose $\AA = \U^* f(\V^* \X)$ where $\U^* \in \R^{m \times k}, \V^* \in \R^{k \times d}$ are both rank-$k$, for all $i \in [n]$, $\X_{*,i} \sim \mathcal{D}_{\gamma}$, and $\| \X_{*,i}\|_2\leq \Gamma$.  Assume WLOG that $\|\V^*_{i,*}\|_2 = 1$ for all rows $i \in [k]$. If $$n \geq \poly\left(d,m,\kappa(\U^*),\kappa(\V^*),\Gamma, \frac{1}{\gamma}, \|\U^*\|_2, \frac{1}{\sigma_{\min}(\expec{}{\X_{*,i} \X_{*,i}^T})}, \frac{1}{\sigma_{\min}(\U^*)}, \frac{1}{\sigma_{\min}(\M)}\right)$$ 
then there exists an algorithm that runs in $\poly(n)$-time and recovers $(\U^*)^T,\V^*$ exactly up to a permutation of the rows w.h.p. (in $d,m$).  
\end{corollary}

\begin{remark}\label{rem:symmetric}
To prove the correctness of Algorithm \ref{alg:recoversigns} on $\mathcal{D}_\gamma$, we need only to generalize Lemmas \ref{prop:uniquesign} and \ref{lem:uniquesign} which together give the uniqueness of sparsity patterns of the rows of $f(\V^*\X)$ in the rowspan of $\AA$.
We note that Lemma \ref{prop:uniquesign} can be easily generalized by first conditioning on the input being Gaussian, which in $\mathcal{D}_\gamma$ occurs with $\gamma$ probability, and then going applying the same argument, replacing $\frac{1}{\kappa}$ with $\frac{\gamma}{\kappa}$ everywhere. The only change in the statement of Lemma \ref{prop:uniquesign} is that we now require $\ell = t \poly(k , \kappa, \frac{1}{\gamma})$ to handle $\X \sim \mathcal{D}_\gamma$. 

Next, the proof of Lemma \ref{lem:uniquesign} immediately goes through as the argument in the proof which demonstrates the determinant in question is non-zero only requires that the distribution $\mathcal{D}_\gamma$ is non-zero everywhere in the domain. Namely, the proof requires that the support of $\mathcal{D}_\gamma$ is all of $\R^d$. Note, this condition is always the case for the mixture $\mathcal{D}_{\gamma}$ since Gaussians are non-zero everywhere in the domain.  
\end{remark}

\subsection{Necessity of $\poly(\kappa(\V^*))$ Sample Complexity} \label{subsec:kappadependency}
So far, our algorithms for the exact recovery of $\U^*,\V^*$ have had polynomial dependency on the condition numbers of $\U^*$ and $\V^*$. In this section, we make a step towards justifying the necessity of these dependencies. As always, we work without loss of generality under the assumption that $\|\V_{i,*}\|_2 = 1$ for all rows $i \in [k]$. Specifically, demonstrate the following.
\begin{lemma}\label{lem:kappadependency}
 Any algorithm which, when run on $\AA, \X$, where that $\AA = \U^* f(\V^*\X)$, and $\X$ has i.i.d. Gaussian $\mathcal{N}(0,1)$ entries, recovers $(\U^*)^T ,\V^*$ exactly (up to a permutation of the rows) with probability at least $1 - c$ for some sufficiently small constant $c>0$, requires $n = \Omega(\kappa(\V^*))$ samples.
\end{lemma}
\begin{proof}
  We construct two instances of $\AA^1 = \U^1f(\V^1\X)$ and $\AA^2= \U^1f(\V^2\X)$ . Let \[\U^1 = \begin{bmatrix}\sqrt{1+ a^2}/2 & \sqrt{1+a^2}/2  \\
  \end{bmatrix}  \; \; \; \; \; \; \; \V^1 = \begin{bmatrix} \frac{1}{\sqrt{1+ a^2}} & \frac{a}{\sqrt{1+ a^2}} \\
     \frac{1}{\sqrt{1+ a^2}} &- \frac{a}{\sqrt{1+ a^2}} \\
  \end{bmatrix}  \]
  
  \[\U^2 = \begin{bmatrix}\sqrt{1+ (2a)^2}/2 & \sqrt{1+(2a)^2}/2  \\
  \end{bmatrix} \; \; \; \; \; \; \; \V^2 = \begin{bmatrix} \frac{1}{\sqrt{1+ (2a)^2}} & \frac{a}{\sqrt{1+ (2a)^2}} \\
     \frac{1}{\sqrt{1+ (2a)^2}} &- \frac{(2a)}{\sqrt{1+ (2a)^2}} \\
  \end{bmatrix}  \]
  Now note that for $a \in [0,1]$, the rows of $\V^1$ have unit norm, and $\kappa(\V^1) = \frac{1}{a}$. Now let $a^i = ia$, and note, however, that for the $j$-th sample $\X_{*,j} = [x_1^j,x_2^j]$ i.i.d. Gaussian, we have for $i \in \{1,2\}$
  \[ \U^i f(\V^i \X_{*,j}) =  \frac{f(x_1^j + a^ix_2^j) + f(x_1^j - a^ix_2^j)}{2}  \]
 Now note that when $|x_1^j| > (2a)|x_2^j|$, for \textit{both} $i \in \{1,2\}$ we have either 
   \[\AA^i_{*,j} = \U^i f(\V^i \X_{*,j}) =  0  \]
   or 
    \[\AA^i_{*,j} =  \U^i f(\V^i \X_{*,j}) = x_1^j   \]
    And in either case we do not get any information about $a$. In such a case, the $j$-th column of $\AA^1$ and $\AA^2$ are the same. In particular, conditioned on a given $\X$ such that $|x_1^j| > (2a)|x_2^j|$ for all columns $j$, we have $\AA^1 = \AA^2$. Now note that the probability that one Gaussian is $\frac{1}{2a}$ times larger than another is $\Theta(\frac{1}{a})$, thus any algorithm that takes less than $c \frac{1}{a}$ samples, for some absolute constant $c > 0$, cannot distinguish between $\AA^1$ and $\AA^2$, since we will have $\AA^1 = \AA^2$ with $\Omega(1)$ probability in this case, which completes the proof.
\end{proof}

\section{A Polynomial Time Algorithm for Gaussian input and Sub-Gaussian Noise}  
\label{sec:noisycase}
In the last section, we gave two algorithms for exact recovery of $\U^*,\V^*$ in the noiseless (exact) case. Namely, where the algorithm is given as input $\AA = \U^*f(\V^*\X)$ and $\X$. Our general algorithm for this problem first utilized a tensor decomposition algorithm which allowed for approximate recovery of $\V^*$, up to the signs of its rows. Observe that this procedure, given by Theorem \ref{thm:anandkumar}, can handle mean zero subgaussian noise $\E$, such that $\AA = \U^*f(\V^*\X) + \E$. In this section, we will show how to utilize this fact as a sub-procedure to recover approximately recover $\U^*, \V^*$ in this noisy case.

We begin with using the algorithm corresponding to Theorem \ref{thm:anandkumar} to get an approximate solution to $\V^*$, up to permutations and $\pm 1$ scaling. We note that the guarantees of Theorem \ref{thm:anandkumar} still hold when the noise $\E$ is sub-Gaussian. Therefore, we obtain a matrix $\widetilde{\V}$ such that $\|\widehat{\V} - \D \mathbf{\Pi} \V^* \|_F \leq \epsilon$, where $\D$ is a diagonal $\pm 1$ matrix and $\mathbf{\Pi}$ is a permutation matrix.

\begin{Frame}[\textbf{Algorithm  \ref{alg:recoversigns2}: Recovering Signs}$(v_i, \AA, \X)$]
\label{alg:recoversigns2}
\ttx{Input:} Matrices  $\AA = \U^*f(\V^*\X) + \E$, and $v^T_i \in \R^d$ s.t. $\|v_i - \xi_i\V_{i,*}^*\|_2 \leq \eps$ for some unknown $\xi_i \in \{ 1,-1\}$ and $i=1,2,\dots,k$, where $\eps = O(\frac{1}{\poly(d,m,k,\kappa(\U^*),\kappa(\V^*)})$.
\begin{enumerate}
\item Let $\overline{\X} \in \R^{d \times \ell}$ be the first $\ell = \poly(k,d,m,\kappa(\U^*),\kappa(\V^*),\sigma)$ columns of $\X$, and similarly define $\overline{\E}$, and let $\overline{A} = \U^*f(\V^*\overline{\X}) + \overline{\E}$. 
    \item For $i \in [k]$, let $$S_{i,+} = \{f(v_j\overline{\X}), f(-v_j\overline{\X})\}_{j \neq i} \cup \{f(v_i\overline{\X})\}$$ and $$S_{i,-} = \{f(v_j\overline{\X}, f(-v_j\overline{\X})\}_{j \neq i} \cup \{f(-v_i\overline{\X})\}$$ 
    
    \item Let $\P_{S_{i,+}}$ be the orthogonal projection matrix onto the row span of vectors in $S_{i,+}$. Compute
    \[ a^+_{i,j} = \|\overline{\AA}_{j,*}(\mathbb{I} - \P_{S_{i,+}}) \|_2^2 \]
      \[ a^-_{i,j} = \|\overline{\AA}_{j,*}(\mathbb{I} - \P_{S_{i,-}}) \|_2^2 \]
    For each $j \in [m]$.
    \item Let $a^+_i = \sum_j a^+_{i,j}$, and $a^-_i = \sum_j a^-_{i,j}$. If $a^+_i < a^-_i$, set $\V_{i,*} = v_i$, otherwise set $\V_{i,*} = -v_i$.
\end{enumerate}
\ttx{Output:} $\V$ such that $\|\V - \V^*\|_2 \leq \eps$, thus recovering $\xi_i$ for $i \in [k]$.
\end{Frame}

Recall that in the noiseless case, we needed to show that given a approximate version of $\V^*$ up to the signs of the rows, we can recover both the signs and $\V^*$ exactly in polynomial time. Formally, we were given rows $v_i$ such that $\|v_i - \xi_i \V^*_{i,*}\|_2$ was small for some $\xi_i \in \{1,-1\}$, however we did not know $\xi_i$. This issue is a non-trivial one, as we cannot simply guess the $\xi_i$'s (there are $2^k$ possibilities), and moreover we cannot assume WLOG that the $\xi_i$'s are $1$ by pulling the scaling through the ReLU, which is only commutes with \textit{positive} scalings. Our algorithm for recovery of the true signs $\xi_i$ in the exact case relied on combinatorial results about the sparsity patterns of $f(\V^*\X)$. Unfortunately, these combinatorial results can no longer be used as a black-box in the noisy case, as the sparsity patterns can be arbitrarily corrupted by the noise. Thus, we must develop a refined, more general algorithm for the recovery of the signs $\xi_i$ in the noise case. Thus we begin by doing precisely this.

\subsection{Recovering the Signs $\xi_i$ with Subgaussian Noise}
\label{subsec:noiseysigns}

\begin{lemma}\label{prop:subgaussian}
Let $g \in \R^n$ be a row vector of i.i.d. mean zero variables with variance $\sigma$, and let $\mathcal{S}$ be any fixed $k$ dimensional subspace of $\R^n$. Let $\P_\mathcal{S} \in \R^{n \times n}$ be the projection matrix onto $\mathcal{S}$. Then for any $\delta>0$ with probability $1-\delta$, we have
\[ \|g \P_\mathcal{S}\|_2 = \sigma\sqrt{ k /\delta} \]
\end{lemma}
\begin{proof}
We can write $P_\mathcal{S}  = \W^T \W$ for matrices $\W \in \R^{k \times n}$ with orthonormal rows. Then $\ex{\|g \W^T\|_2^2} = \sigma^2 k$, and by Markov bounds with probability $1-\delta$ we have $\|g \W^T\|_2^2 = \|g \W^T \W\|_2^2 <  \sigma^2 k /\delta$ as needed.

\end{proof}

\begin{lemma}\label{prop:projectionbound}
Let $\Q \in \R^{k \times \ell}$ be a matrix of row vectors for $\ell > \poly(k)$ (for some sufficiently large polynomial) with $1 \leq \|\Q\|_2$ and let $\P_{\Q} = \Q^T(\Q^T \Q)^{-1} \Q$ be the projection onto them. Let $\E$ be such that $\|\E\|_F  \leq \frac{\eps}{\big(\kappa(\Q)\|\Q\|_2\big)^4}$, and let $\P_{\Q+\E}$ be the projection onto the rows of $\Q+\E$ . Then for any vector $x^T \in \R^\ell$, we have
\[  \|x\P_{\Q + \E} \|_2=      \| x\P_{\Q} \|_2 \pm O(\eps\|x\|_2)\]
\end{lemma}
\begin{proof}
  We have $\P_{\Q+\E} =  (\Q + \E)^T(  (\Q + \E)^T  (\Q + \E))^{-1}  (\Q + \E) $.  Now $$(\Q + \E)^T  (\Q + \E) = \Q^T \Q + \E^T \Q + \Q^T \E + \E^T \E$$ 
  Further, $\| \E^T \Q + \Q^T \E + \E^T \E\|_F \leq \|\E\|_F \|\Q\|_2 + \|\E\|_F^2 \leq 2\frac{\eps}{\kappa^4(\Q)\|\Q\|_2^2}$. Thus we can write $(\Q + \E)^T  (\Q + \E)  = \Q^T\Q + \ZZ$ where $\|\ZZ\|_F \leq 2\frac{\eps}{\kappa^4(\Q)\|\Q\|_2^2}$. 
  Applying Corollary \ref{cor:psuedobound} with $\BB = \Q^T\Q$, and $\E = \ZZ$, we can write $(\Q + \E)^T  (\Q + \E))^{-1} = (\Q^T\Q)^{-1} + \ZZ'$, where $\|\ZZ'\|_F \leq O(\frac{\eps}{\kappa^2(\Q)\|\Q\|_2^2})$. Thus 
  \begin{equation*}
      \begin{split}
        \P_{\Q+\E} & = (\Q + \E)^T((\Q^T\Q)^{-1} + \ZZ') (\Q + \E) \\
        & =  \P_{\Q}  + \Q^T \ZZ' (\Q + \E)  + \Q^T (\Q^T\Q)^{-1} \E  + \E^T((\Q^T\Q)^{-1} + \ZZ')(\Q+ \E) 
      \end{split}
  \end{equation*}
   Therefore,
   \begin{equation*}
      \begin{split}
        \|\Q^T \ZZ' (\Q + \E)\|_F & \leq \|\Q^T\|_2 \|\ZZ' (\Q + \E)\|_F \\
        & \leq \|\Q^T\|_2 \|\ZZ' \|_F \|(\Q + \E)\|_2 \\
        & \leq \|\Q^T\|_2 \|\ZZ' \|_F  (\|\Q\|_2 + \|\E\|_F) \\
        & = O\left(\frac{\eps}{\kappa^2(\Q)}\right) = O(\eps)
      \end{split}
  \end{equation*}
   Next, we have 
   \begin{equation*}
      \begin{split}
       \| \Q^T (\Q^T\Q)^{-1} \E\|_F & \leq \| \Q^T \|_2 \|(\Q^T\Q)^{-1}\|_2 \|\E\|_F \\ 
       & \leq \frac{\eps}{\|\Q\|_2^3} \\
       & < \eps
      \end{split}
  \end{equation*}
  where in the second to last inequality we used the fact that $\|\Q\|_2 > 1$ so $\sigma^{-2}_{\min}(\Q) = \|(\Q^T \Q)^{-1}\|_2 < 1/\kappa^2(\Q)$. Applying the above bounds similarly, we have $\|\E^T(\Q^T\Q)^{-1}\Q\|_F \leq O(\eps)$, $\|\E^T(\Q^T\Q)^{-1}\E\|_F \leq O(\eps),$ and $\|\E^T\ZZ'(\Q+\E)\|_F \leq O(\eps)$. We conclude $\P_{\Q+\E} = \P_{\Q} + \ZZ''$, where $\| \ZZ''\|_F \leq O(\eps)$. It follows that for any $x \in \R^\ell$, we have
  \begin{equation*}
      \begin{split}
       \|x \P_{\Q+\E}\|_2 & = \| x \P_{\Q}\|_2 \pm \|x \ZZ''\|_2 \\
       & = \| x \P_{\Q}\|_2 \pm  O(\eps\|x\|_2)
      \end{split}
  \end{equation*}
\end{proof}

\begin{lemma}\label{prop:reallylongprojectionprop}
Let $\Q \in \R^{r \times \ell}$ for $1 \leq r \leq 2k$ be any matrix whose rows are formed by taking $r$ distinct rows from the set $\{f(\V^*_{i,*}\overline{\X}), f(-\V^*_{i,*}\overline{\X}) \}_{i \in [k]}$, where $\overline{\X}$ is $\X$ restricted to the first $\ell = \poly(k,d,m,\kappa(\V^*))$ columns. Then w.h.p. (in $\ell$), both $\|\Q\|_F^2 \leq 10 r \ell$ and $\sigma_{\min}(\Q) = \Omega(\frac{\sqrt{\ell}}{(\kappa(\V^*))^2}))$.
\end{lemma}
\begin{proof}
  The first bound $\|\Q\|_F^2 \leq 10 r \ell$ follows from the fact that the $\|\cdot\|_2^2$ norm of each row is distributed as a $\chi^2$ random variable, so the claim follows from standard tail bounds for such variables \cite{laurent2000}. 
  For the second claim, write $\Q = f(\W \overline{\X})$, where the rows of $\W$ are the $r$ distinct rows from the set $\{\V^*_{i,*},-\V^*_{i,*} \}_{i \in [k]}$ corresponding to the rows of $\Q$. Let $\W^+$ be the subset of rows of the form $\V^*_{i,*}$, and $\W^-$ its complement. 
  There now there is a rotation matrix $\RR$ that rotates $\W^+$ to be lower triangular, so that the $j$-th row of $\W^+\R$ is supported on the first $j$ columns. Let $\W$ be such that the pairs of rows $\{\V^*_{i,*},-\V^*_{i,*} \}$ with the same index $i$ are placed together. Then $\W\R$ is block-upper triangular, where the $j$-th pair of rows of the form $\{\V^*_{i,*},-\V^*_{i,*}\} $ are supported on the first $j$ columns. 
     Since Gaussians are rotationally invariant, $\W\RR \overline{\X}$ has the same distribution as $\W\overline{\X}$, thus we can assume that $\W$ is in this block lower triangular form, and $\V^*$ is in lower triangular form.

    WLOG assume the rank of $\W$ is $k$ (the following arguements will hold when the rank is $k' < k$). We now claim that we can write $\W_{r,*} = \alpha + \varphi e_k$, where $\alpha$ is in the span of $e_1,\dots,e_{k-1}$ and $\varphi = \Omega(\frac{1}{\kappa})$ where $\kappa = \kappa(\V^*)$. 
    To see this, note that if this were not the case, the projection of $\W_{r,*}$ onto the all prior rows with the same sign (i.e. all either of the form $\V_{i,*}^*$ or $-\V_{i,*}^*$) would be less than $\frac{1}{\kappa}$, since the prior span all of $\R^{k-1}$ on the first $k-1$ columns, and the only part of $\W_{r,*}$ outside of this span has weight $\varphi$. Let $w$ be such that $w\W'$ is this projection, where $\W'$ excludes the last row of $\W$ WLOG this row is of the form $\V_{i,*}^*$, and WLOG $i=k$. Then can write $w \W' = v (\V^*)'$ where $(\V^*)'$ excludes the last row of $\V^*$. 
    Then $\|[v,-1]\V^*\|_2 < \frac{1}{\kappa}$, and since $\|\V^*\|_2 \geq 1$ it follows that $\kappa(\V^*) > \kappa$, a contradiction since $\kappa$ is defined as the condition number of $\V^*$. 
    
     Now let $x^i$ be the $i$-th column of $\overline{\X}$, and let $\mathcal{E}_i$ be the event that $\varphi x_{k}^i > \lambda | \langle \alpha, (x_1^i,\dots,x_{k-1}^i)\rangle |$, where $\lambda$ will later be set to be a sufficiently large constant. Since $\varphi x_k^i$ and $ \langle \alpha, (x_1^i,\dots,x_{k-1}^i)\rangle$ are each distributed as a Gaussian with variance at most $1$, by anti-concentration of Gaussians $\pr{\mathcal{E}_i} = \Omega(\frac{1}{\lambda \kappa})$. So let $\mathcal{S} \subset [\ell]$ be the subset of $i$ such that $\mathcal{E}_i$ holds, and let $\X_\mathcal{S}$ be $\overline{\X}$ restricted to this subset. Let $\V_{i,*}^*$ be the last column of $\W$ (WLOG we assume it is $\V_{i,*}^*$ and not $-\V_{i,*}^*$). We now upper bound the norm of the projection of $f(\V_{i,*}^*\X_{\mathcal{S}})$ onto the row span of $f(\W\X_{\mathcal{S}})$. Now $f(-\V_{i,*}^*\X)$, if it exists as a row of $f(\W\X)$, will be identically $0$ on the coordinates in $\mathcal{S}$ (because $\V_{i,*}^* \X$ is positive on these coordinates by construction). So we can disregard it in the following analysis.  
     By construction of $\mathcal{S}$ we can write $$f(\V^*_{i,*} \X_{\mathcal{S}}) = f(\varphi(\X_{\mathcal{S}})_{k,*}) + b$$
     where $\|b\|_2 \leq  \|\frac{1}{\lambda}f(\V^*_{i,*} \X_{\mathcal{S}})\|_2$, where $(\X_{\mathcal{S}})_{k,*}$ is the $k$-th row of $\X_{\mathcal{S}}$. 
     By the triangle inequality, the projection of $f(\V_{i,*}^*\X_{\mathcal{S}})$ onto the rowspan of $f(\W'\X_{\mathcal{S}})$ (where $\W'$ is $\W$ excluding $\V_{i,*}^*$), is 
     $$\|f(\V^*_{i,*} \X_{\mathcal{S}})\P_{\W'} \|_2 \leq  \frac{1}{\lambda} \|f(\V^*_{i,*} \X_{\mathcal{S}})\|_2 +  \|f(\varphi (\X_{\mathcal{S}})_{k,*})\P_{\W'}\|_2$$ 
     where $\P_{\W'}$ is the projection onto the rowspan of $f(\W'\X_{\mathcal{S}})$. 
     Crucially, observe that $f(\W'\X_{\mathcal{S}})$, and thus $\P_{\W'}$, does not depend on the $k$-th row $(\X_{\mathcal{S}})_{k,*}$ of $\X_{\mathcal{S}}$. Now  $$\|f(\varphi (\X_{\mathcal{S}})_{k,*})\P_{\W'}\|_2 \leq \|f(\varphi (\X_{\mathcal{S}})_{k,*})\P_{\W'+1}\|_2$$ 
     where $\P_{\W'+1}$ is the projection onto the row span of $\{f(\W'\X_{\mathcal{S}})_{j,*} \}_{\text{rows } j \text{ of } \W'} \bigcup \{ \mathbf{1} \}$ where $\mathbf{1}$ is the all $1's$ vector. This holds since adding a vector to the span of the subspace being projected onto can only increase the length of the projection.  Let $\P_{1}$ be the projection just onto the row $\mathbf{1}$.

     Now observe that $\varphi f( (\X_{\mathcal{S}})_{k,*})(\mathbb{I} -\P_{1})$ is a mean $0$ i.i.d. shifted rectified-Gaussian vector with variance strictly less than $\varphi^2$ (here rectified means $0$ with prob $1/2$ and positive Gaussian otherwise). Moreover, the mean of the entries of $f(\varphi (\X_{\mathcal{S}})_{k,*})$ is $\Theta(\varphi)$. The $L_2$ of these vectors are thus sums of sub-exponential random variables, so by standard sub-exponential concentration (see e.g. \cite{wainwright2019high}) we have $\|f(\varphi (\X_{\mathcal{S}})_{k,*})(\mathbb{I} -\P_{1})\|_2 = \Theta(\varphi ) \sqrt{|S|} ,$ and moreover
     \begin{equation}\label{eqn:projlemma1}
     \|f(\varphi (\X_{\mathcal{S}})_{k,*})\|_2-\|f(\varphi (\X_{\mathcal{S}})_{k,*})(\mathbb{I} -\P_{1})\|_2 = \Omega(\varphi )\sqrt{|\mathcal{S}|} 
     \end{equation}
     w.h.p in $\log(|\mathcal{S}|)$ where $|S| \geq \Theta(1) \frac{\ell}{\kappa \lambda} = \poly(d,m,k, \kappa)$. Now by Lemma \ref{prop:subgaussian}, we have $$\|f(\varphi (\X_{\mathcal{S}})_{k,*})(\mathbb{I} -\P_{1}) \P_{\W'+1}\|_2 \leq \varphi \sqrt{(2k+1) /\delta}$$
     with probability $1-\delta$, for some $\delta = 1/\poly(k,d,m)$. So 
     $$\|f(\varphi (\X_{\mathcal{S}})_{k,*})(\mathbb{I} -\P_{1})\P_{\W'+1}\|_2 \leq O(\frac{\poly(k,d,m)}{\sqrt{|\mathcal{S}|}})\|f(\varphi (\X_{\mathcal{S}})_{k,*})(\mathbb{I} -\P_{1})\|_2$$ 
     Write $f(\varphi (\X_{\mathcal{S}})_{k,*}) =f(\varphi (\X_{\mathcal{S}})_{k,*})\P_{+1} + f(\varphi (\X_{\mathcal{S}})_{k,*}) (\mathbb{I} -\P_{1})$. 
Then by triangle inequality, we can upper bound $\| f(\varphi( \X_\mathcal{S})_{k,*}) \P_{\W'+ 1}\|_2$ by
\begin{equation*}
    \begin{split}
        & \leq \|f(\varphi (\X_{\mathcal{S}})_{k,*})\P_{1}\P_{\W'+1} \|_2    + \|f(\varphi (\X_{\mathcal{S}})_{k,*})(\mathbb{I} -\P_{1})\P_{\W'+1}\|_2 \\
        & \leq \|f(\varphi (\X_{\mathcal{S}})_{k,*})\P_{1} \|_2    + O(\frac{\poly(k,d,m)}{\sqrt{|S|}})\|f(\varphi (\X_{\mathcal{S}})_{k,*})(\mathbb{I} -\P_{1})\|_2 \\
        & = \Big(\|f(\varphi (\X_{\mathcal{S}})_{k,*}) \|_2^2 - \|f(\varphi (\X_{\mathcal{S}})_{k,*})(\mathbb{I} -\P_{1}) \|_2^2  \Big)^{1/2}  + O(\frac{\poly(k,d,m)}{\sqrt{|S|}})\|f(\varphi (\X_{\mathcal{S}})_{k,*})\|_2 
    \end{split}
\end{equation*}
     Using the bound from Equation \ref{eqn:projlemma1}, for some constants $c,c'<1$ bounded away from $1$, we have
     \begin{equation*}
         \begin{split}
             & =     \|f(\varphi (\X_{\mathcal{S}})_{k,*}) \|_2(1 - c) + O(\frac{\poly(k,d,m)}{\sqrt{|S|}})\|f(\varphi (\X_{\mathcal{S}})_{k,*})\|_2\\
           & \leq     \|f(\varphi (\X_{\mathcal{S}})_{k,*}) \|_2(1 - c')
         \end{split}
     \end{equation*}
           Thus $\|f(\V^*_{i,*} \X_{\mathcal{S}})\P_{\W'}\|_2 \leq \|f(\V^*_{i,*} \X_{\mathcal{S}})\P_{\W'+1}\|_2 \leq (1 -\Theta(1))\|f(\V^*_{i,*} \X_{\mathcal{S}})\|_2$. Now by setting $\lambda > 2c'$ a sufficently large constant, the $b$ term becomes negligible, and the above bound holds replacing $c'$ with $c'/2$. Since we have $\|f(\V^*_{i,*} \X_{\mathcal{S}})\|_2 = \Theta( \varphi \sqrt{|\mathcal{S}|})$, and $\|f(\V^*_{i,*} \overline{\X})\|_2 = \Theta(\sqrt{\ell})$, if $\overline{\P_{\W'}}$ is projection of  onto the rows of $f(\W' \overline{X})$, we have
           \begin{equation*}
               \begin{split}
                   \|f(\V^*_{i,*} \overline{\X})\overline{\P_{\W'}}\|_2 & \leq \|f(\V^*_{i,*} \overline{\X})\|_2\left(1 - \Theta\left(\frac{\varphi}{\kappa \lambda}\right)\right) \\ 
                   & <\|f(\V^*_{i,*} \overline{\X})\|_2(1 - \Theta(\frac{1}{\kappa^2 }))
               \end{split}
           \end{equation*}        
           Where here we recall $\lambda = \Theta(1)$. Since this argument used no facts about the row $i$ we were choosing, it follows that the norm of the projection of any row onto the subspace spanned by the others others in $f(\W \overline{\X})$ is at most a $(1 - \Theta(\frac{1}{\kappa^2 }))$ factor less than the norm was before the projection. In particular, this implies that $f(\W\overline{\X})$ is full rank. Note by sub-exponential concentration, each row norm is $\Theta(\sqrt{\ell})$ w.h.p.  We are now ready to complete the argument. Write $f(\W \overline{\X}) = \BB \Sigma \Q^T$ in its singular value decomposition.
 Since the projection of one row onto another does not change by a row rotation,, we can rotate $\Q^T$ to be the identity, and consider $\BB\Sigma$. Let $u_i$ be a unit vector in the direction of the $i$-th row projected onto the orthogonal space to the prior rows.  Now for any unit vector $u$, write it as $u = \sum_i u_i a_i$ (which we can do because $f(\W\overline{\X})$ is full rank). Noting that $\frac{\|f(\W \overline{\X})_{i,*}\|_2}{\|f(\W \X_{\mathcal{S}})_{i,*}\|_2} = O(\poly(k) \kappa^2)$ for any row $i$, we have 
 \begin{equation*}
     \begin{split}
         \|f(\W\overline{\X}) u\|_2^2  & \geq \sum_i \langle u_i ,u\rangle^2 \Omega(\frac{\ell}{\kappa^4 }) \\
         & \geq \sum_i a_i^2 \Omega(\frac{\ell}{\kappa^4 })\\
         & = \Omega(\frac{\ell}{\kappa^4})
     \end{split}
 \end{equation*}
 Thus $\sigma_{\min}(\Q) = \sigma_{\min}(f(\W\overline{\X})) =\Omega( \frac{\sqrt{\ell}}{\kappa^2})$ as needed, where recall we have been writing $\kappa = \kappa(\V^*)$. 
\end{proof}
Using the bounds developed within the proof of the prior lemma gives the following corollary.

\begin{corollary}\label{cor:projectnegV}
Let $ \P_{S_{i,+}}, \P_{S_{i,-}}$ be as in Algorithm \ref{alg:recoversigns2}. Then 
\[ \| f(\V_{i,*}^* \X) \P_{S_{i,-}} \|_2 = \|f(\V_{i,*}^* \X)\|_2 \Big( 1 - \Omega(\frac{1}{\kappa(\V^*)^2 \poly(k)})\Big)\]
and 
\[ \| f(-\V_{i,*}^* \X) \P_{S_{i,+}} \|_2 = \|f(-\V_{i,*}^* \X)\|_2 \Big( 1 - \Omega(\frac{1}{\kappa(\V^*)^2 \poly(k)})\Big)\]
\end{corollary}

\begin{theorem}\label{thm:recoversignsnoisy}

Let $\AA = \U^*f(\V^*\X) + \E$, where $\E$ is i.i.d. mean zero with variance $\sigma^2$. Then given $v^T_i \in \R^d$ such that. $\|v_i - \xi_i\V_{i,*}^*\|_2 \leq \eps$ for some unknown $\xi_i \in \{ 1,-1\}$ and $i=1,2,\dots,k$, where $\eps = O(\frac{1}{\poly(d,m,k,\kappa(\U^*),\kappa(\V^*)})$ is sufficiently small, with high probability, Algorithm \ref{alg:recoversigns2} returns $\V$ such that $\|\V - \V^*\|_2 \leq \eps$ in $\poly(d,m,k,\kappa(\U^*),\kappa(\V^*),\sigma)$ time.
\end{theorem}
\begin{proof}
  Consider a fixed $i \in [k]$, and WLOG assume $\xi_i = 1$, so we have $\|v_i - \V_{i,^*}^*\|_2 \leq \eps =  O(\frac{1}{\poly(d,m,k,\kappa(\U^*),\kappa(\V^*)})$. We show $a^+_i < a^-_i$ with high probability. Now fix a row $j$ of $\overline{\AA} = \U^*f(\V^* \overline{\X}) + \overline{\E}$  as in Algorithm \ref{alg:recoversigns2}, where $\overline{\Q}$ refers to restricting to the first $\ell = \poly(k,d,m,\kappa(\U^*), \allowbreak \kappa(\V^*),\sigma)$ columns of a matrix $\Q$. Note that we choose $\eps$ so that $\eps < 1/\poly(\ell)$ (which is achieved by taking $n$ sufficiently large). This row is given by $\overline{\AA}_{j,*} = \U^*_{j,*} f(\V^*\overline{\X}) + \overline{\E}_{j,*}$. As in the proof of Theorem \ref{thm:recoverexactsigns}, using Lemma \ref{prop:gaussiancondition} to bound $\|\overline{\X}\|_2$, we have $\|f(v_i\overline{\X}) - f(\xi \V_{i,*}^* \overline{\X})\|_2 \leq O(\eps \sqrt{\ell})$. We now using Lemma \ref{prop:projectionbound} to bound the projection difference between using approximate projection matrix $\P_{S_{i,+}}$ formed by our approximate vectors $f(v_i\overline{\X})$, and the true projection matrix  $\P_{S_{i,+}}^*$ formed by the vectors $f(\V^*_{i,*}\overline{\X})$. By Lemma \ref{prop:reallylongprojectionprop}, the 
  condition number and the spectral norm of the matrix formed by the rows that span $\P_{S_{i,+}}^*$ are at most $O( r \poly(k) \kappa^2)$ and $O(r\ell) = \poly(k,d,m,\kappa(\U^*),\kappa(\V^*),\sigma)$ respectively, w.h.p. (in $k,d,m$). Setting $\eps = \eps' /\poly(k,d,m,\kappa(\U^*),\kappa(\V^*),\sigma)$ sufficiently small,  Lemma \ref{prop:projectionbound} gives $\|x\P_{S_{i,+}}\|_2 = \|x\P_{S_{i,+}}^*\|_2 \pm O(\eps' \|x\|_2) $ for $\eps' = \frac{1}{\poly(k,d,m,\kappa(\V^*),\kappa(\U^*),\sigma}$ and any vector $x$. 
  
  Now we have
  $a_{i,j}^+ =( \|\overline{\E_{j,*}}(\mathbb{I} - \P_{S_{i,+}})\|_2 \pm  O(\eps' \sigma \sqrt{\ell}))^2 \leq \|\overline{ \E_{j,*}} \|_2^2 +  \pm  O(\sigma^2 \eps' \ell \poly(k,d,m))$. Here we used that $\|\E\|_2^2 \leq \sigma^2 \ell \poly(k,d,m)$, w.h.p. in $k,d,m$ (by Chebyshev's inequality), and the fact that w.h.p. we have $\|(\U_{j,i}^*f(\V_{i,*}^*\overline{\X})  + \overline{\E_{j,*}})\|_2 = O(\sigma \sqrt{\ell})$. Then, setting $\eps' = \eps''/\poly(\ell)$, we have
  \begin{equation*}
      \begin{split}
      (a_{i,j}^-)^2  & = \|(\U_{j,i}^*f(\V_{i,*}^*\overline{\X})  + \overline{\E_{j,*}})(\mathbb{I} - \P_{S_{i,-}})\|_2^2 \pm O(\eps'') \\
          & = \|(\U_{j,i}^*f(\V_{i,*}^*\overline{\X})  + \overline{\E_{j,*}})\|_2^2 -\|(\U_{j,i}^*f(\V_{i,*}^*\overline{\X})  + \overline{\E_{j,*}}) \P_{S_{i,-}}\|_2^2 \pm  O(\eps'')\\
   & \geq \|(\U_{j,i}^*f(\V_{i,*}^*\overline{\X})  + \overline{\E_{j,*}})\|_2^2 - \Big(\|(\U_{j,i}^*f(\V_{i,*}^*\overline{\X}) \P_{S_{i,-}}\|_2 +  \|\|\overline{\E_{j,*}}\P_{S_{i,-}}\|_2       \Big)^2 \pm O(\eps'') 
      \end{split}
  \end{equation*}
  \noindent
  where the second equality follows  by the Pythagorean Theorem.
  Applying Lemma \ref{prop:subgaussian} and Corollary \ref{cor:projectnegV}, writing $\kappa = \kappa(\V^*)$, with probability $1-\delta$ we have
    \begin{equation*}
    \begin{split}
        (a_{i,j}^-)^2   \geq & \|(\U_{j,i}^*f(\V_{i,*}^*\overline{\X})  +\overline{ \E_{j,*}})\|_2^2 \\ 
        & - \Big(\|(\U_{j,i}^*f(\V_{i,*}^*\overline{\X})\|_2 ( 1 - \Omega(\frac{1}{\kappa^2 \poly(k)})) + 2 \sigma\sqrt{ k /\delta} \Big)^2 \pm   O(\eps'')
    \end{split}
    \end{equation*}
  \noindent  Setting $\delta < 1/\poly(k,d,m)$ to get high probability gives
  \begin{equation*}
    \begin{split}
        (a_{i,j}^-)^2    &\geq   \|(\U_{j,i}^*f(\V_{i,*}^*\overline{\X})  + \overline{\E_{j,*}})\|_2^2 - \|(\U_{j,i}^*f(\V_{i,*}^*\overline{\X})\|_2^2 ( 1 -\Omega(\frac{1}{\kappa^2 \poly(k)})) \\
        & - 4 \sigma\sqrt{ \frac{k} {\delta}} \|(\U_{j,i}^*f(\V_{i,*}^*\overline{\X})\|_2 \pm  O(\eps' \sqrt{\ell}\sigma  + \frac{\sigma^2 k}{\delta})\\
        & =  \Omega(\frac{1}{\kappa^2 \poly(k)}) \|\U_{j,i}^*f(\V_{i,*}^*\overline{\X}) \|_2^2 + \|\overline{\E_{j,*}}\|_2^2 + 2\langle \U_{j,i}^*f(\V_{i,*}^*\overline{\X}) ,  \overline{\E_{j,*}} \rangle \\
        & -  4 \sigma\sqrt{ k /\delta}\|(\U_{j,i}^*f(\V_{i,*}^*\overline{\X})\|_2  \pm O(\eps' \sqrt{\ell}\sigma + \frac{\sigma^2 k}{\delta}) 
    \end{split}
    \end{equation*}
   Thus,
   \begin{equation*}
       \begin{split}
           a^-_i - a^+_i > & \sum_{j \in [m]} \Omega(\frac{1}{\kappa^2 \poly(k)}) \|\U_{j,i}^*f(\V_{i,*}^*\overline{\X}) \|_2^2 +  2\langle \U_{j,i}^*f(\V_{i,*}^*\overline{\X}) , \overline{ \E_{j,*}} \rangle \\
           & -  4 \sigma\sqrt{ k /\delta}\|(\U_{j,i}^*f(\V_{i,*}^*\overline{\X})\|_2 \pm  O(\eps' \sqrt{\ell}\sigma + \frac{\sigma^2 k}{\delta} )
       \end{split}
   \end{equation*}
   By Chebyshev's inequality, we have, $|2\langle \U_{j,i}^*f(\V_{i,*}^*\overline{\X}) ,  \overline{\E_{j,*}} \rangle |  <\poly(dkm) \sigma \| \U_{j,i}^*f(\V_{i,*}^*\overline{\X})\|_2$ w.h.p. in $d,k,m$.  Thus $\sum_j |2\langle \U_{j,i}^*f(\V_{i,*}^*\overline{\X}) ,  \overline{\E_{j,*}} \rangle | < \poly(dkm) \sigma \| \U_{j,i}^*f(\V_{i,*}^*\overline{\X})\|_2 $. Now $\|\U_{*,i}^*\|_2 > 1/\kappa(\U^*)$, otherwise we would have $\|\U^* e_i\|_2 < 1/\kappa(\U^*)$, which is impossible by definition. Thus there is at least one entry of $\U^*_{*,i}$ with magnitude at least $1/(m\kappa(\U^*))$. So 
   \begin{equation*}
       \begin{split}
           \sum_j \|\U_{j,i}^*f(\V_{i,*}^*\overline{\X}) \|_2^2 & \geq \frac{1}{m \kappa(\U^*)}\|f(\V_{i,*}^*\overline{\X}) \|_2^2 \\
           & = \Omega(\ell \frac{1}{m\kappa(\U^*)})
       \end{split}
   \end{equation*}
   where the last bound follows via bounds on $\chi^2$ variables \cite{laurent2000}.
   
   The above paragraph also demonstrates that $\frac{ |2\langle \U_{j,i}^*f(\V_{i,*}^*\overline{\X}) ,  \overline{\E_{j,*}} \rangle |}{\|\U_{j,i}^*f(\V_{i,*}^*\overline{\X}) \|_2^2} \leq \frac{ \poly(dkm) \sigma }{\sqrt{\ell} }$, so taking $\ell$ sufficiently large this is less than $1/2$. Thus 
   \begin{equation*}
       \begin{split}
            \sum_{j \in [m]} \Omega(\frac{1}{\kappa^2 \poly(k)}) \|\U_{j,i}^*f(\V_{i,*}^*\overline{\X}) \|_2^2 +  2\langle \U_{j,i}^*f(\V_{i,*}^*\overline{\X}) , \overline{ \E_{j,*}} \rangle > \frac{1}{2}  \sum_{j \in [m]} \Omega(\frac{1}{\kappa^2 \poly(k)}) \|\U_{j,i}^*f(\V_{i,*}^*\overline{\X}) \|_2^2
       \end{split}
   \end{equation*}
   and we are left with 
    \begin{equation*}
       \begin{split}
           a^-_i - a^+_i & > \sum_{j \in [m]} \Omega(\frac{1}{\kappa^2 \poly(k)}) \|\U_{j,i}^*f(\V_{i,*}^*\overline{\X}) \|_2^2 -  4 \sigma\sqrt{ k /\delta}\|(\U_{j,i}^*f(\V_{i,*}^*\overline{\X})\|_2 -  O(\eps' \sqrt{\ell}\sigma + \frac{\sigma^2 k}{\delta})\\
&  \geq \Omega(\ell \frac{1}{m\kappa^2 \kappa(\U^*) \poly(k)}) -  O( \sigma\sqrt{ k \ell /\delta} + \eps' \sqrt{\ell}\sigma+ \frac{\sigma^2 k}{\delta} )
       \end{split}
   \end{equation*}
   Taking $\delta = 1/\poly(d,k,m)$ as before and $\ell$ sufficiently larger than $1/\delta^2$, the above becomes $ a^-_i - a^+_i =  \Omega(\ell \frac{1}{m\kappa^2 \kappa(\U^*) \poly(k)}) = \omega(1)$ w.h.p. in $d,k,m$. Thus the algorithm correctly determines $\xi_i = 1$ after seeing $a^-_i > a^+_i$, and the analysis is symmetric in the case that $xi_i = -1$.

\end{proof}
\subsection{Recovering the Weights $\U^*,\V^*$}

We have now shown in the prior section that given approximate $\V^*_{i,*}$'s where the signs $\xi_i$ are unknown, we can recover the signs exactly in polynomial time in the noisy setting. Thus we recover $\V$ such that $\|\V - \V^*_{i,*}\|_2 \leq \eps$ for some polynomially small $\eps$. To complete our algorithm, we simply find $\U^*$ by solving the linear regression problem
\[      \min_{\U} \|\U f(\V\X) - \AA\|_F \]
It is well know that the solution to the above regression problem can be solved by computing the pseduoinverse of the matrix $f(\V\X)$, thus the entire procedure can be carried out in polynomial time. The following lemma states that, even in the presence of noise, the solution $\U$ from the regression problem must in fact be very close to the true solution $\U^*$.

\begin{lemma}\label{lem:approxregression}
Let $\overline{\X}$ be the first  $\ell = \poly(k,d,m,\kappa(\U^*),\kappa(\V^*),\sigma)$ columns of $\X$, and similarly define $\overline{\E}$. Set $\overline{\AA} = \U^*f(\V^*\overline{\X}) + \overline{\E}$. Then given $\V$ such that $\|\V - \V^*\|_2 \leq \eps$ where $\eps = \frac{\eps'}{\poly(\ell)}$ for some $\eps'>0$, 
if $\U$ is the regression solution to $\min_{\U} \|\U f(\V \overline{\X}) - \overline{\AA}\|_2$, then $\|\U - \U^*\|_2 < O(\eps')$ with high probability in $m,d$. 
\end{lemma}
\begin{proof}
Note $\|f(\V\overline{\X}) - f(\V^*\overline{\X})\|_2 \leq  O(1) \sqrt{\ell} \eps$ by standard operator norm bounds on Gaussian matrices (see Lemma \ref{prop:gaussiancondition}), and the fact that $|a - b| > |f(a) - f(b)|$ for any $a,b \in \R$. 
  The regression problem is solved row by row, so fix a row $i \in [m]$ and consider $\min_u \|uf(\V\overline{\X}) - \overline{\AA}_{i,*}\|_2 = \min_u \|u f(\V\overline{\X}) - (\U_{i,*}^* f(\V^*\overline{\X}) + \overline{\E}_{i,*})\|_2 =  \|u f(\V\overline{\X}) - (\U_{i,*}^* f(\V \overline{\X})  + \overline{\E_{i,*}} + \U_{i,*}^*\ZZ)\|_2$, where $\ZZ$ is a matrix such that $\|\ZZ\|_F \leq  O(1) \sqrt{\ell} \eps$. Now by the normal equations\footnote{See \url{https://en.wikipedia.org/wiki/Linear_least_squares} }, if $u^*$ is the above optimizer, we have 
  \begin{equation*}
      \begin{split}
           u^* & = \Big(  \U_{i,*}^* f(\V \overline{\X}) + \overline{\E_{i,*}}+ \U_{i,*}^*\ZZ\Big) f(\V \overline{\X})^T \big[f(\V \overline{\X}) f(\V \overline{\X})^T  \big]^{-1}  \\
  & = \U_{i,*}^*  + \Big(\overline{\E_{i,*}} +\U_{i,*}^*\ZZ\Big) f(\V \overline{\X})^T \big[f(\V \overline{\X}) f(\V \overline{\X})^T  \big]^{-1}  \\
  & = \U_{i,*}^*  +  \overline{\E_{i,*}}f(\V \overline{\X})^T \big[f(\V \overline{\X}) f(\V \overline{\X})^T  \big]^{-1} + \U_{i,*}^*\ZZ' \\
      \end{split}
  \end{equation*}
 
 Where $\ZZ'$ is a matrix such that $\|\ZZ'\|_F = O\big( \frac{\sqrt{\ell} \eps \kappa(f(\V\overline{\X})}{\sigma_{\min}(f(\V\overline{\X}))} \big)$. Note that we can scale $\overline{\AA}$ at the beginning so that no entry is larger than $\ell^2$, which implies w.h.p. that each row of $\U^*$ has norm at most $\ell^2$. Thus \[\|\U_{i,*}\ZZ'\|_F = O\big( \frac{\ell^{5/2} \eps \kappa(f(\V\overline{\X})}{\sigma_{\min}(f(\V\overline{\X}))} \big)\]
 Now note $\ex{\|\overline{\E_{i,*}}f(\V \overline{\X})^T \big[f(\V \overline{\X}) f(\V \overline{\X})^T  \big]^{-1}\|_2^2} = O( \sqrt{\ell k} \frac{\sigma^2}{ \sigma^2_{\min}(f(\V\overline{\X}) ) } )$ using $\|f(\V\overline{\X}\|_2 = O(\sqrt{\ell k})$ by the same operator norm bounds as before. By Markov bounds, w.h.p. in $m,d,$ we have \[\|\overline{\E_{i,*}}f(\V \overline{\X})^T \big[f(\V \overline{\X})   f(\V \overline{\X})^T   \big]^{-1} \|_2^2= O( \sqrt{\ell }\poly( d,m) \frac{\sigma^2}{ \sigma^2_{\min}(f(\V\overline{\X}) ) } )\]  Now by Courant-Fischer theorem and application of the triangle inequality, we have $\sigma_{\min}(f(\V\overline{\X})) > \sigma_{\min}(f(\V^*\overline{\X})) - O(\sqrt{\ell} \eps)$, and by Lemma \ref{prop:reallylongprojectionprop} (see Section \ref{subsec:noiseysigns}), we have $ \sigma_{\min}(f(\V^*\overline{\X})) = \Omega(\frac{\sqrt{\ell}}{\kappa(\V^*) \poly(k)})$, thus for $\ell$ sufficiently large we obtain $\sigma_{\min}(f(\V\overline{\X})) = \Omega(\frac{\sqrt{\ell}}{\kappa(\V^*) \poly(k)})$, in which case we have
\[\|\overline{\E_{i,*}}f(\V \overline{\X})^T \big[f(\V \overline{\X})   f(\V \overline{\X})^T   \big]^{-1}\|_2^2= O\Big( \poly( d,m) \frac{\sigma^2}{ \sqrt{\ell } } \Big)\]
Setting $\ell,1/\eps$ to be sufficiently large polynomials in $(d,m,k,\kappa(\U^*),\kappa(\V^*),\sigma)$, we obtain 
\[ \|u^* - \U^*_{i,*}\|_2 \leq O(\eps'/\sqrt{m}) \]
from which the Lemma follows.
\end{proof}

We now state our main theorem for recovery of the weight matrices in the noisey case.
\begin{theorem}\label{thm:noisyfinal}
Let $\AA = \U^*f(\V^*\X) + \E$ be given, where $\U^* \in \R^{m \times k},\V^*\in \R^{k \times d}$ are rank-$k$ and $\E$ is a matrix of i.i.d. mean zero subgaussian random variables with variance $\sigma^2$. Then given $n = \Omega\Big(  \poly\big( d,m,\kappa(\U^*),\kappa(\V^*) , \sigma, \frac{1}{\eps}\big) \Big)$, there is an algorithm that runs in $\poly(n)$ time and w.h.p. outputs $\V,\U$ such that 
\[ \|\U - \U^*\|_F \leq \eps \; \; \; \; \; \|\V - \V^* \|_F \leq \eps\]
\end{theorem}
\begin{proof}
  The proof of correctness of the Tensor Decomposition based approximate recovery of $\V^*$ up to the signs is the same as in the exact case, via Theorem \ref{thm:anandkumar}. By Theorem \ref{thm:recoversignsnoisy}, we can recover the signs $\xi_i$, and thus recover $\V$ so that $\|\V - \V^* \|_F \leq \eps$. Observe that while the results in Section \ref{subsec:noiseysigns} were stated for $\eps = \Theta\left(\frac{1}{\poly(d,m,\kappa(\U^*),\kappa(\V^*) \sigma)}\right)$, they can naturally be generalized to any $\eps$ which is at least this small by increasing $n$ by a $\poly(1/\eps)$ factor before running the tensor decomposition algorithm. Then by Lemma \ref{lem:approxregression}, we can recover $\U$ in polynomial time such that $\|\U - \U^*\|_F \leq \eps$ as desired, which completes the proof.
\end{proof}

\begin{remark}
As in Remark \ref{remark:normalize}, we have implicitly normalized the entire matrix $\AA$ so that the columns of $\AA$ have at most unit norm. If one seeks bounds for the recovery of the \textit{unnormalized} $\U^*$, the error becomes $\|\U - \U^*\|_F \leq \eps \|\U^*\|_2$. To see why this holds, note that the normalization factor of Remark \ref{remark:normalize} is at least $\Omega(\frac{1}{\|\U^*\|_2 + \sqrt{m\log(\ell)}})$, where $\ell = \poly(d,m,\kappa(\U^*),\kappa(\V^*),\sigma)$ is as in Section \ref{subsec:noiseysigns}, and $O(\sqrt{m\log(\ell)})$ is a bound on the max column norm of $\E$ by subgaussian concentration. Thus multiplying by the inverse of this normalization factor blows up the error to $\|\U^* \|_2\eps$ after scaling $\eps$ down by a polynomial factor.
\end{remark}

\section{A Fixed-Parameter Tractable Exact Algorithm for Arbitrary $\U^*$}
\label{sec:FPT}

In the prior sections, we required that $\U^* \in \R^{m \times k}$ have rank $k$ in order to recover it properly. Of course, this is a natural assumption, making $\U^*$ identifiable. In this section, however, we show that even when $m<k$ and $\U^*$ does not have full column rank, we can still recover $\U^*\V^*$ exactly in the noiseless case where we are given $\AA = \U^*f(\V^*\X)$ and $\X$, as long as the no two columns of $\U^*$ are non-negative scalar multiples of each other. Observe that this excludes columns from being entirely zero, but allows for columns of the form $[u,-u]$ for for $u \in \R^m$, as long as $u$ is non-zero. Our algorithm requires $n = \poly(d,k) \kappa^{\Omega(k)}$,  samples, and runs in time $O(n  \poly(d,m,k))$. Here $\kappa = \kappa(\V^*)$ is the condition number of $\V^*$. Our algorithm does not have any dependency on the condition number of $\U^*$. 

\begin{Frame}[\textbf{Algorithm \ref{alg:overall_exact_fpt_real} : FPTExactNeuralNet}$(\V', \X, \mathcal{S})$.]
\label{alg:overall_exact_fpt_real}
\ttx{Input:} Matrices  $\AA = \U^*f(\V^*\X) \in \R^{d \times n}$ and $\X \in \R^{r \times n}$ such that each entry in $\X \sim \mathcal{N}(0, 1)$. \\
\begin{enumerate}
    \item Find a subset $S$ of columns of non-zero columns of $\AA$ such that each for each $i \in S$ there is a $j \in S $, $j \neq i$, with $\AA_{*,i} = c \AA_{*,j}$ for some $c > 0$. 
    \item Partition $S$ into $S_r$ for $r \in [k]$ such that for each pair $i,j \in S_r$, $i \neq j$, we have $\AA_{*,i} = c \AA_{*,j}$ for some $c \in \R^{\neq 0}$. 
    \item For each $i \in [k]$, choose a representative $j_i \in S_i$, and let $\U_{*,i} = \AA_{*,j_i}$. For each $j \in S_i$, let $c_{i,j}$ be such that $ c_{i,j} \U_{*,i} = \AA_{*,j} $.
    
    \item  let $\W$ be the matrix where the $i$-th row is given by the solution $w_i$ to the following linear system:
    \begin{equation*}
        \begin{split}
            \forall i \in [k]: \;\;\;\;   
            & w_{i}\X_{*,j} = c_{i,j}  \;\;\;\; \textrm{ if } j  \in S_i \\
        \end{split}
    \end{equation*}
    \item Set $\V_{i,*} = \W_{i,*}/\|\W_{i,*}\|_2$, and let $\U$ be the solution to the following linear system:
    \begin{equation*}
        \begin{split}
            \U f(\V \X) = \AA
        \end{split}
    \end{equation*}
\end{enumerate}
\ttx{Output:}  $(\U, \V)$.
\end{Frame}

The runtime of our algorithm is polynomial in the sample complexity $n$ and the size of the networks $d,k$, but simply requires $\poly(d,k) \kappa^{\Omega(k)}$ samples in order to obtain columns of $f(\V^*\X)$ which are $1$-sparse, in which case the corresponding column of $\AA$ will be precisely a positive scaling of a column of $\U^*$. In this way, we are able to progressively recover each column of $\U^*$ simply by finding columns of $\AA$ which are scalar multiples of each other. The full algorithm, Algorithm \ref{alg:overall_exact_fpt_real}, is given formally below. 

\begin{lemma}\label{lem:kappaksigns}
For each $i \in [k]$, with probability $1-\delta$, at least $d$ columns of $f(\V^* \overline{\X})$  are positive scalings of $e_i^T$, where $\overline{\X}$ is the first $\ell$ columns of $\X$ for $n  = \Omega(d \log(k/\delta)\kappa^{O(k)})$. In other words, $|S_i| \geq d$. 
\end{lemma}
\begin{proof} Let $\kappa = \kappa(\V^*)$. 
  As in the proof of Lemma \ref{prop:reallylongprojectionprop}, we can assume that $\V^*$ is lower triangular by rotating the rows by a matrix $\RR$, and noting that $\RR \X$ has the same distribution as $\X$ by the rotational invariance of Gaussians.
  We now claim that $\pr{\|\V^*g \|_2 < \frac{1}{k \kappa}} = \Omega((\frac{1}{k \kappa})^k)$, where $g \sim \mathcal{N}(0,\mathbb{I}_d)$ is a Gaussian vector. To see this, since $\V^*$ is rank $k$ and in lower triangular form, $\V^*$ is supported on its the first $k$ columns. Thus it suffices to compute the value $\|\V^* g\|_2$ were $g \in \R^k$ is a $k$-dimensional Gaussian. By the anti-concentration of Gaussian, each $g_i < \frac{1}{k\kappa}$ with probability at least $\Omega(1/(k\kappa))$. Since the entries are independent, it follows that $\pr{\|g\|_2 \leq \frac{1}{\sqrt{k} \kappa }} = \Omega(1/(k\kappa)^k)$. Let $\mathcal{E}_1$ be the event that this occurs. Since $\V^*$ has unit norm rows, it follows by Cauchy-Schwartz that conditioned on $\mathcal{E}_1$, we have $\tilde{g} = \V^* g$ satisfies $\|\tilde{g}\|_2 = O( \frac{1}{\kappa})$

  Now consider the pdf of the $k$-dimensional multivariate Gaussian $\tilde{g}$ that has covariance $\Sigma = \V^* (\V^*)^T$, which is given by
 \[ p(x) = \frac{\exp\big( - \frac{1}{2} x \Sigma^{-1} x \big)}{\sqrt{(2 \pi)^k \det(\Sigma) }} \]
 for $x \in \R^k$. Now condition on the event $\mathcal{E}_2$ that $\tilde{g}$ is contained within the ball $\mathcal{B}$ of radius $O(\frac{1}{ \kappa })$ centered at $0$. Since $\mathcal{E}_1$ implies $\mathcal{E}_2$, we have $\pr{\mathcal{E}_2}=\Omega(1/(k\kappa)^k)$ Now the eigenvalues of $\Sigma$ are the squares of the singular values of $\V^*$, which are all between $1/\kappa$ and $\sqrt{k}$. So  all eigenvalues of $\Sigma^{-1}$ are between $1/k$ and $\kappa^2$. Thus for all $x \in \mathcal{B}$, we have 
 \[\frac{1}{2} \leq  \frac{1}{e^{1/2}} \leq \exp\big( - \frac{1}{2} x \Sigma^{-1} x \big) \leq 1  \]
  It follows that
  \[ \sup_{x,y \in \mathcal{B}} \frac{p(x)}{p(y)} \leq 2\] 
  Now let $\mathcal{O}_1,\mathcal{O}_2,\dots,\mathcal{O}_{2^k}$ be the intersection of all $2^k$ orthants in $\R^k$ with $\mathcal{B}$. The above bound implies that 
  \[  \max_{i,j \in [2^k]} \frac{\int_{\mathcal{O}_i} p(x) dx} { \int_{\mathcal{O}_j} p(y) dy } \leq 2 \]
  Thus conditioned on $\mathcal{E}_2$ for the i.i.d. gaussian vector $\tilde{g}\sim \mathcal{N}(0,\Sigma) \in \R^k$, the probability that $\tilde{g}$ is in a given $\mathcal{O}_i$ is at most twice the probability that $g$ is in $\mathcal{O}_j$ for any other $j$. Thus $ \min_{i \in [2^k] }\pr{\tilde{g} \in \mathcal{O}_i} > \frac{1}{2^{k+1}}$. Thus for any sign pattern $\mathcal{S}$ on $k$-dimensional vectors, and in particular for the sign partner $\mathcal{S}_i$ of $e_i$, the probability that $\tilde{g}$ has this sign pattern conditioned on $\mathcal{E}_2$ is at least $\frac{1}{2^{k+1}}$. Since $\pr{\mathcal{E}_2}=\Omega(1/(k\kappa)^k)$, it follows that in $n = \Omega( d \log(k/\delta) (k \kappa)^{2k})$ repetitions, a scaling of $e_i$ will appear at least $d$ times in the columns of $f(\V^*\X)$ with probability $1 - \delta/k$, and the Lemma follows by a union bound over $e_i$ for $i \in [k]$.

\end{proof}

\begin{theorem}\label{thm:fptfinal}
Suppose $\AA = \U^*f(\V^*\X)$ for $\U^* \in \R^{m \times k}$ for any $m \geq 1$ such that no two columns of $\U^*$ are non-negative scalar multiples of each other, and $\V^* \in \R^{k \times n}$ has rank$(\V^* ) = k$, and $n > \kappa^{O(k)} \poly(dkm)$. Then
Algorithm \ref{alg:overall_exact_fpt_real} recovers $\U^*,\V^*$ exactly with high probability in time $\kappa^{O(k)} \poly(d,k,m)$.
\end{theorem}
\begin{proof}
  By Lemma \ref{lem:kappaksigns}, at least $d$ columns of $f(\V^* \X)$ will be scalar multiples of $e_i$  for each $i$. Thus the set $S$ of indices, as defined in Step $2$ of Algorithm \ref{alg:overall_exact_fpt_real}, will contain each column of $\U^*$ as a column. It suffices to show that no two columns of $\overline{\AA}$ can be scalar multiples of each other if they are not a scalar multiple of $\U^*$. To see this, if two columns of $f(\V^*\X)$ were not $1$-sparse, then the distribution of $\U^* f(\V^*\X)$ on these columns is supported on a $t$-dimensional manifold living inside $\R^m$, for some $t \geq 2$. In particular, this manifold is the conic hull of at least $t' \geq t \geq 2$ columns of $\U^*$ (where $t'$ is the sparsity of the columns of $f(\V^*\X)$. This follows from the fact that the conic hull of any subset of $2$ columns of $\U^*$ is $2$-dimensional, since no columns two of $\U^*$ are non-negative scalings of each other. Thus the probability that two draws from such a distribution lie within the same $1$-dimensional subspace, which has measure $0$ inside of any $t \geq 2$-dimensional conic hull, is therefore $0$, which completes the claim.

  To complete the proof of the theorem, by pulling a diagonal matix $\D$ through $f$, we can assume $\U = \U^*$. By construction then, $c_{i,j}$ is such that $(\D \V_{i,*} \X_{*,j}) = c_{i,j}$, as it is the scaling which takes $\U_{*,i}$ to $\AA_{*,j}$. Thus $w_i$, as defined in step $4$ of Algorithm \ref{alg:overall_exact_fpt_real}, is the solution to a linear equation $w_i \X_{S_i} = c$ for some fixed vector $c$, where $\X_{S_i}$ is $\X$ restricted to the columns in $S_i$. Since $|S_i| \geq d$ by Lemma \ref{lem:kappaksigns}, to show that $w_i$ is unique it suffices for $\X_{S_i}$ to be full rank. But as argued in the proof of Theorem \ref{thm:recoverVgivenU}, any subset of $d$ columns of $\X$ will be rank $d$ and invertible with probability $1$. Thus $w_i$ is unique, and must therefore be a scaling of $\V_{i,*}^*$, which we find by normalizing $w_i$ to have unit norm. After this normalization, we can renormalize $\U$, or simply solve a linear system for $\U$ as in Step 5 of Algorithm \ref{alg:overall_exact_fpt_real}. By  Lemma \ref{prop:random}, $f(\V^*\X)$ will have full rank w.h.p., so the resulting $\U$ will be unique and therefore equal to $\U^*$ as desired.
  
\end{proof}

%% file: 1bitCompressed.tex
\section{A Fixed-Parameter Tractable Algorithm for Arbitrary Non-Adversarial Noise}\label{sec:1bit}

In the noisy model, the observed matrix $\AA$ is generated by a perturbation $\E$ of some neural network $\U^* f(\V^*\X)$ with rank $k$ matrices $\U^* \in \R^{m \times k}, \V \in \R^{k \times d}$, and i.i.d. Gaussian $\mathcal{N}(0,1)$ input $\X \in \R^{d \times n}$. Formally, we are given as input $\X$ and $\AA= \U^* f(\V\X) + \E$, which is a noisy observation of the underlying network $\U^* f(\V\X)$, and tasked with recovering approximations to this network.
In Section \ref{sec:noisycase}, we showed that approximate recovery of the weight matrices $\U^*,\V^*$ is possible in polynomial time when the matrix $\E$ was i.i.d. mean $0$ and sub-Gaussian. In this section, we generalize our noise model substantially to include all error matrices $\E$ which \textit{do not depend} on the input matrix $\X$. 
 Our goal is then to obtain $\U,\V$ such that 
\[ \|\U f(\V\X) - \AA \|_F \leq (1+\eps)\|\E\|_F\]
\noindent
Thus we would like to be able to recover a good approximation to the observed input, where we are competing against the cost $\text{OPT} = \|\AA - \U^* f(\V^*\X)\|_2 = \|\E\|_2$. Observe that this is a slightly different objective than before, where our goal was to recover the actual weights $\U^*,\V^*$ approximately. This is a product of the more general noise model we consider in this Section. The loss function here can be thought of as recovering $\U,\V$ which approximate the observed classification nearly as well as the optimal generative $\U^*,\V^*$ do. This is more similar to the empirical loss considered in other words \cite{arora2016understanding}. 
The main result of this section is the development of a fixed parameter tractable algorithm which returns $\U,\V$ such that 

\begin{equation}\label{eqn:fptbound}
     \|\AA - \U f(\V\X)\|_F \leq \|\E \|_F +  O\Big( \Big[\sigma_{\min} \eps \sqrt{n m}  \|\E\|_2\Big]^{1/2} \Big)
\end{equation} 
Where $\sigma_{\max} = \sigma_{\max}(\U^*)$, and $\|\E\|_2$ is the spectral norm of $\E$. In this section, to avoid clustering, we will write $\sigma_{\max},\sigma_{\min}$, and $\kappa$ to denote the singular values and condition number of $\U^*$. Our algorithm has no dependency on the condition number of $\V^*$. The runtime of our algorithm is  $(\frac{\kappa}{\eps})^{O(k^2)} \poly(n, r,d)$, which is fixed-parameter tractable in $k,\kappa, \frac{1}{\eps}$. Here the sample complexity $n$ satisfies $n = \Omega(\poly(r,d,\kappa, \frac{1}{\eps}))$.  

We remark that the above bound in Equation \ref{eqn:fptbound} may at first seem difficult to parse. Intuitively, this bound will be a $(1+\eps)$ multiplicative approximation whenever the Frobenius norm of $\E$ is roughly an $\sqrt{m}$ factor larger than the spectral norm--in other words, when the error $\E$ is relatively flat. Note that these bounds will hold when $\E$ is drawn from a very wide class of random matrices, including matrices with heavier tails (see \cite{vershynin2010introduction} and discussion below). When this is not the case, and $\|\E\|_2 \approx \|\E\|_F$, then we lose an additive $\sqrt{m}$ factor in the error guarantee. Note that this can be compensated by scaling $\eps$ by a $\frac{1}{\sqrt{m}}$ factor, in which case we will get a $(1 + \eps)$ multiplicative approximation for any $\E$ which is not too much smaller than $\U^* f(\V^*\X)$ (meaning $\|\E\|_F = \Omega( \eps \|\U^* f(\V^*\X)\|_F)$). The runtime in this case will be $(m \kappa/\eps)^{O(k^2)}$, which is still $(\kappa/\eps)^{O(k^3)}$ whenever $m = O(2^k)$. Note that if the noise $\E$ becomes arbitrarily smaller than the signal $\U^* f(\V^*\X)$, then the multiplicative approximation of Equation \ref{eqn:fptbound} degrades, and instead becomes an additive guarantee.

To see why this is a reasonable bound, we must first examine the normalizations implicit in our problem. As always we assume that $\V^*$ has unit norm rows. Using the $2$-stability of Gaussians, we know that $\ex{(\V^*\X)_{i,j}^2} = 1$ for any $i,j \in [k] \times [n]$, and by symmetry of Gaussians we have that $\ex{f(\V^*\X)_{i,j}^2} = 1/2$. By linearity of expectation we have $\ex{\|f(\V^*\X)\|_F^2} = kn/2$. Since $ \sigma_{\min}^2 \|f(\V^*\X)\|_F^2\leq \|\U f(\V^*\X)\|_F^2 \leq \sigma_{\max}^2 \|f(\V^*\X)\|_F^2$, it follows that
\[\frac{\sigma_{\min}^2(\U)kn}{2} \leq \ex{ \|\U f(\V^*\X)\|_F^2} \leq \frac{\sigma_{\max}^2(\U)kn}{2}\]
Thus for the scale of the noise $\E$ to be within a $\Omega(1)$ factor of the average squared entry of $\U f(\V^*\X)$ on average, we expect $\|\E\|_F = O( \sigma_{\max} \sqrt{nk })$ and  $\|\E\|_F = \Omega( \sigma_{\min} \sqrt{nk })$

Now consider the case where $\E$ is a random matrix, coming from a very broad class of distributions. Since $\E \in \R^{m \times n}$ with $n >> m$, one of the main results of random matrix theory is that many such matrices are \textit{approximately} isometries \cite{vershynin2010introduction}. Thus, for a such a random matrix $\E$ normalized to be within a constant of the signal, we will have  $\|\E\|_2 = O( \sigma_{\max}  \sqrt{\frac{nk}{m}})$. This gives

\[ \|\AA - \U f(\V^*\X)\|_F \leq \|\E\|_F(1 + O(\eps))\]
after scaling $\eps$ by a quadratic factor. In general, we get multiplicative approximations whenever either the spectrum of $\E$ is relatively flat, or when we allow $(m \kappa/\eps)^{O(k^2)}$ runtime. Note that in both cases, for the above bound to be a $\|\AA - \U f(\V^*\X)\|_F \leq (1+\eps)\|\E\|_F$ approximation, we must have $\|\E\|_F = \Omega( \eps \|\U^* f(\V^*\X)\|_F)$ as noted above. Otherwise, the error we are trying to compete against is too small when compared to the matrices in question to obtain a multiplicative approximation.

\subsection{Main Algorithm}
Our algorithm is then formally given in  Figure \ref{fig:gaussiannoise}. 
Before presenting it, we first recall some fundamental tools of numerical linear algebra. First, we recall the notion of a subspace-embedding.

\begin{definition}[Subspace Embedding]
Let $\U \in \R^{m \times k}$ be a rank-$k$ matrix and, let $\mathcal{F}$ be family of random matrices with $m$ columns, and let $\S$ be a random matrix sampled from $\mathcal{F}$. Then we say that $\S$ is a $(1 \pm \delta)$-$\ell_2$-subspace embedding for the \textit{column space} of $\U$ if for all $x \in \R^k$,
\[ \|\S\U x\|_2 = (1 \pm \delta)\|\U x\|_2 \]
\end{definition}
Note in the above definition, $\S$ is a subspace embedding for the column span of $\U$, meaning for any other basis $\U'$ spanning the same columns as $\U$, we have that $\S$ is also a $(1 \pm \delta)$-$\ell_2$-subspace embedding for $\U'$. For brevity, we will generally say that $\S$ is a subspace embedding for a matrix $\U$, with the understanding that it is in fact a subspace embedding for \textit{all} matrices with the same column span as $\U$. Note that if $\S$ is a subspace embedding for a rank-$k$ matrix $\U$ with largest and smallest singular values $\sigma_{\max}$ and $\sigma_{\min}$ respectively, then $\S \U$ is rank-$k$ with largest and smallest singular values each in the range $(1 \pm \delta)\sigma_{\max}$ and $(1 \pm \delta)\sigma_{\min}$ respectively. The former fact can be seen by the necessity that $\|\S\U x\|_2$ be non-zero for all non-zero $x \in \R^k$, and the latter by the fact that $\max_{x \in \R^k, \; \|x\|_2 = 1} \|\S \U x\|_2 = (1 \pm \delta)\max_{x \in \R^k, \; \|x\|_2 = 1} \|\U x\|_2 = \sigma_{\max}$, and the same bound holds replacing $\max$ with $\min$. Our algorithm will utilize the following common family of random matrices.

\begin{Frame}[\textbf{Algorithm \ref{fig:gaussiannoise} : Neural Net LRA with Gaussian Input and Noise}$(\AA,\X)$.]
\label{fig:gaussiannoise}
\begin{enumerate}
	          \item Generate a random matrix $\S \in \R^{c_1 k / \delta^2 \times m}$ of i.i.d. Gaussian $\mathcal{N}(0,1/k)$ variables for some sufficiently large constant $c_1$ and $\delta = 1/10$. 
	   
	    \item Enumerate all $ k \times c_1 k/\delta^2$ matrices $\M^1,\M^2,\ldots,\M^\nu$ with entries of the form $ \frac{1}{\sigma_{\min}(\U)}(1 + \frac{\eps^4}{c  k^4})^{-i}$ for integers $0 \leq  i  \leq c'k^8 (1/\eps^8) \kappa^8 $ for sufficiently large constants $c,c'$ and any $\frac{1}{\eps}> k$. Note that $\nu = 2^{O(k^2 \log( \frac{1}{\eps}k \kappa) ) }$.

	    \item For $i=1,2,\dots,\nu$
	    \begin{enumerate}
	        \item Generate a matrix $\G \in \R^{k \times n}$ s.t. $\G$ consists of i.i.d. $\mathcal{N}\left(0, \Theta\big(\big(\frac{\eps^{-2} \kappa^2 k \|\M \S \E \|_F}{\sqrt{n}}\big)^2\big)\right)$ random variables. Note, we can guess the value $\|\M \S \E \|_F$ in $O(\log(n))$ powers of $2$ around $\|\M \S \AA \|_F$.
	       	    \item For each row $p \in [k]$ and $q \in [n]$, let $y_q = \ttx{sign}{(\M^i \S \AA + \G)_{p,q}}$.
	       	    \item For each $p=1,2,\dots,[k]$, let $w^p_i$ be the solution to the following convex program:
	       	  \[ \max_{w, } \sum_{i=1}^n y_i \langle w , X_{*,i} \rangle \]
	       	  \[ \text{subject to } \; \|w\|_2^2 \leq 1 \]
	       	  \item Let $\V^i \in \R^{k \times d}$ be the matrix with $p$-th row equal to $w^p_i$.
	    \end{enumerate}
\item Let $\U$ and $\V^{i^*}$ be the matrices that achieve the minimum value of the linear regression problem 
\[\arg \min_{\U, \V^i} \|\AA - \U f(\V^i \X)\|_F^2\]
\end{enumerate}
\textbf{Output:} $(\U,\W^{i^*})$.
\end{Frame}

\begin{proposition}[Gaussian Subspace Embedding \cite{sarlos2006improved}] \label{prop:subspaceembedding}
Fix any rank $k$-matrix matrix $\U \in \R^{m \times k}$, and let $\S\in \R^{c_1 k/\eps^2 \times m}$ be a random matrix where every entry is generated i.i.d. Gaussian $\mathcal{N}(0,1/k)$, for some sufficiently large constant $c_1$. Then with probability $99/100$, $\S$ is a subspace embedding for $\U$.
\end{proposition}

Our algorithm is then as follows. We first sketch the input matrix $\AA$ by a $O(k) \times m$ Gaussian matrix $\S$, and condition on it being a subspace embedding for $\U$. We then left multiply by $\S$ to obtain $\S\AA = \S \U f(\V\X) + \S\E$. Now we would like to ideally recover $f(\V\X)$, and since $\S$ is a subspace embedding for $\U$, we know that $\S\U$ has full column rank and thus has an left inverse. Since we do not know $\U$, we must guess the left inverse $(\S\U)^{-1} \in \R^{k \times O(k)}$ of $\S\U$. We generate guesses $\M^i$ of $(\S\U)^{-1}$, and try each of them. For the right guess, we know that after left multiplying by $\M^i$ we will have $\M^i \S \AA = f(\V\X) + \M^i \S \E + \ZZ$, where $\ZZ$ is some error matrix which arises from our error in guessing $(\S\U)^{-1}$.

We then observe that the signs of each row of this matrix can be thought of as labels to a noisy halfspace classification problem, where the sign of $(\M^i \S \AA )_{p,q}$ is a noisy observation of the sign of $\langle \V_{p,*} , \X_{*,q} \rangle$. Using this fact, we then run a convex program to recover each row $\V_{p,*}$. In order for recovery to be possible, there must be some non-trivial correlation between the labeling of these signs, meaning the sign of $(\M^i \S \AA )_{p,q}$, and the true sign of $\langle \V_{p,*} , \X_{*,q} \rangle$. In order to accomplish this, we must \textit{spread out} the error $\E$ to allow the value of $\langle \V_{p,*} , \X_{*,q} \rangle$ to have an effect on the observed sign a non-trivial fraction of the time. We do this by adding a matrix $\G$ such that the $i$-th row $\G_{i,*}$ consists of i.i.d. $\mathcal{N}\big(0, \Theta\big(\big(\frac{\eps^{-2} \kappa^2 k \|\M \S \E \|_F}{\sqrt{n}}\big)^2\big)\big)$  random variables to $\M^i \S \AA$. We will simply  guess the value $\|\M \S \E \|_F$ here in $O(\log(n))$ powers of $2$ around $\|\M \S \AA \|_F$. We prove a general theorem (Theorem \ref{thm:modular}) about the recovery of hyperplanes $v\in \R^d$ when given noisy labels from a combination of ReLU observations, adversarial, and non-adversarial noise components. Finally, we solve for $\U$ by regression. 
The full procedure is described formally given in Algorithm \ref{fig:gaussiannoise}.

\subsection{Analysis} 
First note that by our earlier bounds on the singular values of $\X$ (Proposition \ref{prop:gaussiancondition}), we have $\|f(\V^*\X)\|_F \leq O(\sqrt{n k})$, thus if $\|\E\|_F > \sigma_{\max}(\U^*)\frac{\sqrt{n k}}{\eps}$, we can simply return $\U^* = 0,\V^* = 0$, and obtain our desired competitive approximation with the cost $OPT = \|\E\|_F$. Thus, where can now assume that $\|\E\|_F < \sigma_{\max}(\U^*)\frac{\sqrt{n k}}{\eps}$.

By Proposition \ref{prop:subspaceembedding}, with probability $99/100$ we have both that $\S\U^*$ is rank-$k$ and that the largest and smallest singular values of $\S\U^*$ are perturbed by at most a $(1 \pm \delta)$ factor, meaning $\sigma_{\max}(\U^*) = (1 \pm \delta) \sigma_{\max}(\S\U^*)$ and $\sigma_{\min}(\U^*) = (1 \pm \delta) \sigma_{\in}(\S\U^*)$, from which it follows that $\kappa(\U^*) = (1 \pm O(\delta)) \kappa(\S\U^*)$. Note that we can repeat the algorithm $O(n)$ times to obtain this result with probability $1-\exp(-n)$ at least once by Hoeffding bounds. So we can now condition on this and assume the prior bounds on the singular values and rank of $\S\U^*$. Thus we will now write $\sigma_{\max} = \sigma_{\max}(\S\U^*)$,
  $\sigma_{\min} = \sigma_{\min}(\S\U^*)$, and $\kappa = \kappa(\S\U^*)$, with the understanding that these values have been perturbed by a $(1 \pm 3\delta) <(1 \pm 1/2)$ factor.

We can assume that we know $\kappa$ and $\sigma_{\min}(\U^*)$ up to a factor of $2$ by guessing them in geometrically increasing intervals.
Note that we can assume $\sigma_{\max}$ is within a $\poly(n)$ factor of the largest column norm of $\AA$, since otherwise $\|\E \|_F$ would necessarily be larger than $\sigma_{\max}(\U^*)\frac{\sqrt{n k}}{\eps}$. Given this column norm, we obtain an interval $[a,b] \subset \R$ $\frac{a}{b} = \poly(n, \kappa)$, such that both $\kappa$ and $\sigma_{\min}(\U^*)$ must live inside $[a,b]$. Then we can make $O(\log^2(\frac{a}{b})) = O(\log^2(n \kappa))$ guesses to find $\kappa$ and $\sigma_{\min}(\U^*)$ up to a factor of $2$. 
Thus guessing the correct approximations to $\kappa, \sigma_{\min}(\U^*)$ will not effect our run time bounds, since our overall complexity is already polynomial in $n$ and $\kappa$. Similarly, we can also guess the value of $\|\M \S \E \|_F$ up to a factor of $2$ using $O(\log(n \kappa))$ guesses, as is needed in step $3a$ of Algorithm \ref{fig:gaussiannoise}.

The following Proposition gives the error bound needed for the right guess of the inverse $(\S\U)^{-1}$

	\begin{proposition}\label{prop:guesserror}
		On the correct guess of $\sigma_{\max}(\S\U^*)$ (up to a constant factor of $2$ error), there is an $i \in [\nu]$ such that $\M^i = (\S\U^*)^{-1} + \mathbf{\Lambda}$ where $\|\mathbf{\Lambda}\|_{\infty} \leq \frac{\eps^4}{ \sigma_{\min} \kappa^4 k^4}$. 
    \end{proposition}
    \begin{proof}
    First note that no entry in $(\S\U^*)^{-1}$ can be greater than $\frac{1}{\sigma_{\min}}$ (since $\sigma_{\min}$ is the smallest singular value of $\S\U^*$, and therefore $\frac{1}{\sigma_{\min}}$ is the largest singular value of $(\S\U^*)^{-1}$.  Thus there is a guess of $\M^i$ such that for each entry $(p,q)$ of $(\S\U^*)^{-1}$ in the range $(\frac{1}{\sigma_{\min}(1/\eps)^4 \kappa^4 k^4}, \frac{1}{\sigma_{\min}})$, we have $\M^i_{p,q} = (\S\U^*)^{-1}_{p,q}(1 \pm \frac{1}{(1/\eps)^4 \kappa^4 k^4}) = (\S\U^*)^{-1} \pm \frac{1}{\sigma_{\min}(1/\eps)^4 \kappa^4 k^4}$. For all other entries less than $\frac{1}{\sigma_{\min}(1/\eps)^4 \kappa^4 k^4}$, we get  $\M^i_{p,q} = (\S\U^*)^{-1}_{p,q} \pm \frac{1}{\sigma_{\min}(1/\eps)^4 \kappa^4 k^4}$ by setting $\M^i_{p,q} = \frac{1}{\sigma_{\min} (1/\eps)^4 \kappa^4 k^4}$ (which is the lowest guess of value which we make for the coordinates of $\M^i$), from which the proposition follows.
    \end{proof}
    
    \noindent

	\subsection{Learning Noisy Halfspaces:} 
By Proposition \ref{prop:guesserror}, we know that for the correct guess of $\M^i$ we can write $\M^i \S \AA =  f(\V^*\X) + (\mathbf{M}^i\S \E) + \ZZ $ where $\ZZ = \mathbf{\Lambda} \S \U^* f(\V^*\X)$. Thus $\M^i \S \AA $ can be thought of as a noisy version of $f(\V^*\X)$. We observe now that our problem can be viewed as the problem of learning a halfspace in $\R^d$ with noisy labels. Specifically, we are obtain examples of the form $\X_{*,q}$ with the label $y_q = \text{Sign}\big(f(\V_{p,*}^*\X_{*:q}) + (\mathbf{M}^i \S \E)_{p,q} +\mathbf{Z}_{p,q} \big) \in \{1,-1\}$, and our goal is to recover $\V_{p,*}^*$ from these labeled examples $\{y_q\}$. Note that if the labeled examples were of the form $\X_{*,q}$ and $ \text{Sign}(\langle \V_{p,*}^*,\X_{*,q} \rangle)$, then this would correspond to the noiseless learning problem for half-spaces. Unfortunately, our problem is not noiseless, as it will often be the case that  $\text{Sign}\big(f(\V_{p,*}\X_{*:q}) + (\mathbf{M}^i \S \E)_{p,q} +\mathbf{Z}_{p,q} \big)  \neq \text{Sign}(\langle \V_{p,*}, \X_{*,q} \rangle)$ (in fact, this will happen very close to half of the time). We will demonstrate, however, that recovery of $\V_{p,*}$ is still possible by showing that there is a non-trivial correlation between the labels $y_q$ and the true sign. To do this, we show the following more general result.

\begin{theorem}\label{thm:modular} 
Given $n$ i.i.d. Gaussian examples $\X \in \R^{d \times n}$ with labels $y_q = \text{Sign}\big( (f(\V\X) + \G +  \mathbf{B})_{p,q}  \big) \in \{1,-1\}$ where $\G$ is an arbitrary fixed matrix independent of $\X$, and $\BB$ is any matrix such that $\|\BB \|_F \leq \frac{\sqrt{n}}{ \omega}$ for any $\omega = o(\sqrt{n})$. Then if $v_{p,*}$ is the solution to the convex program in step $3c$ of Figure \ref{fig:gaussiannoise} run on the inputs examples $\X$ and $\{y_q\}$, then with probability $1-e^{-n^{1/2}/10}$ we have 
\[ \|v_{p,*} - \V_{p,*}\|_2^2 = O\left(\sqrt{\omega} \frac{\|\G\|_F}{\sqrt{n}} \left(\frac{\sqrt{d}}{\sqrt{n}} + \frac{1}{n^{1/4}}+ \frac{ \log(\omega)}{\omega}\right) \right) \]
\end{theorem}
\noindent
Before we prove the theorem, we first show that our setting fits into this model. Observe that in our setting, $\G = \M^i \S \E$, and $\BB = \ZZ$. Note that the Gaussian matrix added in Step 3a of Algorithm \ref{fig:gaussiannoise} is a component of proof of Theorem \ref{thm:modular}, and different than the $\G$ here. Namely, for Theorem \ref{thm:modular} to work, one must first add Gaussian matrix to *smear out* the fixed noise matrix $\M^i \S \E$. See the proof of Theorem \ref{thm:modular} for further details. The following Proposition formally relates our setting to that of Theorem \ref{thm:modular}.
\begin{proposition}\label{prop:fitsinto}
We have $\|(\mathbf{M}^i \S \E) \|_F =O( \frac{1}{\sigma_{\min}} \sqrt{m} \|\E\|_2)$, and $\|\ZZ\|_F = \|\mathbf{\Lambda} \S \U^* f(\V\X)\|_2 \leq \sqrt{n}\frac{2}{(1/\eps)^4 \kappa^4 k^{2}}$
\end{proposition}
\begin{proof}
Since $\S\U^*$ is $\kappa = \sigma_{\max}/\sigma_{\min}$ conditioned (as conditioned on by the success of $\S$ as a subspace embedding for $\U^*$), it follows that for any row $p$, we have $\|((\S\U^*)^{-1})_{p,*} \|_2 \leq \frac{1}{\sigma_{\min}}$. Thus by Proposition \ref{prop:guesserror} we have $\|\M^i_{p,*}\|_2 \leq \frac{1}{\sigma_{\min}} + \frac{1}{\sigma_{\min} (1/\eps)^4 \kappa^4 k^3} \leq \frac{2}{\sigma_{\min}}$, and by Proposition \ref{prop:gaussiancondition}, noting that $\S$ can be written as a i.i.d. matrix of $\mathcal{N}(0,1)$ variables scaled by $\frac{1}{\sqrt{k}}$, we have $\|\M^i_{p,*} \S\|_2 \leq \frac{2}{\sigma_{\min}} \sqrt{\frac{2m}{k}}$. Applying this over all $O(k)$ rows, it follows that $\|\M^i \S \E\|_F =O( \frac{1}{\sigma_{\min}} \sqrt{m} \|\E\|_2)$, where $\|\E\|_2$ is the spectral norm of $\E$.

		For the second, note the bound $\|\mathbf{\Lambda}\|_\infty \leq 1/(\sigma_{\min} (1/\eps)^4 \kappa^4 k^4)$ from Proposition \ref{prop:guesserror} implies that $\|\mathbf{\Lambda}_{p,*}\|_2 \leq  1/(\sigma_{\min}(1/\eps)^4 \kappa^4 k^3)$ (using that $k > c_1/\delta$ where $\delta$ is as in Figure \ref{fig:gaussiannoise}), so $\|\mathbf{\Lambda}_{p,*} \S \U^* \|_2 \leq \frac{\sigma_{\max}}{\sigma_{\min}(1/\eps)^4 \kappa^4 k^3} \leq \frac{1}{(1/\eps)^4 \kappa^4 k^3}$.
		Now by Proposition \ref{prop:gaussiancondition}, we have that the largest singular value of $\X$ is at most $2 \sqrt{n}$ with probability at least $1-2e^{-n/8}$, which we now condition on. Thus $\|\V\X\|_F \leq 2 \sqrt{nk}$, from which it follows $\|f(\V\X)\|_F \leq 2\sqrt{nk}$, giving $\|\ZZ_{p,*}\|_2 \leq 2\sqrt{n}\frac{1}{(1/\eps)^4 \kappa^4 k^{5/2}}$ for every $p \in [k]$, so $\|\ZZ\|_F \leq \sqrt{n}\frac{2}{(1/\eps)^4 \kappa^4 k^{2}}$ as needed. 
		
\end{proof}
By Theorem \ref{thm:modular} and Proposition \ref{prop:fitsinto}, we obtain the following result. 

\begin{corollary}\label{cor:1}
        	Let $i$ be such that $\M^i = (\S\U^*)^{-1} + \mathbf{\Lambda}$, where $\|\mathbf{\Lambda}\|_\infty \leq  1/(\sigma_{\min} (1/\eps)^4 \kappa^4 k^4)$ as in Proposition \ref{prop:guesserror}, and let $\W^i$ be the solution to the convex program as defined in Step $3d$ of the algorithm in Figure \ref{fig:gaussiannoise}. Then with probability $1 - \exp(-\sqrt{n}/20)$, for every row $p \in [k]$ we have 
        	\[ \|\V_{p,*} - \W^i_{p,*}\|_2^2 \leq \frac{\eps \sqrt{m} \|\E\|_2 }{\sigma_{\min} \sqrt{n}}\] 
\end{corollary}
\begin{proof}
By Proposition \ref{prop:fitsinto} we can apply Theorem \ref{thm:modular} with $\omega = \eps^{-4} \kappa^4 k^2$ and $\|\G\|_F = O(\frac{1}{\sigma_{\min}} \sqrt{m} \|\E\|_2)$, we obtain the stated result for a single row $p$ with probability at least $1 - e^{-\sqrt{n}/10}$ after taking $n = \poly(\kappa,d)$ sufficiently large. Union bounding over all $k$ rows gives the desired result.  

\end{proof}

\paragraph*{Proof of Theorem \ref{thm:modular}}  
 To prove the theorem, we will use techniques from \cite{plan2013robust}.  Let $v \in \R^d$ be fixed with $\|v\|_2 =1 $, and let $\X \in \R^{d \times n}$ be a matrix of i.i.d. Gaussian $\mathcal{N}(0,1)$ variables. Let $y_q$ be a noisy observation of the value $\sgn{\langle v , \X_{*,q} \rangle}$, such that the $y_q$'s are independent for different $q$. 
 We say that the $y_q$'s are \textit{symmetric} if  $\ex{y_q \; | \; \X_{*,q} } = \theta_q(\langle v, \X_{*,q}\rangle)$ for each $q \in [n]$. In other words, the expectation of the noisy label $y_q$ given the value of the sample $\X_{*,q}$ depends only on the value of the inner product $\langle v, \X_{*,q}\rangle$.
We consider now the following requirement relating to the correlation between $y_q$ and $\sgn{\langle v, \X_{*,q} \rangle}$. 
\begin{equation}\label{eqn:assump1}
     \mathbb{E}_{g \sim \mathcal{N}(0,1)} \Big[ \theta_q(g)g \Big] = \lambda_q \geq 0
\end{equation}
Note that the Gaussian $g$ in Equation \ref{eqn:assump1} can be replaced with the identically distributed variable $\langle v, X_{*,q}\rangle$. In this case, Equation \ref{eqn:assump1} simply asserts that there is indeed some correlation between the observed labels $y_q$ and the ground truth $\sgn{\langle v , \X_{*,q} \rangle}$. When this is the case, the following convex program is proposed in \cite{plan2013robust} for recovery of $v$
\begin{equation}\label{eqn:opt}
 \max_{w, \: \|w\|_2 \leq 1 }\sum_{q=1}^n y_q \langle w , X_{*,q} \rangle 
\end{equation}

We remark that we must generalize the results of \cite{plan2013robust} here in order to account for $\theta_q$ depending on $q$. Namely, since $\E$ is not identically distribution, we must demonstrate bounds on the solution to the above convex program for the range of parameters $\{\lambda_q\}_{q \in [n]}$.

Now fix a row $p \in [k]$ and let $v = \V_{p,*}^*$. We will write $\G' = \G + \G''$, where $\G''$ is an i.i.d. Gaussian matrix distributed $(\G'')_{i,j} \sim \mathcal{N}(0,\eta^2 )$ for all $i,j \in [k] \times [n]$, where $\eta = 100 \sqrt{\omega} \|\G\|_F/\sqrt{n}$. For technical reasons, we replace the matrix $\G$ with $\G'$ be generating and adding $\G''$ to our matrix $(f(\V\X) + \G + \BB)$. 
Then the setting of Theorem \ref{thm:modular}, we have $y_q = \text{Sign}( (f(\V^*\X) + \G' + \BB)_{p,q} )$. Note that by the definition of $\eta$, at most $\frac{n}{100\omega}$ entries in $\G$ can be larger than $\eta/10 = 10 \sqrt{\omega} \|\G\|_F/\sqrt{n}$. Let $\BB'$ be the matrix of entries of $\G$ which do not satisfy this, so we instead write $y_q = \text{Sign}( (f(\V^*\X) + \G' + \BB' + \BB)_{p,q} )$, where $\G' = \G + \G'' - \BB'$. Thus $\G_{p,q}' \sim \mathcal{N}(\mu_{p,q}, \eta^2)$ where $\mu_{p,q} <  10 \sqrt{\omega} \|\G\|_F/\sqrt{n} = \eta/10$. Note that $\BB'$ is $\frac{n}{100\omega}$ sparse, as just argued.

Note that the above model does not fully align with the aforementioned model, because $\BB$ is an arbitrary matrix that can depend potentially on $f(\V^*\X)$, and not just $\langle \V_{p,*}^*, \X_{*,q} \rangle$. So instead, suppose hypothetically that in the place of $y_q$ we were given the labels $y_q' =  \text{Sign}( (f(\V^*\X) + \G')_{p,q} )$, which indeed satisfies the above model. Note that we have also removed $\BB'$ from the definition of $y_q'$, since we will handle it at the same time as we handle $\BB$. In this case we can write $\ex{y_q' \; | \; \X_{*,q} } = \mathbb{E}_g [ \text{sign}(f( \langle \X_{*,q}, \V_{p,*}^* \rangle) + g_{p,q} )\; \big| \;\langle \X_{*,q}, \V_{p,*} \rangle ]$ where $g_{p,q} \sim \mathcal{N}(\G_{p,q} - \BB_{p,q}', \eta^2)$ is a Gaussian independent of $X$.

Proposition \ref{prop:lambda} gives the corresponding value of $\lambda$ for this model. 

\begin{proposition}\label{prop:lambda}
The function $\theta_q$ as defined by the hypothetical labels $y_q'$ satisfies Equation \ref{eqn:assump1} with $\lambda_q \geq \frac{c}{\eta}$ for some constant $c >0$, where $\eta = 100 \sqrt{\omega} \|\G\|_F/\sqrt{n}$.
\end{proposition}
\begin{proof}
We can write $\ex{y_q' \; | \; \X_{*,q} } = \mathbb{E} [ \text{sign}(f( \langle \X_{*,q}, \V_{p,*} \rangle) + g_{p,q} )\; \big| \;\langle \X_{*,q}, \V_{p,*} \rangle ]$ where $g_{p,q} \sim \mathcal{N}(\G_{p,q} - \BB_{p,q}, \eta^2)$. Let $\mu_q =\G_{p,q} - \BB_{p,q}$ (for a fixed row $p$.  Then $\theta(z) = 1 - 2\pr{g \leq - f(z)}$, and Equation \ref{eqn:assump1} can be evaluated by integration by parts. Let $p_q(z) = \frac{1}{\sqrt{2\pi \eta^2}}e^{-\frac{(z - \mu_q)^2}{2\eta^2}}$ is the p.d.f. of $g_{p,q}$. Note by the prior paragraphs we have $\eta^2 > 10\mu_q^2$ for all $q$. Then we have 
\begin{equation*}
    \begin{split}
        \lambda & = \ex{\theta'(g)} =\ex{2p(-f(z))} \\
        & =  \mathbb{E}_{z \sim \mathcal{N}(0,1)} \Big[ \sqrt{\frac{2}{\pi (\eta^2)}} e^{-(f(z)+\mu_q)^2/(2 (\eta^2)) } \Big] \\
            & =  \sqrt{\frac{2}{\pi (\eta^2)}}  \mathbb{E}_{z \sim \mathcal{N}(0,1)} \Big[ e^{-\frac{f(z)^2+2\mu_q f(z) + \mu_q^2}{2 \eta^2} } \Big] \\
            & =  \Omega(\frac{1}{\eta})
    \end{split}
\end{equation*}

  \end{proof}\noindent
Now for any $ z\in \R^d$ with $\|z\|_2 \leq 1$, let $h(z) =\frac{1}{n}\sum_{q=1}^n y_q\langle z , X_{*,q} \rangle$, and let $h'(z) =\frac{1}{n}\sum_{q=1}^n y_q'\langle z , X_{*,q} \rangle$.  Observe that the hypothetical function $h'$ corresponds to the objective function of Equation \ref{eqn:opt} with values of $y_q'$ which satisfy the model of \ref{eqn:assump1}, whereas $h$, corresponding to the labels $y_q$ which we actually observe, does not. Let $B^d_2 = \{ x \in \R^d \: | \: \|x\|_2 \leq 1 \}$ and let $\mathcal{B}^d_2 = B-B= \{ x-z \: | \: x,y \in B\}$ be the Minkowski difference. The following follows immediately from \cite{plan2013robust}.
\begin{lemma}[Lemma 4.1 \cite{plan2013robust}]\label{propexp}
For any $z \in B^d_2$, we have $\ex{h'(z)} = \frac{1}{n}\sum_{q=1}^n \lambda_q \langle z, \V_{p,*}\rangle$ and thus because $h'$ is a linear function, we have 
\[\ex{h'(\V_{p,*}) -h'(z) } = \ex{h'(\V_{p,*} - z) }= \frac{1}{n}\sum_{q=1}^n\lambda_q(1 - \langle \V_{p,*},z\rangle) \geq \frac{1}{n}\sum_{q=1}^n \frac{\lambda_q}{2} \|\V_{p,*} - z\|_2^2\]
\end{lemma}

We now cite  Proposition 4.2 of \cite{plan2013robust}. We remark that while the proposition is stated for the concentration of the value of $h'(z)$ around its expectation when the $\lambda_q$ are are all uniformly the same $\lambda_q = \lambda$, we observe that this fact has no bearing on the proof of Proposition \ref{prop:4.2} below. This is because only the $y_q \in \{1,-1\}$ depend on the $\lambda_q$'s, and the concentration result of Proposition \ref{prop:4.2}, in fact, holds for \textit{any} possible values of the $y_q$'s. Thus one could replace $h'(z)$ below with any function of the form $\hat{h}(z) = \frac{1}{n}\sum_{q=1}^n y_q \langle z,g_q\rangle$ for any values of $y_q \in \{1,-1\}$, and the following concentration result would hold as long as $\{g_q\}_{q \in [n]}$'s is a collection of independent $\mathcal{N}(0,\mathbb{I}_d)$ variables. 
\begin{proposition}[Proposition 4.2 \cite{plan2013robust}]\label{prop:4.2}
For each $t > 0$, we have
\[ \bpr{ \sup_{z \in \mathcal{B}^d_2} \big|h'(z) - \ex{h'(z)}     \big| \geq \frac{4\sqrt{d}}{\sqrt{n}} + t } \leq 4 \exp(-\frac{-n t^2}{8})  \]
\end{proposition}\noindent 
We now demonstrate how to utilize these technical results in our setting. First, however, we must bound $\sup_{z \in B} |h'(z) - h(z)|$, since in actuality we will need bounds on the value of $h(z)$. We first introduce a bound on the expected number of flips between the signs $y_{p,*}$ and $y_{p,*}'$.

\begin{proposition}\label{prop:expbound}
Let $T = \{q \in [n] \; | \; y_q \neq y_q'\}$. Then with probability $1-e^{-10\sqrt{n}}$, we have $|T| \leq 11\frac{n}{\omega}$. 

\end{proposition}
\begin{proof}
We have $\|\BB_{p,*}\|_1 \leq  \sqrt{n}\|\BB_{p,*}\|_2 \leq  n/\omega $ by the original assumption on $\BB$ in Theorem \ref{thm:modular}. Then $\pr{q \in T}$ is at most the probability $\G_{p,q}'$ is in some interval of size $2|\BB_{p,q}|$, which is at most $2|\BB_{p,q}|$ by the anti-concentration of Gaussians. Thus $\ex{|T|} \leq 2 \|\BB_{p,*}\|_1 \leq 2n/\omega$, and by Chernoff bounds $\pr{|T| >10n/\omega} < e^{-10\sqrt{n}}$ as needed. To handle $\BB'$, we simple recall that $\BB'$ was $\frac{n}{100 \omega}$ sparse, and thus can flip at most $\frac{n}{100\omega} < n/\omega$ signs. 
\end{proof}

\begin{proposition}\label{Prop17}
Let $h,h'$ be defined as above. Let $\hat{w} \in B^r_2$ be the solution to the optimization problem 
\begin{equation}
 \max_{w, \: \|w\|_2} n\: h(w) =  \max_{w, \: \|w\|_2} \sum_{q=1}^n y_q \langle w , \X_{*,q} \rangle 
\end{equation}
Then if with probability $1-\exp(-\sqrt{n})$ we have 
\[\bpr{ \sup_{z \in B} \big|h'(z) - h(z)\big| \leq \frac{3 \log(\omega)}{\omega} } \geq 1 - e^{-\sqrt{n}}\]
\end{proposition}
\begin{proof}
Let $S \subset \{x \in \R^d \; | \; \|x\|_\infty \leq 1 \}$ be an $\eps$-net for $\eps = 1/n^3$. Standard results demonstrate the existence of $S$ with $|S| < 2^{12d\log(n)}$ (see e.g. \cite{vershynin2010introduction, woodruff2014sketching}). Fix $z \in S$ and observe $|h'(z) - h(z)| = \frac{2}{n}\sum_{q \in T}| \langle z ,\X_{*,q} \rangle|$. Note that we can assume $\|z\|_2 = 1$, since increasing the norm to be on the unit sphere can only make $|h'(z) - h(z)|$ larger. By Proposition \ref{prop:expbound}, we have $|T| \leq  n/\tau$, where $\tau = \frac{\omega}{11}$ with probability $1-e^{-10\sqrt{n}}$, so we can let $\mathcal{F} = \{T' \subset [n] \; | \; |T'| \leq  n/\tau \}$. Note $|\mathcal{F}| \leq n (e\tau)^{n/\tau}$. Fix $T' \in \mathcal{F}$. The sum $\sum_{q \in T'}| \langle z ,\X_{*,q} \rangle|$ is distributed as the $L_1$ of a Gaussian $\mathcal{N}(0,1)$ vector in $|T'|$, dimensions, and is $\sqrt{|T'|}$-Lipschitz  with respect to $L_2$, i.e. $| \|x\|_1 - \|y\|_1 | \leq \|x - y\|_1 \leq \sqrt{|T'|}\|x-y\|_2$. So by Lipschitz concentration (see \cite{vershynin2010introduction} (Proposition 5.34)), we have $\pr{ \frac{1}{n}\sum_{q \in T'}| \langle z ,\X_{*,q} \rangle| >  \frac{\log(e\tau)}{\tau}}  \leq \exp(- \log^2(e\tau) n/\tau)$. We can then union bound over all $T' \in \mathcal{F}$ and $z \in S$ to obtain the result with probability 
\[1- \exp\Big(- \frac{n\log^2(e\tau)}{\tau} + \frac{n \log(e\tau)}{\tau} + \log(n) + 12r\log(n) \Big) > 1 -  \exp\Big(- \log^2(\tau) n/(2\tau) \Big) \]
So let $\mathcal{E}_1$ be the event that $\sum_{q \in T'}| \langle z ,\X_{*,q} \rangle| < \sqrt{\log(\tau)} |T'|$ for all $T' \in \mathcal{F}$ and $z \in S$. 
Now fix $w \in \R^d$ with $\|w\|_2  \leq 1$, and let $y \in S$ be such that $\|y-z\|_2 \leq 1/n^3$. Observing that $h$ and $h'$ are linear functions, we have $|h(z) - h'(z)| \leq |h(y) - h'(y)| + |h(z-y) - h'(z-y)| \leq \frac{\log(e \tau)}{\tau} + |h(z-y) - h'(z-y)|$.  Now condition on the event $\mathcal{E}_2$ that $\|X\|_F^2 \leq 10nd$, where $\pr{\mathcal{E}_2} >1 - \exp(nd)$ by standard concentration results for $\chi^2$ distributions \cite{laurent2000}. Conditioned on $\mathcal{E}_2$ we have $|h(z-y)| + |h'(z-y)| \leq 4\sqrt{10nd}/n^3 \leq 1/\tau$, giving $|h(z) - h'(z)| \leq \frac{3 \log(e\tau)}{\tau} < \frac{3 \log(\omega)}{\omega}$, from which the proposition follows after union bounding over the events $\mathcal{E}_1,\mathcal{E}_2$ and Proposition \ref{prop:expbound}, which hold with probability $ 1 - (\exp(- \log^2(\tau) n/(2\tau)) +\exp(nd) + \exp(-10\sqrt{n})) > 1-\exp(-\sqrt{n})$.

\end{proof}

\begin{lemma}\label{lem:modular}
Let $\hat{w}$ be the solution to the optimization Problem in Equation \ref{eqn:opt} for our input labels $y_q = \text{sign}((f(\V\X) + \G' + \BB+ \BB')_{p,q})$. Then w with probability $1-e^{-n^{1/2}/10} $, we have $\|\hat{w} - \V_{p,*}\|_2^2 = O( \sqrt{\omega} \|\G\|_F/\sqrt{n} )\big(\frac{4\sqrt{d}}{\sqrt{n}} + \frac{1}{n^{1/4}}+ \frac{6 \log(\omega)}{\omega}\big) $ for some constant $c$.
\end{lemma} 

\begin{proof}
Applying Lemma \ref{propexp}, and a union bound over the probabilities of failure in Proposition \ref{prop:4.2} with $t = n^{1/4}$ and Proposition \ref{Prop17}, we have
\begin{equation*}
    \begin{split}
         0 & \leq h(\hat{w}) - h(\V_{p,*}) \\
         & \leq h'(\hat{w}) - h'(\V_{p,*}) +\frac{6 \log(\omega)}{\omega} \\
         & = h'\big(\hat{w} - \V_{p,*}\big) +\frac{6 \log(\omega)}{\omega}\\
        &  \leq \ex{ h'\big(\hat{w} - \V_{p,*}\big)} + \frac{4\sqrt{d}}{\sqrt{n}} +\frac{1}{n^{1/4}}+  \frac{6 \log(\omega)}{\omega} \\
 & \leq -\frac{\lambda}{2}\|\hat{w} - \V_{p,*}\|_2^2  + \frac{4\sqrt{d}}{\sqrt{n}} + \frac{1}{n^{1/4}}+  \frac{6 \log(\omega)}{\omega}
    \end{split}
\end{equation*} 
Applying Proposition \ref{prop:lambda}, which yields $\frac{1}{\lambda} = O(\eta) = O(\sqrt{\omega} \|\G\|_F/\sqrt{n})$ completes the proof. 
\end{proof}

\begin{proof}[Proof of Theorem \ref{thm:modular}]
The proof of the theorem follows directly from Lemma \ref{lem:modular}. 
\end{proof}

\subsection{Completing the Analysis}
We will now need the following straightforward lemma to complete the proof.

\begin{theorem}\label{thm:vershyninfinal}
Let $\AA = \U^* f(\V\X) + \E$ be the input, where each entry of $\X\in \R^{d \times n}$ is i.i.d. $\mathcal{N}(0,1)$ and $\E$ independent of $\X$. Then the algorithm in Figure \ref{fig:gaussiannoise} outputs $\U \in \R^{m \times k} ,\V \in \R^{k \times d}$ in time $2^{O(k^2 \log((1/\eps) \kappa ))} \poly(n,d)$ such that with probability $1 - \exp(-\sqrt{n})$ we have 
\[  \|\AA - \U f(\V\X)\|_F \leq \|\E \|_F + O\Big( \Big[\sigma_{\min} \eps \sqrt{n m}  \|\E\|_2\Big]^{1/2} \Big)\]
Where $\|\E||_2$ is the spectral norm of $\E$.
\end{theorem}

\begin{proof}
Let $\W^i \in \R^{k \times r}$ be as in Corollary \ref{cor:1}. Then, taking $n = \poly(d,\kappa,\frac{1}{\eps})$ large enough, we have $\W^i = \V + \mathbf{\Gamma}$ where $\|\mathbf{\Gamma}_{p,*}\|_F^2 \leq  \frac{\eps k  \sqrt{m} \|\E\|_2 }{\sigma_{\min} \sqrt{n}}$ for each row $p$ with probability $1-\exp(-r^4)$ by Corollary \ref{cor:1}. Then applying the spectral norm bound on Gaussian matrices from Proposition \ref{prop:gaussiancondition}, we obtain that $\|\V_{p,*}\X - \W^i_{p,*}\X\|_F^2 = O( \sqrt{n} \frac{\eps k  \sqrt{m} \|\E\|_2 }{\sigma_{\min}} )$ with probability at least $1-e^{-9n}$. Since $f$ just takes the maximum of $0$ and the input, it follows that $\|f(\W^i\X) - f(\V\X)\|_F^2  = O( \sqrt{n} \frac{\eps k  \sqrt{m} \|\E\|_2 }{\sigma_{\min}} )$, and therefore $\|\U^* f(\W^i\X) - \U^* f(\V\X)\|_F^2  = O(\sigma_{\max}^2  \sqrt{n} \frac{\eps k  \sqrt{m} \|\E\|_2 }{\sigma_{\min}} )$, which is at most $O(\sigma_{\min} \eps \sqrt{nm}  \|\E\|_2 )$ after rescaling $\eps$ by a $\frac{1}{\kappa^2 k }$ factor.
Now if $\U$ is the minimizer to  the regression problem $\min_{\U} \|\AA - \U f(\W^i\X)\|_F^2 $ in step $5$ of Figure \ref{fig:gaussiannoise}, then note 
\[\|\AA - \U f(\W^i\X)\|_F \leq \|\AA - \U^* f(\W^i\X)\|_F \leq \|\E\|_F +O\Big( \Big[\sigma_{\min} \eps \sqrt{nm}  \|\E\|_2\Big]^{1/2} \Big)\]
as needed.

For the probability of failure, note that Corollary \ref{cor:1} holds with probability $1 - \exp(-\Omega(\sqrt{n}))$. To apply this, we needed only to condition on the fact that $\S$ was a subspace embedding for $\U$, which occurs with probability $99/100$ for a single attempt. Running the algorithm $O(n)$ times, by Hoeffding bounds at least one trial will be successful with probability  $1 - \exp(-\Omega(\sqrt{n}))$ as desired.  To analyze runtime, note that we try at most $\poly(nd)$ guesses of $\S$ and guesses of $\sigma_{\min}$ and $\kappa$. Moreover, there are at most $\nu = (\frac{\kappa}{\eps})^{O(k^2)}$  guesses $\M^i$ carried out in Step $2$ of Figure \ref{fig:gaussiannoise}). For every such guess, we run the optimization program in step $3c$. Since the program has a linear function and a convex unit ball constraint, it is will known that such programs can be solved in polynomial time \cite{boyd2004convex}.  Finally, the regression problem in step $4$ is linear, and thus can be solved in $\poly(n)$ time, which completes the proof. 

\end{proof}

%% file: SparseNoise.tex
\section{A Polynomial Time Algorithm for Exact Weight Recovery with Sparse Noise}\label{sec:sparsenoise}
In this section, we examine recovery procedures for the weight matrices of a \textit{low-rank} neural network in the presence of 
arbitrarily large sparse noise. Here, by low rank, we mean that $m > k$.  It has frequently been observed in practice that many pre-trained neural-networks exhibit correlation and a low-rank structure \cite{denil2013predicting, denton2014exploiting}. Thus, in practice it is likely that $k$ need not be as large as $m$ to well-approximate the data. 

More formally, we are given $\AA=\U^*f(\V^*\X) + \E$ where $\E$ is some sparse noise matrix with possibly very large entries. We show that under the assumption that $\U^*$ has orthonormal columns and satisfies an incoherence assumptions (which is fairly standard in the numerical linear algebra community) \cite{candes2007sparsity, candes2009exact,keshavan2010matrix,candes2011robust, jain2013low, hardt2014understanding}, 
we can recover the weights $\U^*,\V^*$ exactly, even when the sparsity of the 
matrices is a constant fraction of the number of entries. Our algorithm utilizes 
results on the recovery of low-rank matrices in the presence of a sparse noise. The error matrix $\E \in \R^{m
\times n}$ is a sparse matrix whose non-zero entries are uniformly chosen from the set 
of all coordinates of an arbitrary matrix $\overline{\E}$. Formally, we
define the following noise procedure:

\begin{definition}(Sparse Noise.)
A matrix $\E$ is said to be generated from a $s$-sparse-noise procedure if there is an arbitrary matrix $\overline{\E}$, such that $\E$ is generated by setting all but $s \leq  mn$ entries of $\overline{\E}$ to be $0$ uniformly at random.
\end{definition}

\begin{definition}(Incoherence.)
A rank $k$  matrix $\M \in \R^{m \times n}$ is said to be $\mu$-incoherent if 
$ \texttt{svd}(\MM)= \P \Sigma \Q $ is the singular value decomposition of $\MM$ and
\begin{equation}\label{eqn:inco1}
\begin{split}
    & \max_i \| \P^T e_i \|_2^2 \leq \frac{\mu k}{m} \\
    & \max_i \| \Q e_i \|_2^2 \leq \frac{\mu k}{n} \\
\end{split}
\end{equation}   
and 
\begin{equation}
    \label{eqn:incol2}
    \max_i \| \P \Q \|_\infty \leq\sqrt{ \frac{\mu k}{nm}}
\end{equation}
\end{definition}

\begin{remark}
The values $\|\P e_i\|_2^2$ and $\|\Q e_i\|_2^2$ are known as the (left and right, respectively) \emph{leverage-scores} of $\M$. For an excellent
survey on leverage scores, we refer the reader to \cite{mahoney2011randomized}. We note that the set of leverage scores of $\M$ does not depend on the choice of orthonormal basis $\P$ or $\Q$ \cite{woodruff2014sketching}. Thus, to obtain the bounds given in Equation \ref{eqn:inco1}, it suffices let $\P$ be any matrix with orthonormal columns which spans the columns of $\M$, and similarly it suffices to let $\Q$ be any matrix with orthonormal rows which spans the rows of $\M$.  

\end{remark}

\begin{lemma}
The entire matrix $\U^* f(\V^*\X)$, where $\X$ is i.i.d. Gaussian, $(\U^*)^T,\V^*$ have orthonormal rows, and $\U^*$ is $\mu$-incoherent, meaning $\max_i \| (\U^*)^T e_i\|_2^2 \leq \frac{\mu k}{m}$,  is $\overline{\mu}$-incoherent for 
\[\overline{\mu} = O\big( (\kappa(\V^*))^{2} \sqrt{k \log(n) \mu} + \mu +  (\kappa(\V^*))^4 \log(n) \big)\]
\end{lemma}
\begin{proof} For $t \in \{\max, \min \}$, let $\sigma_t = \sigma_t(\U^* f(\V^*\X))$.  Let $\P \Sigma \Q$ be the SVD of $\U^* f(\V^*\X)$, and let
For any $i$, since $\U^*$ and $\V^*$ are orthonormal we have 
\begin{equation*}
    \begin{split}
        \|\Q^T e_i\|_2^2  \leq \frac{ \| \mathbf{\Sigma} \Q^T e_i\|_2^2}{\sigma^2_{min}}
        & = \frac{ \| \U^* f(\V^*\X) e_i\|_2^2}{\sigma^2_{min}} \\
        & =  \frac{\| f(\V^*\X) e_i\|_2^2}{\sigma^2_{min}} \\
        & \leq \frac{\| \V^*\X e_i\|_2^2}{\sigma^2_{min}} \\
    \end{split}
\end{equation*}
    
Now each entry in of $\V^*\X$ is an i.i.d. Gaussian, and so is at most $10\sqrt{\log(n)}$ with probability $1 - e^{-10n}$, so 
$ \| \V^*\X e_i\|_2^2 \leq 100k \log(n) $ with probability $1 - e^{-9n}$ by a union bound. Since the columns of $\U^*$ are orthonormal, $\sigma^2_{\min} = \sigma^2_{\min}(f(\V^*\X))$, which is at least $\frac{n}{(\kappa(\V^*))^4}$ by Lemma \ref{prop:reallylongprojectionprop}. Thus we have that $ \|\Q^T e_i\|_2^2  = O(k (\kappa(\V^*))^4 \log(n)/n)$. This shows the $O(  (\kappa(\V^*))^4 \log(n))$-incoherehnce for the second part of Equation \ref{eqn:inco1}, and the first part follows from the $\mu$-incoherence assumption on $U$. The incoherence bound of $(\kappa(\V^*))^{2} \sqrt{k \log(n) \mu} $ for Equation \ref{eqn:incol2} follows by applying Cauchy Schwartz to the LHS and using the bounds just obtained for Equation \ref{eqn:inco1}. 
\end{proof}

\begin{theorem}(Extending Theorem 1.1 in \cite{candes2011robust}.)\label{thm:sparse}
If $\AA = \U^* f(\V^*\X) + \E$ where $\E$ is produced by the sparsity procedure outlined above with $s\leq \gamma nm$ for a fixed constant $\gamma > 0$. Then if $\U^*$ has orthonormal columns, is $\mu$-incoherent, $\X$ is Gaussian, and the sample complexity satisfies $n = \poly(d,m,k,\kappa(\V^*))$, then there is a polynomial time algorithm which, given only $\AA$, outputs both matrices $\M=\U^* f(\V^*\X)$ and $\E$, given that $k \leq \frac{m}{\overline{\mu}\log^2(n)}$, where $\overline{\mu} = O\big( (\kappa(\V^*))^{2} \sqrt{k \log(n) \mu} + \mu +  (\kappa(\V^*))^4 \log(n) \big)$.
\end{theorem}

\begin{proof}
The results of \cite{candes2011robust} demonstrate that solving

\begin{equation*}
\setlength\arraycolsep{.5pt}
  \begin{array}{l@{\quad} r c r c r}
    \min_{\Y}   & \| \AA - \Y \|^2_F &    \\
    \mathrm{s.t.} &  \textrm{rank}(\Y) \leq k  & 
  \end{array}
\end{equation*}
\noindent
recovers the optimal low-rank matrix given that conditions of the previous lemma are satisfied. That is, if we do not care about running time, the above optimization problem recovers $\U^* f(\V^*\X)$ exactly. However, the above problem is highly non-convex and instead we optimize over the nuclear norm.

\begin{equation*}
\setlength\arraycolsep{.5pt}
  \begin{array}{l@{\quad} r c r c r}
    \min_{\Y, \E}   & \|\Y \|_{*} + \|\E\|_1 &    \\
    \mathrm{s.t.} &  \Y + \E = \AA  & 
  \end{array}
\end{equation*}

By Theorem 1.1 in \cite{candes2011robust}, we know that the solution to the above problem is unique and equal to $\U^* f(\V^*\X)$.  
It remains to show that the above optimization problem can be solved in polynomial time. Note, the objective function is convex. As mentioned in \cite{lee2015faster}, we can then run an interior point algorithm and it is well known that in order to achieve additive error $\epsilon$, we need to iterate poly$(\log(1/\epsilon))$ times. Observe, for exact recovery we require a dual certificate that can verify optimality. Section 2.3 in \cite{candes2011robust} uses a modified analysis of the golfing scheme introduced by \cite{gross2011recovering} to create a dual certificate for the aforementioned convex program. We observe that this construction of the dual is independent of the kind of factorization we desire and only requires $\Y$ to be rank $k$. Given that $\U^*, \V^*, \X, \E$ have polynomially bounded bit complexity, this immediately implies a polynomial time algorithm to recover $\U^* f(\V^*\X)$ in unfactored form. 
\end{proof}

As an immediate corollary of the above theorem, our exact algorithms of Section \ref{sec:polyexact} can be applied to the matrix $\M$ of Theorem \ref{thm:sparse} to recover $\U^*, \V^*$. Formally, 

\begin{corollary}\label{cor:sparse}
 Let $\U^* \in \R^{m \times k}, \V^* \in \R^{k \times d}$ be rank $k$ matrices, where $\U^*$ has orthonormal columns, $\max_i \| (\U^*)^T e_i\|_2^2 \leq \frac{\mu k}{m}$ for some $\mu$, and $k \leq \frac{m}{\overline{\mu} \log^2(n)}$, where $\overline{\mu} =  O\big( (\kappa(\V^*))^{2} \sqrt{k \log(n) \mu} + \mu +  (\kappa(\V^*))^4 \log(n) \big)$. Here $\kappa(\V^*)$ is the condition number of $\V^*$. Let $\E$ be generated from the $s$-sparsity procedure with $s =  \gamma nm$ for some constant $\gamma > 0$ and let $\AA = \U^* f(\V\X) + \E$. Suppose the sample complexity satisfies $n = \poly(d,m,k,\kappa(\V^*))$ 
Then on i.i.d. Gaussian input $\X$ there is a $\poly(n)$ time algorithm that recovers $\U^*,\V^*$ exactly up to a permutation and positive scaling with high probability.
\end{corollary}

\section*{Acknowledgements}
The authors would like to thank Anima Anandkumar, Mark Bun, Rong Ge, Sam Hopkins and Rina Panigrahy for useful discussions.

%% file: Appendix.tex
\newpage

\section{Appendix}

\subsection{Generalizing to nonlinear activation functions $f$}\label{sec:generalf}
As remarked in the introduction, for the setting where $\X$ has Gaussian marginals, with several slight changes made to the proofs of our main exact and noisy algorithms, we can generalize our results to any activation function $f$ of the form:
\[ f(x) = \begin{cases} 0 & \text{ if } x \leq 0 \\
\phi(x) & \text{ otherwise} \end{cases}\]
where $\phi(x): [0,\infty] \to [0,\infty]$ is a smooth (on $(0,\infty)$), injective function. We now describe the modifications to our proofs required to generalize to such a function. Note that the primary property of the ReLU used is it's sparsity patterns, i.e. $f(\V\X)$ is zero where-ever $\V\X$ is negative. This is the only property (along with the invertibility of $f(x)$ for $x$ positive) which allows for exact recovery of $\U^*,\V^*$ in Section \ref{sec:polyexact}. Thus the Lemmas which guarantee the uniqueness of the sign patterns of the rows of $f(\V^*\X)$ in their rowspan, namely Lemmas \ref{prop:uniquesign} and \ref{lem:uniquesign}, hold without any modification. Note that continuity of $f$ is required for Lemma \ref{lem:uniquesign}. 

Given this, the linear systems in  Algorithm \ref{alg:overall_exact_fpt}, and similarly in the more general exact Algorithm \ref{alg:recoversigns}, are unchanged. Note that these systems require the injectivity, and thus invertibility, of $f$ for $x > 0$. For instance, in step $8$ of Algorithm \ref{alg:overall_exact_fpt}, instead of solving the linear system $(\W_{i,*})_{S_i} = \V_{i,*} \overline{\X}_{S_i}$ for $\V$, we would need to solve $f^{-1}((\W_{i,*})_{S_i}) = \V_{i,*} \overline{\X}_{S_i}$, where $f^{-1}$ is the inverse of $f$ for $x>0$. Moreover, note that as long as $\phi(x)$ admits a polynomial approximation (with a polynomial of degree at most $\poly(n)$ in the range $[0,\poly(n)]$, the bounds given by the Tensor decomposition algorithms in Section \ref{sec:polyexact} will still hold. 

One detail to note is that if $\phi$ is not multiplicative (i.e. $\phi(xy) = \phi(x)\phi(y)$), then $f$ will no longer commute with positive scalings, so we can no longer pull positive diagonal matrices out of $f(\cdot)$. Observe that $\phi(x) = x^c$ for any $c > 0$ is multiplicative, and this is a non-issue for such $\phi$. On the other hand, note that this does not effect the fact that we recover the true sign pattern needed. Thus in the exact algorithms of Sections \ref{sec:polyexact} and \ref{sec:FPT}, once we multiply by our guessed left inverse of $\U^*$ and we obtain $\D f(\V^*\X) + \ZZ$ where $\ZZ$ is some small, negligible error matrix, As long as the entries of $\D$ are not too extreme (that is, exponentially small, see Lipschitz discussion below, as for any of the example functions that follow this will be the case), we will still recover the true sign pattern of $f(\V^*\X)$ after rounding. Thus we can always recover $\V^*$ from this sign pattern without using any properties of $f$. Note that our noisy algorithm of Section \ref{sec:noisycase} does not need to run any such linear system to find $\V^*$, and so this is a non-issue here. 

Finally, perhaps the important modification which must be made involves a Lipschitz property of $\phi$. Namely, we frequently use the fact that for the ReLU $f$, if $\|\V\X - \V^*\X\|_F^2 < \eps$, then $\|f(\V\X) - f(\V^*\X)\|_F^2 < \eps$.  Here $\V,\V^*$ have unit norm rows, $\X \in \R^{d \times \ell}$ is i.i.d. Gaussian, $\V$ refers to our approximation of $\X$ returned by tensor decomposition, $\eps < 1/\poly(\ell)$, and $\ell > \poly(d,m,k,\kappa)$. Note also that we will have $n = \poly(\ell)$, since to our error $\eps$ depends on the sample complexity $n$. So once we obtain our estimate $\V$ using $n$ samples, we thereafter restrict our attention to a smaller subset of $\ell$ samples.
Now observe that if $\phi(x)$ grows faster than $x$, this bound will no longer hold as is. However, using the fact that the entries of $\V\X$ and $\V^*\X$ are Gaussian, and thus have sub-Gaussian tails, we will be able to bound the blow up. First note that with high probability, we have $\|\V\X\|_\infty + \|\V^*\X\|_\infty < O(\sqrt{\log(\ell)})$. Thus, for a fixed $\phi$, define the $B$-bounded Lipschitz constant $L_B(\phi)$ by
\[ L_B(\phi) = \sup_{x \neq y, |x| < B , |y| <B }\frac{|f(x) - f(y)|}{|x-y|} \]

Note if $\phi(x) = x^c$ for some constant $c \in \mathbb{N}$, we
have $L_B(\phi) < cB^{c-1}$. Given this, it follows that if $\|\V\X - \V^*\X\|_F^2 < \eps$, then $\|f(\V\X) - f(\V^*\X)\|_F^2 \leq (L_B(\phi)k\ell)^2 \eps$ (here we can use bounds between $L_1^2$ and $L_2^2$ of $k \ell$ on these matrices to explicitly relate this difference in Frobenius norm to the Lipschitz constant). 
Now observe that in all cases where we have $\|\V\X - \V^*\X\|_F^2 < \eps$, and need a bound on  $\|f(\V\X) - f(\V^*\X)\|_F^2$, we can handle a $\poly(\ell)$ blow-up, as $\eps<1/\poly(\ell)$ can be made arbitrarily large by increasing the sample complexity on which the algorithm which originally recovered $\V$ was run on. Thus, we claim that we can handle any function $\phi(x)$ such that $L_{\Theta(\sqrt{\log(\ell)})}(\phi) < \poly(\ell)$. Note that this includes all polynomials $\phi(x)=x^c$ of constant degree, as well as even \textit{exponential} functions $\phi(x) = e^x-1$ or $\phi(x) = e^{x^2}-1$.  

However, importantly, in addition to the $L_{\Theta(\sqrt{\log(\ell)})}(\phi)$ blow-up in runtime, for the specific case of our noisy algorithm of Theorem \ref{thm:noisyfinal}, our our runtime also blows up by a $\phi(\kappa)^2$ factor. This is because the projection bounds in Corollary \ref{cor:projectnegV} will become \[ \| f(\V_{i,*}^* \X) \P_{S_{i,-}} \|_2 = \|f(\V_{i,*}^* \X)\|_2 \Big( 1 - \Omega(\frac{1}{\phi(\kappa(\V^*))^2 \poly(k)})\Big)\]
instead of 
\[ \| f(\V_{i,*}^* \X) \P_{S_{i,-}} \|_2 = \|f(\V_{i,*}^* \X)\|_2 \Big( 1 - \Omega(\frac{1}{\kappa(\V^*)^2 \poly(k)})\Big)\]
Thus we must make $\ell,1/\eps >> \phi(\kappa(\V^*))^2 $ in order to recover the correct signs in our projection based algorithm (Algorithm  \ref{alg:recoversigns2}).

To summarize , the primary change that occurs is blow-up the runtime by the bounded Lipschitz constant $L_{\Theta(\sqrt{\log(n)}}(\phi)$ of $\phi$ in the runtime of our exact recovery algorithms for the noiseless case, which is polynomial as long as $\phi(x) = O(e^{x^2})$. This also holds for the case of our fixed parameter tractable noiseless algorithm of Section \ref{sec:FPT}, and the fixed parameter tractable noisy algorithm of Section \ref{sec:1bit}. For the noisy case of Theorem \ref{thm:noisyfinal}, which is our polynomial time algorithm for sub-Gaussian noise, we also get a blowup of $\phi(\kappa(\V^*))^2$ in the runtime, which is still polynomial as long as $\phi(x)$ is bounded by some constant degree polynomial.